\RequirePackage{fix-cm} 
\documentclass[a4paper, twoside, reqno, dvips, 12pt]{amsart}
\usepackage{fixltx2e}   

\usepackage{etex}

\usepackage[latin1]{inputenc}
\usepackage[T1]{fontenc}

\usepackage{esint}
\usepackage{dsfont}
\usepackage{xspace}
\usepackage{amsgen}
\usepackage{amsthm}
\usepackage{amssymb}
\usepackage{amsmath}
\usepackage{wasysym}
\usepackage{upgreek}
\usepackage{amsfonts}
\usepackage{scalerel}
\usepackage{stmaryrd}
\usepackage{mathtools}
\usepackage{stackengine}

\usepackage{relsize}
\usepackage[mathcal, mathscr]{euscript}

\usepackage{mathrsfs}
\DeclareMathAlphabet{\mathscrbf}{OMS}{mdugm}{b}{n}

\usepackage[cbgreek]{textgreek}

\usepackage{a4wide}

\usepackage[dvipsnames, table]{xcolor}
\definecolor{bckg}{RGB}{20.8, 20.8, 20.8}
\definecolor{oneblue}{rgb}{0.0, 0.0, 0.85}
\definecolor{Lightblue}{RGB}{214, 214, 214}
\definecolor{bluepigment}{rgb}{0.2, 0.2, 0.6}
\definecolor{charcoal}{rgb}{0.21, 0.27, 0.31}
\definecolor{denimblue}{rgb}{0.08, 0.38, 0.74}
\definecolor{Lightgray}{rgb}{0.89, 0.89, 0.89}
\definecolor{darkgrey}{rgb}{0.273, 0.281, 0.30}
\definecolor{darkelectricblue}{rgb}{0.33, 0.41, 0.47}

\usepackage[sort&compress, comma, square, numbers]{natbib}

\usepackage{psfrag}
\usepackage{graphicx}
\usepackage{subfigure}
\usepackage{morefloats}
\usepackage{indentfirst}

\usepackage{acronym}
\usepackage{microtype}
\usepackage[labelsep=period,%
            labelfont={bf,sf,color=bluepigment},%
            justification=raggedright]{caption}

\usepackage[perpage, symbol]{footmisc}

\usepackage[usenames, dvipsnames, pdf]{pstricks}
\usepackage{epsfig}
\usepackage{pst-grad} 
\usepackage{pst-plot} 

\usepackage[colorlinks,
           urlcolor=oneblue,
           linkcolor=denimblue,
           citecolor=NavyBlue,
           bookmarksopen=false,
           pdfpagemode=UseNone,
           pagebackref]{hyperref}

\usepackage[explicit]{titlesec}

\titleformat{\section}[block]
  {\color{NavyBlue}\Large\sffamily\bfseries}
  {}
  {0.0em}
  {\colorbox{bckg!5}{\strut\parbox{\dimexpr\linewidth-2\fboxsep\relax}{\thesection. #1}}}
  [\vspace*{0.33em}]

\titleformat{name=\section,numberless}[block]
  {\color{NavyBlue}\Large\sffamily\bfseries}
  {}
  {0.0em}
  {\colorbox{bckg!5}{\strut\parbox{\dimexpr\linewidth-2\fboxsep\relax}{#1}}}
  [\vspace*{0.33em}]

\titleformat{\subsection}
  {\color{NavyBlue}\large\sffamily\bfseries}
  {}
  {0.0em}
  {\colorbox{bckg!5}{\parbox{\dimexpr\linewidth-2\fboxsep\relax}{\thesubsection. #1}}}
  [\vspace*{0.33em}]

\titleformat{name=\subsection,numberless}
  {\color{NavyBlue}\Large\sffamily\bfseries}
  {}
  {0em}
  {\colorbox{bckg!5}{\parbox{\dimexpr\linewidth-2\fboxsep\relax}{#1}}}
  [\vspace*{0.33em}]

\titleformat{\subsubsection}
  {\color{bluepigment}\sffamily\normalsize\bfseries}
  {\thesubsubsection}
  {0.5em}
  {#1}
  [\vspace*{0.33em}]

\titleformat{\paragraph}[runin]
  {\color{bluepigment}\sffamily\small\bfseries}
  {}
  {0em}
  {#1}

\titlespacing{\section}{0.0em}{1.5em plus 2pt minus 2pt}%
{1.0em plus 2pt minus 2pt}[0em]
\titlespacing{\subsection}{0.5em}{1.5em plus 2pt minus 2pt}%
{1.0em}[0em]
\titlespacing{\subsubsection}{0.5em}{1.5em plus 2pt minus 2pt}%
{1.0em plus 2pt minus 2pt}[0em]

\usepackage{titletoc}

\contentsmargin{0.5em}
\newlength{\tocsep} 
\titlecontents{section}[\tocsep]
  {\addvspace{10pt}\bfseries\sffamily}
  {\contentslabel[\thecontentslabel]{\tocsep}}
  {}
  {\ \titlerule*[0.75pc]{.}\ \thecontentspage}
  []
\titlecontents{subsection}[\tocsep]
  {\addvspace{8pt}\sffamily}
  {\contentslabel[\thecontentslabel]{\tocsep}}
  {}
  {\ \titlerule*[0.5pc]{.}\ \thecontentspage}
  []
\titlecontents*{subsubsection}[\tocsep]
  {\addvspace{2pt}\footnotesize\sffamily}
  {}
  {}
  {\ \titlerule*[0.35pc]{.}\ \thecontentspage}
  [\\*]

\makeatletter
\def\@setauthors{%
  \begingroup
  \def\thanks{\protect\thanks@warning}%
  \trivlist
  \centering\footnotesize \@topsep30\p@\relax
  \advance\@topsep by -\baselineskip
  \item\relax
  \author@andify\authors
  \def\\{\protect\linebreak}%
  \textsc{\normalsize\textcolor{darkelectricblue}{\authors}}%
  \ifx\@empty\contribs
  \else
    ,\penalty-3 \space \@setcontribs
    \@closetoccontribs
  \fi
  \endtrivlist
  \endgroup
}
\def\@settitle{\begin{center}%
  \baselineskip14\p@\relax
    \bfseries
    \textsc{\Large\textcolor{charcoal}{\@title}}
  \end{center}%
}
\makeatother

\usepackage{enumitem}
\setlist[description]{%
  topsep=30pt,               
  itemsep=5pt,               
  font={\bfseries\sffamily\color{NavyBlue}}, 
}

\usepackage{fancyhdr}
\usepackage{lastpage}

\usepackage{tikz}
\usepackage{tikz-cd}

\newcommand*\Title{\textcolor{bluepigment}{Wave/moving wall interaction}}
\newcommand*\Authors{\textcolor{bluepigment}{G.~Khakimzyanov \& D.~Dutykh}}
\newcommand*{\plogo}{\textcolor{gray}{{\texttt{arXiv.org} / \textsc{hal}}}} 

\numberwithin{equation}{section}

\newtheorem{lemma}{Lemma}
\newtheorem{remark}{Remark}

\newtheorem{theorem}{Theorem}

\newcommand{\up}[1]{$^{\mathrm{\small\textsf{#1}}}$} 

\makeatletter
\newcommand*{\mint}[1]{%
  \mint@l{#1}{}%
}
\newcommand*{\mint@l}[2]{%
  \@ifnextchar\limits{%
    \mint@l{#1}%
  }{%
    \@ifnextchar\nolimits{%
      \mint@l{#1}%
    }{%
      \@ifnextchar\displaylimits{%
        \mint@l{#1}%
      }{%
        \mint@s{#2}{#1}%
      }%
    }%
  }%
}
\newcommand*{\mint@s}[2]{%
  \@ifnextchar_{%
    \mint@sub{#1}{#2}%
  }{%
    \@ifnextchar^{%
      \mint@sup{#1}{#2}%
    }{%
      \mint@{#1}{#2}{}{}%
    }%
  }%
}
\def\mint@sub#1#2_#3{%
  \@ifnextchar^{%
    \mint@sub@sup{#1}{#2}{#3}%
  }{%
    \mint@{#1}{#2}{#3}{}%
  }%
}
\def\mint@sup#1#2^#3{%
  \@ifnextchar_{%
    \mint@sup@sub{#1}{#2}{#3}%
  }{%
    \mint@{#1}{#2}{}{#3}%
  }%
}
\def\mint@sub@sup#1#2#3^#4{%
  \mint@{#1}{#2}{#3}{#4}%
}
\def\mint@sup@sub#1#2#3_#4{%
  \mint@{#1}{#2}{#4}{#3}%
}
\newcommand*{\mint@}[4]{%
  \mathop{}%
  \mkern-\thinmuskip
  \mathchoice{%
    \mint@@{#1}{#2}{#3}{#4}%
        \displaystyle\textstyle\scriptstyle
  }{%
    \mint@@{#1}{#2}{#3}{#4}%
        \textstyle\scriptstyle\scriptstyle
  }{%
    \mint@@{#1}{#2}{#3}{#4}%
        \scriptstyle\scriptscriptstyle\scriptscriptstyle
  }{%
    \mint@@{#1}{#2}{#3}{#4}%
        \scriptscriptstyle\scriptscriptstyle\scriptscriptstyle
  }%
  \mkern-\thinmuskip
  \int#1%
  \ifx\\#3\\\else_{#3}\fi
  \ifx\\#4\\\else^{#4}\fi  
}
\newcommand*{\mint@@}[7]{%
  \begingroup
    \sbox0{$#5\int\m@th$}%
    \sbox2{$#5\int_{}\m@th$}%
    \dimen2=\wd0 %
    \let\mint@limits=#1\relax
    \ifx\mint@limits\relax
      \sbox4{$#5\int_{\kern1sp}^{\kern1sp}\m@th$}%
      \ifdim\wd4>\wd2 %
        \let\mint@limits=\nolimits
      \else
        \let\mint@limits=\limits
      \fi
    \fi
    \ifx\mint@limits\displaylimits
      \ifx#5\displaystyle
        \let\mint@limits=\limits
      \fi
    \fi
    \ifx\mint@limits\limits
      \sbox0{$#7#3\m@th$}%
      \sbox2{$#7#4\m@th$}%
      \ifdim\wd0>\dimen2 %
        \dimen2=\wd0 %
      \fi
      \ifdim\wd2>\dimen2 %
        \dimen2=\wd2 %
      \fi
    \fi
    \rlap{%
      $#5%
        \vcenter{%
          \hbox to\dimen2{%
            \hss
            $#6{#2}\m@th$%
            \hss
          }%
        }%
      $%
    }%
  \endgroup
}

\newcommand{\rhou}{\uprho}

\newcommand{\g}{\mathbf{g}}

\newcommand{\R}{\mathds{R}}

\renewcommand{\phi}{\varphi}

\newcommand{\A}{\mathscr{A}}
\newcommand{\B}{\mathscr{B}}

\newcommand{\G}{\mathcal{G}}
\newcommand{\K}{\mathcal{K}}

\newcommand{\ud}{\mathrm{d}}
\newcommand{\ui}{\mathrm{i}}
\newcommand{\ue}{\mathrm{e}}
\newcommand{\C}{\mathcal{C}}
\newcommand{\F}{\mathscr{F}}
\newcommand{\J}{\mathcal{J}}

\newcommand{\Q}{\mathcal{Q}}
\newcommand{\Bf}{\mathscr{B}}
\newcommand{\Qf}{\mathscr{Q}}
\newcommand{\Ru}{\mathcal{R}}

\newcommand{\Dd}{\mathscr{D}}

\newcommand{\alphau}{\upalpha}
\newcommand{\Aa}{\mathfrak{A}}
\newcommand{\Ba}{\mathfrak{B}}

\renewcommand{\leq}{\leqslant}
\renewcommand{\geq}{\geqslant}

\renewcommand{\O}{\mathcal{O}}

\renewcommand{\H}{\mathcal{H}}

\newcommand{\n}{\boldsymbol{n}}

\newcommand{\q}{\boldsymbol{q}}
\newcommand{\x}{\boldsymbol{x}}
\newcommand{\vO}{\boldsymbol{0}}

\newcommand{\phit}{\tilde{\phi}}
\newcommand{\nuv}{\mbox{\textnu}}
\renewcommand{\j}{\boldsymbol{j}}
\renewcommand{\u}{\boldsymbol{u}}
\renewcommand{\v}{\boldsymbol{v}}
\newcommand{\phin}{\mbox{\textphi}}
\newcommand{\rhot}{\mbox{\textrho}}
\newcommand{\const}{\mathrm{const}}
\newcommand{\No}{$\mathrm{N}^\circ$}
\renewcommand{\psi}{\mbox{\textpsi}}
\newcommand{\betat}{\mbox{\textbeta}}
\newcommand{\zetat}{\mbox{\textzeta}}
\newcommand{\alphat}{\mbox{\textalpha}}
\newcommand{\gammat}{\mbox{\textgamma}}
\newcommand{\Gammao}{\mathring{\Gamma}}
\renewcommand{\mu}{\mbox{\textmugreek}}
\renewcommand{\theta}{\mbox{\texttheta}}
\renewcommand{\sigma}{\mbox{\textsigma}}

\newcommand{\vsigma}{\mbox{\textvarsigma}}
\newcommand{\etab}{\mbox{\textbf{\texteta}}}

\newcommand{\m}{\mathsf{m}}

\newcommand{\cm}{\mathsf{cm}}
\newcommand{\Hz}{\mathsf{Hz}}

\newcommand{\cf}{\emph{cf.}\xspace}
\newcommand{\ie}{\emph{i.e.}\xspace}
\newcommand{\eg}{\emph{e.g.}\xspace}

\newcommand{\etal}{\emph{et al.}\xspace}

\renewcommand{\sim}{\thicksim}
\renewcommand{\div}{\grad\scal}
\newcommand{\sech}{\mathrm{sech}}

\newcommand{\scal}{\boldsymbol{\cdot}}
\newcommand{\grad}{\boldsymbol{\nabla}}
\newcommand{\rot}{\mathop{\mathrm{rot}}}

\newcommand{\abs}[1]{\lvert\, #1\, \rvert}

\newcommand{\norm}[1]{\lVert\, #1\, \rVert}

\newcommand{\normh}[1]{\lVert\, #1\, \rVert_{\,\H}}
\newcommand{\pd}[2]{\frac{\partial\/ #1}{\partial\/ #2}}
\newcommand{\scalh}[2]{\langle\, #1,\,#2\,\rangle_{\,\H}}
\newcommand{\scalb}[2]{\langle\, #1,\,#2\,\rangle_{\,\Ba}}
\newcommand{\pdd}[2]{\dfrac{\partial\/ #1}{\partial\/ #2}}
\newcommand{\od}[2]{\frac{\mathrm{d}\/ #1}{\mathrm{d}\/#2}}

\newcommand{\odd}[2]{\dfrac{\mathrm{d}\/ #1}{\mathrm{d}\/#2}}
\newcommand{\eqdef}{\mathop{\stackrel{\,\mathrm{def}}{:=}\,}}
\newcommand{\defeq}{\mathop{\stackrel{\,\mathrm{def}}{=:}\,}}
\newcommand{\scalhh}[2]{\langle\, #1,\,#2\,\rangle_{\,\H_{\,1},\,j_{\,2}}}

\newcommand{\half}{{\textstyle{1\over2}}}

\usepackage{acronym}
\acrodef{ode}[ODE]{Ordinary Differential Equation}
\acrodef{sgn}[SGN]{Serre--Green--Naghdi}
\acrodef{bvp}[BVP]{Boundary Value Problem}
\acrodef{NSWE}{Nonlinear Shallow Water Equations}

\begin{document}

\title[\Title]{Numerical modelling of surface water wave interaction with a moving wall}

\author[G.~Khakimzyanov]{Gayaz Khakimzyanov}
\address{\textbf{G.~Khakimzyanov:} Institute of Computational Technologies, Siberian Branch of the Russian Academy of Sciences, Novosibirsk 630090, Russia}
\email{Khak@ict.nsc.ru}
\urladdr{http://www.ict.nsc.ru/ru/structure/Persons/ict-KhakimzyanovGS}

\author[D.~Dutykh]{Denys Dutykh$^*$}
\address{\textbf{D.~Dutykh:} LAMA, UMR 5127 CNRS, Universit\'e Savoie Mont Blanc, Campus Scientifique, F-73376 Le Bourget-du-Lac Cedex, France}
\email{Denys.Dutykh@univ-savoie.fr}
\urladdr{http://www.denys-dutykh.com/}
\thanks{$^*$ Corresponding author}


\begin{titlepage}
\thispagestyle{empty} 
\noindent
{\Large Gayaz \textsc{Khakimzyanov}}\\
{\it\textcolor{gray}{Institute of Computational Technologies, Novosibirsk, Russia}}
\\[0.02\textheight]
{\Large Denys \textsc{Dutykh}}\\
{\it\textcolor{gray}{CNRS, Universit\'e Savoie Mont Blanc, France}}
\\[0.16\textheight]

\vspace*{0.99cm}

\colorbox{Lightblue}{
  \parbox[t]{1.0\textwidth}{
    \centering\huge\sc
    \vspace*{0.79cm}
    
    \textcolor{bluepigment}{Numerical modelling of surface water wave interaction with a moving wall}
    
    \vspace*{0.79cm}
  }
}

\vfill 

\raggedleft     
{\large \plogo} 
\end{titlepage}


\newpage
\thispagestyle{empty} 
\par\vspace*{\fill}   
\begin{flushright} 
{\textcolor{denimblue}{\textsc{Last modified:}} \today}
\end{flushright}


\newpage
\maketitle
\thispagestyle{empty}


\begin{abstract}

In the present manuscript we consider the practical problem of wave interaction with a vertical wall. However, the novelty here consists in the fact that the wall can move horizontally due to a system of springs. The water wave evolution is described with the free surface potential flow model. Then, a semi-analytical numerical method is presented. It is based on a mapping technique and a finite difference scheme in the transformed domain. The idea is to pose the equations on a fixed domain. This method is thoroughly tested and validated in our study. By choosing specific values of spring parameters, this system can be used to damp (or in other words \emph{to extract the energy} of) incident water waves. \\

\bigskip
\noindent \textbf{\keywordsname:} wave/body interaction; full Euler equations; movable wall; wave run-up; wave damping; wave energy extraction \\


\end{abstract}


\newpage
\thispagestyle{empty}
\tableofcontents
\thispagestyle{empty}


\newpage
\section{Introduction}

The mathematical and numerical high fidelity modeling of water waves is a central topic in coastal and naval engineering. The incident waves come and interact with various coastal features. Nowadays the interaction of water waves with bathymetric features is relatively well understood (at least with some special features \cite{Kurkin2015}). A more challenging situation is the interaction of waves with various coastal structures \cite{Linton2001}. At this level the most studied situation is the wave/wall interaction and the wall is assumed to be fixed and impermeable. Violent wave impacts have to be modelled in general using more CFD-like methods in the \textsc{Eulerian} mesh-based \cite{Dias2008, Chen2014} or \textsc{Lagrangian} particle-based \cite{Didier2014} approaches. In the present study we apply the free surface approximation by neglecting all the processes happening in the air above the free surface \cite{Cooker1997}. In this line of thinking, the interaction of periodic waves with a fixed vertical wall was recently studied in \cite{Carbone2013} in the framework of a fully nonlinear weakly dispersive wave model \cite{Dutykh2011a}. The conditions under which an extreme wave run-up on a vertical wall may happen were describe in \cite{Carbone2013} as well.

In this study, we focus on wave interactions with a movable vertical wall. The wall motion can be prescribed. In this case we model the wave generation process with a piston-type wave maker \cite{Guizien2002}. For instance, this problem was considered in the framework of \textsc{Boussinesq}-type equations in \cite{Orszaghova2012}. Otherwise, the wall can move under the action of incident waves. We can even assume that a system of horizontal tension/extension springs (with tunable rigidities) is attached to the wall. Thus, the wall may present a certain resistance to the action of waves. This problem can be regarded also as a piston's free motion under the forcing of water waves. The piston mechanical energy conversion and recuperation is a different technological problem which is out of scope of the present study. The extraction of ocean wave energy on industrial scales is not yet a very common practice \cite{Falnes2007}. However, very active researches in this field are conducted since at least forty years \cite{Evans1981, Falnes2002}. Consequently, the numerical methods developed below can be used to design and to optimize such Wave Energy Conversion (WEC) devices. Movable walls have been used, for example, in a triplate system proposed back in 1977 by Dr.~Francis~\textsc{Farley}. The rigidity of strings is chosen to minimize the reflected wave amplitude so that most of the wave energy be converted into the mechanical energy of the device. Mathematical and numerical modeling of this type of WEC devices is considered below. The main point is that (ordinary and partial) differential equations posed in time-varying domains are known to pose notorious theoretical and numerical difficulties \cite{Knobloch2015}.

In the present work we consider a two-dimensional non-stationary problem of surface wave motion in a domain with one moving wall. We assume that the wall remains vertical during its motion. Moreover, at least in a vicinity of the moving wall the bathymetry has to be flat to allow its free motion under the action of waves or to follow a prescribed trajectory. The fluid flow is assumed to be potential and we address the complete (\ie fully nonlinear and fully dispersive) water wave problem \cite{Stoker1957}. We propose first a reformulation of this problem on a fixed and domain (a square) by introducing a new curvilinear coordinate system. Then, we propose a robust finite difference discretization of this problem with mathematically proven good qualities. The performance of this algorithm is illustrated on several examples of practical interest. In particular, we study the influence of springs rigidity on the oscillatory motion of the wave/wall system.

The results presented in this study can be viewed under a different angle and, thus, they can be applied to another problem of coastal structures protection from wave loads and impacts. In particular, the catastrophic consequences of 2011 \textsc{Tohoku} earthquake and tsunami are widely known nowadays \cite{Mori2011}. The protecting structures are made very solid by implementing various security coefficients. However, more economic protections can be designed if we allow them to move under the action of waves. The moving parts can be related to fixed ones by a system of flexible strings and by tuning their rigidity one can reduce significantly the wave run-up on moving parts as well as wave loads on fixed solid structural elements.

We would like to mention also related works. The generation of water waves by an accelerating moving vertical plate in a channel of constant depth was studied in \cite{Yang1992}. The Authors measured the free surface elevation and pressure distributions on the wall for three different values of plate's accelerations. In another recent work \cite{Wang2011} a nonlinear model for incompressible free surface flows was proposed. Then, this model was used to study wave run-up on a sloping beach and on a moving vertical wall. The two-dimensional unsteady problem of wave interaction with a moving vertical wall was studied analytically in \cite{Korobkin2009}. Moreover, the comparisons with numerical solutions obtained with the complex boundary element method was presented as well. Another numerical model was presented in \cite{He2009} as well. In last two works the waves were generated by an initial disturbance located close to the moving wall.

The present manuscript is organized as follows. The physical and mathematical problems are formulated in Section~\ref{sec:problem} and then it is reformulated on a fixed domain in Section~\ref{sec:trans}. The proposed finite difference approximation is described in Section~\ref{sec:fd} and studied mathematically in Section~\ref{sec:study}. The construction of the mapping in the fixed domain is discussed in Section~\ref{sec:ell} and the numerical algorithm in general is described in Section~\ref{sec:alg}. Some numerical results are presented in Section~\ref{sec:res}. Finally, the main conclusions and perspectives are outlined in Section~\ref{sec:concl}.


\section{Problem formulation}
\label{sec:problem}

Consider a two-dimensional motion of an incompressible, homogeneous and ideal fluid with free surface. The sketch of the fluid domain is depicted in Figure~\ref{fig:sketch}. The \textsc{Cartesian} coordinate system $O\,x\,y\,$, $\x\ =\ (x,\,y)$ is chosen such that the horizontal axis $O\,x$ is directed along the undisturbed free surface and the axis $O\,y$ points vertically upwards. The fluid domain $\Omega\,(t)$ is assumed to be simply connected and bounded from below by a horizontal bottom $\Gamma_b\ \eqdef\ \{\x\ \in\ \Omega\,(t)\ |\ y\ =\ -h_0\}$ and from above by the free surface --- $\Gamma_s\ \eqdef\ \{\x\ \in\ \Omega\,(t)\ |\ y\ =\ \eta(x,\,t)\}\,$. On the sides the domain $\Omega\,(t)$ is bounded by two vertical walls. The right wall $\Gamma_r\ \eqdef\ \{\x\ \in\ \Omega\,(t)\ |\ x\ =\ \ell\}$ is assumed to be fixed\footnote{Thus, the length of the undisturbed domain is equal to $\ell\ >\ 0\,$.}, while the left one $\Gamma_\ell\ \eqdef\ \{\x\ \in\ \Omega\,(t)\ |\ x\ =\ s\,(t)\}$ is connected to a system of springs and can move horizontally under the wave action. By $s\,(t)$ we denote the displacement of the moving wall\footnote{We assume additionally that the wall is non-deformable and remains vertical during the interaction with waves.} with respect to its undisturbed position $x\ =\ 0\,$. The closed domain $\Omega\,(t)$ together with its boundaries is denoted by $\bar{\Omega}\,(t)\,$.

\begin{remark}
In order to be able to simulate water waves with a shape, which is not necessarily a graph, we can adopt a parametric representation of the free surface $\Gamma_s\,$:
\begin{equation}\label{eq:param}
  x\ =\ \xi\,(q,\,t)\,, \qquad y\ =\ \eta\,(q,\,t)\,,
\end{equation}
where $q\ \in\ \R$ is a real parameter and $t$ is the time evolution variable. Functions $\xi\,(\cdot,\,\cdot)\,:\ \R\times\R^{\,+}\ \mapsto\ \R$ and $\eta\,(\cdot,\,\cdot)\,:\ \R\times\R^{\,+}\ \mapsto\ \R$ are assumed to be sufficiently smooth. Without any restriction we can assume that $q\ \in\ [\,0,\,1\,]\,$. Below in the text (see Section~\ref{sec:trans}) this parametrization will appear more naturally when we map the unknown fluid domain $\Omega\,(t)$ to a fixed reference domain.
\end{remark}

\begin{figure}
  \centering
  \includegraphics[width=0.65\textwidth]{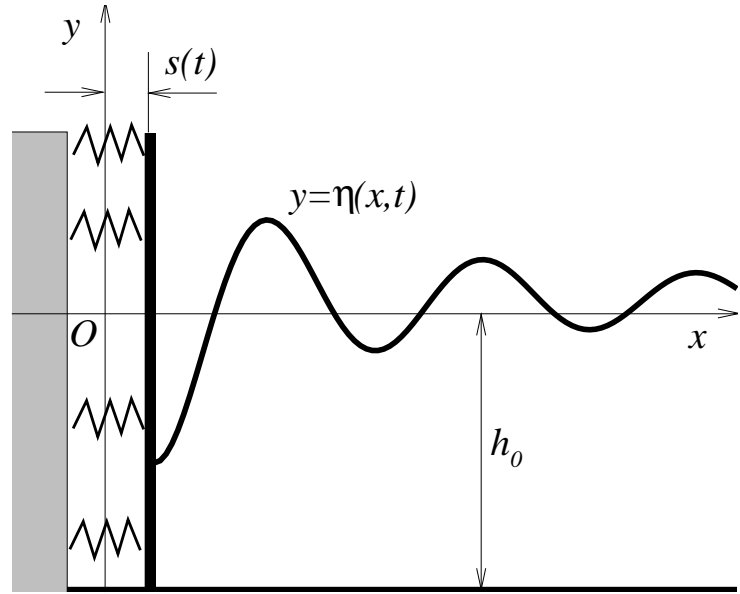}
  \vspace*{-1em}
  \caption{\small\em Sketch of the physical problem of wave interaction with a moving wall.}
  \label{fig:sketch}
\end{figure}

The fluid flow is described by Euler equations \cite{Lavrentev1977}:
\begin{align}
  \div\u\ &=\ 0\,, \label{eq:incompress} \\
  \u_{\,t}\ +\ (\u\scal\grad)\,\u\ +\ \frac{1}{\rho}\;\grad p\ &=\ \g\,, \label{eq:momentum}
\end{align}
where $\u\ =\ (u,\,v)$ is the velocity vector, $\rho\ =\ \const\ >\ 0$ is the fluid density, $p$ is the total pressure and $\g\ =\ (0,\, -\,g)$ is the gravity acceleration. The subscripts such as $\u_{\,t}$ denote partial derivatives, \ie~$\u_{\,t}\ \eqdef\ \pd{\u}{t}\,$. These equations have to be supplemented by appropriate boundary conditions in order to have a well-posed problem. The impermeability condition on fixed solid surfaces (the right vertical wall $\Gamma_r$ and the bottom $\Gamma_b$) reads
\begin{equation*}
  \u\scal\n\ =\ 0\,, \qquad \x\ \in\ \Gamma_b\ \cup\ \Gamma_r\,, \qquad t\ \geq\ 0\,,
\end{equation*}
where $\n$ is the vector of unitary exterior normal to the corresponding boundary of the domain $\Omega\,(t)\,$. The left wall is impermeable as well and the boundary condition reads
\begin{equation*}
  \u\scal\n\ =\ \dot{s}\,(t)\,, \qquad \x\ \in\ \Gamma_\ell\,, \qquad t\ \geq\ 0\,.
\end{equation*}
Taking into account the horizontal character of the wall motion, this condition can be further simplified to give
\begin{equation*}
  u\ =\ \dot{s}\,(t)\,, \qquad \x\ \in\ \Gamma_\ell\,, \qquad t\ \geq\ 0\,.
\end{equation*}

Traditionally on the free surface we prescribe the kinematic
\begin{equation}\label{eq:kinem}
  \eta_{\,t}\ +\ u\cdot\eta_{\,x}\ =\ v\,, \qquad \x\ \in\ \Gamma_s\,, \qquad t\ \geq\ 0
\end{equation}
and dynamic 
\begin{equation}\label{eq:dynam}
  p\ =\ p_a\ =\ \const\,, \qquad \x\ \in\ \Gamma_s\,, \qquad t\ \geq\ 0
\end{equation}
boundary conditions \cite{Stoker1957}. Here $p_a$ is the constant atmospheric pressure. The boundary conditions mentioned above are enough to obtain a closed system of equations which describe the motion of an \emph{ideal} fluid. One has to prescribe also the initial conditions, but their form is highly problem-dependent. Consequently, it will not be discussed at the current level of generality.

\begin{remark}
Let us obtain also kinematic free surface boundary conditions when the free surface $\Gamma_s$ is given in a parametric form \eqref{eq:param}. The free surface $\Gamma_s$ is characterized by the fact that its velocity is determined by the velocity of fluid particles constituting this boundary. In other words, the point $\x\ \in\ \Gamma_s$ moves with the velocity of the fluid particle located in this point, \ie
\begin{equation}\label{eq:cond}
  \od{\x}{t}\ =\ \u\,\bigl(\x(t),\,t\bigr)\,.
\end{equation}
By taking the material derivative of both sides of the parametric form \eqref{eq:param}, we have
\begin{align*}
  \od{x}{t}\ &\equiv\ \pd{\xi}{t}\ +\ \pd{\xi}{q}\cdot\od{q}{t}\,, \\
  \od{y}{t}\ &\equiv\ \pd{\eta}{t}\ +\ \pd{\eta}{q}\cdot\od{q}{t}\,.
\end{align*}
Thus, by taking into account the physical information \eqref{eq:cond} we obtain the following kinematic boundary conditions for $\forall\,\x\ \in\ \Gamma_s$ and for $\forall\, t\ \geq\ 0\,$:
\begin{align}\label{eq:kin1}
  \pd{\xi}{t}\ +\ \pd{\xi}{q}\cdot\od{q}{t}\ -\ u\ &=\ 0\,, \\
  \pd{\eta}{t}\ +\ \pd{\eta}{q}\cdot\od{q}{t}\ -\ v\ &=\ 0\,.\label{eq:kin2}
\end{align}
In the conditions above the quantity $\od{q}{t}$ is the rate of change of the parameter $q$ for fluid particles located at the free surface. If one uses the \textsc{Lagrangian} description of a flow, fluid particles keep their labels and, thus, $\od{q}{t}\ \equiv\ 0\,$. However, we use an arbitrary parametrization and in our case $\od{q}{t}\ \neq\ 0$ in general.
\end{remark}


\subsection{Potential flow model}

The formulation presented above is still too general. We shall simplify it further by assuming the flow to be irrotational, \ie
\begin{equation*}
  \rot\u\ =\ -\pd{u}{y}\ +\ \pd{v}{x}\ =\ 0\,.
\end{equation*}
It implies the existence of a function $\phi\,:\ \R^2\times\R^+\ \mapsto\ \R$ which is called the \emph{velocity potential} such that the velocity field is given by
\begin{equation}\label{eq:pot}
  \u\ =\ \grad\,\phi\,.
\end{equation}
Substituting this form into the incompressibility condition \eqref{eq:incompress} yields that the velocity potential has to be a harmonic function
\begin{equation*}
  \Delta\,\phi\ =\ 0\,, \qquad (\x,\,t)\ \in\ \Omega\,(t)\times\R^+\,.
\end{equation*}
The kinematic boundary condition is obtained straightforwardly by substituting into \eqref{eq:kinem} the velocity components by their representations in terms of the velocity potential:
\begin{equation*}
  \eta_{\,t}\ +\ \phi_{\,x}\cdot\eta_{\,x}\ =\ \phi_{\,y}\,, \qquad y\ =\ \eta\,(x,\,t)\,.
\end{equation*}
We note also that velocity components $u$ and $v$ in the kinematic boundary conditions \eqref{eq:kin1}, \eqref{eq:kin2} are expressed in terms of the velocity potential $\phi$ according to \eqref{eq:pot}.

The momentum equation \eqref{eq:momentum} can be integrated and combined with the dynamic boundary condition \eqref{eq:dynam} to give the so-called \textsc{Euler}--\textsc{Lagrange} integral\footnote{The \textsc{Euler}--\textsc{Lagrange} integral becomes the \textsc{Bernoulli} integral for steady flows. We implicitly chose the gauge where the \textsc{Bernoulli} `constant' is identically zero.} evaluated at the free surface:
\begin{equation*}
  \phi_{\,t}\ +\ \half\,\abs{\grad\phi}^{\,2}\ +\ g\,\eta\ =\ 0\,, \qquad y\ =\ \eta\,(x,\,t)\,.
\end{equation*}
On solid stationary walls the velocity potential satisfies the homogeneous Neumann condition
\begin{equation*}
  \grad\phi\scal\n\ =\ 0\,, \qquad \x\ \in\ \Gamma_b \cup \Gamma_r\,,
\end{equation*}
(with $\n$ defined above) while on the moving wall $\Gamma_\ell$ it satisfies the following non-homogeneous Neumann condition
\begin{equation*}
  \grad\phi\scal\n\ \equiv\ \pd{\phi}{x}\ =\ \dot{s}\,(t)\,, \qquad \x\ \in\ \Gamma_\ell\,.
\end{equation*}

So, instead of seeking for two components of the velocity field $\u(\x,\,t)\ =\ \bigl(u(\x,t),\, v(\x,t)\bigr)$ and the pressure function $p(\x,\,t)\,$, the formulation was simplified to one unknown function, \ie~the velocity potential $\phi(\x,\,t)\,$, thanks to the irrotationality assumption. The potential flow formulation with free surface is known as the \emph{full water wave problem} and it was shown in numerous studies to be an excellent model for water waves \cite{Stoker1957, Johnson2004} (see \cite{Craik2004} for water wave problem history review).


\subsubsection{Initial conditions}
\label{sec:icond}

In order to obtain a well-posed problem, we have to specify also the appropriate initial conditions. The free surface elevation in the parametric form is initially given by two functions:
\begin{equation*}
  \xi\,(q,\,0)\ =\ \xi_{\,0}\,(q)\,, \qquad
  \eta\,(q,\,0)\ =\ \eta_{\,0}\,(q)\,, \qquad q\ \in\ [\,0,\,1\,]\,.
\end{equation*}
Let us assume that the initial distribution of fluid particles is known as well:
\begin{equation*}
  \u\,(\x,\,0)\ =\ \u_{\,0}\,(\x)\,, \qquad \x\ \in\ \Omega\,(0)\,.
\end{equation*}
Then, we can construct an initial condition for the velocity potential $\phi$ in the following way:
\begin{equation*}
  \phi\,(\x,\,0)\ =\ \phi_{\,0}\,(\x)\ \equiv\ \phi_{\,0}\,(\x_{\,0})\ +\ \int_{\gammat_{\x_{\,0}\,\leadsto\x}} u_{\,0}\;\ud x\ +\ v_{\,0}\;\ud y\,,
\end{equation*}
where $\phi_{\,0}\,(\x_{\,0})$ is the value of the velocity potential in an arbitrary point $\x_{\,0}\ \in\ \bar{\Omega}\,(0)$ and $\gammat_{\x_{\,0}\,\leadsto\x}\ \subseteq\ \bar{\Omega}\,(0)$ is an arbitrary piecewise smooth curve connecting points $\x_{\,0}$ with $\x\,$. We assume that the initial distribution of the velocities $\u_{\,0}\,(\x)$ is potential. Thus, the curvilinear integral above does not depend on the path $\gammat_{\x_{\,0}\,\leadsto\x}\,$.


\subsection{Piston motion}

In order to describe the horizontal motion of the piston we adopt the following very simple model based on the second law of \textsc{Newton} \cite{Newton1687, Landau1976}:
\begin{equation}\label{eq:piston}
  m\,\ddot{s}\ +\ k\,s\ =\ -\,\bigl[\,\F\,(t)\ -\ \F\,(0)\,\bigr]\,,
\end{equation}
where $m$ is the wall mass and $k$ is the stiffness coefficient of springs. Finally, $\F\,(t)$ is the force acting on the left moving wall. It can be computed by integrating the contributions of the whole water column
\begin{equation*}
  \F\,(t)\ =\ \int_{\,-h_0}^{\,\eta_s\,(t)}\, p\,\bigl(s(t),\, y,\, t\bigr)\;\ud y\,, \qquad 
  \eta_{\,s}\,(t)\ \eqdef\ \eta\,\bigl(s\,(t),\, t\bigr)\,.
\end{equation*}
If initially the fluid was at rest, the force $\F(0)$ consists only of the hydrostatic loading on the wall, \ie
\begin{equation}\label{eq:f0}
  \F\,(0)\ =\ g\;\frac{h_{\,0}^{\,2}}{2}\,.
\end{equation}
The pressure $p\,(\x,\,t)\ \equiv\ p\,(x,\,y,\,t)$ can be computed at any point inside the fluid thanks to the \textsc{Euler}--\textsc{Lagrange} integral:
\begin{equation*}
  p\,(x,\,y,\,t)\ =\ -\,\phi_{\,t}\ -\ \half\,\abs{\grad\phi}^{\,2}\ -\ g\,y\,.
\end{equation*}
A more detailed derivation of equation~\eqref{eq:piston} can be found in Appendix~\ref{app:pist}.

The second order nonlinear and non-autonomous \acf{ode} \eqref{eq:piston} has to be completed with two initial conditions. If nothing is indicated, we start the integration from the rest state, \ie
\begin{equation*}
  s\,(0)\ =\ \dot{s}\,(0)\ =\ 0\,.
\end{equation*}


\subsection{Dimensionless problem}
\label{sec:dimless}

Equations given in previous Section can be further simplified by choosing appropriate scaled variables. Namely, we introduce the following scaled independent
\begin{equation*}
  x^*\ =\ \frac{x}{h_0}\,, \qquad y^*\ =\ \frac{y}{h_0}\,, \qquad t^*\ =\ t\;\sqrt{\frac{g}{h_0}}
\end{equation*}
and dependent
\begin{equation*}
 s^*\ =\ \frac{s}{h_{\,0}}\,, \qquad \eta^*\ =\ \frac{\eta}{h_{\,0}}\,, \qquad \phi^*\ =\ \frac{\phi}{h_{\,0}\sqrt{g\,h_{\,0}}}\,, \qquad p^*\ =\ \frac{p}{\rho\,g\,h_{\,0}}\,, \qquad \F^*\ =\ \frac{\F}{\rho\,g\,h^{\,3}_{\,0}}
\end{equation*}
variables. The dimensional coefficients and parameters appearing in governing equations scale as follows:
\begin{equation*}
  \ell^*\ =\ \frac{\ell}{h_{\,0}}\,, \qquad
  m^*\ =\ \frac{m}{\rho\, h^{\,3}_{\,0}}\,, \qquad 
  k^*\ =\ \frac{k}{\rho\, g\, h^{\,2}_{\,0}}\,.
\end{equation*}
Finally, dimensionless governing equations read (where for simplicity we drop out all asterisk $*$ symbol from superscripts):
\begin{align}
  \grad^2\,\phi\ &=\ 0\,, \qquad (\x,\,t)\ \in\ \Omega\,(t)\times\R^+\,,\label{eq:laplace} \\
  \eta_{\,t}\ +\ \phi_{\,x}\cdot\eta_{\,x}\ &=\ \phi_{\,y}\,, \qquad y\ =\ \eta\,(x,\,t)\,, \label{eq:kinfs} \\
  \phi_{\,t}\ +\ \half\,\abs{\grad\,\phi}^{\,2}\ +\ \eta\ &=\ 0\,, \qquad y\ =\ \eta\,(x,\,t)\,, \label{eq:dynfs} \\
  \phi_{\,y}\ &=\ 0\,, \qquad y\ =\ -1\,, \label{eq:bottom} \\
  \phi_{\,x}\ &=\ \dot{s}\,, \qquad x\ =\ s\,(t)\,, \quad -1\ \leq\ y\ \leq\ \eta\,(s(t),\, t)\,,\label{eq:lwall} \\
  \phi_{\,x}\ &=\ 0, \qquad x\ =\ \ell\,, \quad -1\ \leq\ y\ \leq\ \eta\,(\ell,\,t)\,,\label{eq:rwall} \\
  m\,\ddot{s}\ +\ k\,s\ &=\ -\,\bigl[\,\F\,(t)\ -\ \F\,(0)\,\bigr]\,, \label{eq:spring} \\
  \F\,(t)\ &=\ \int_{\,-1}^{\,\eta\,\left(s(t),\, t\right)}\, p\,\bigl(s\,(t),\, y,\, t\bigr)\;\ud y\,,\label{eq:force} \\
  p\,(\x,\,t)\ &=\ -\,\phi_t\ -\ \half\,\abs{\grad\,\phi}^{\,2}\ -\ y\,.\label{eq:pressure}
\end{align}
In numerical simulations below we shall solve this dimensionless system of equations. It is equivalent to setting simply $h_0\ =\ 1$ and $g\ =\ 1$ in the computer code.


\section{Equations on a fixed domain}
\label{sec:trans}

The main difficulty of the problem described above is that the computational domain $\Omega\,(t)$ is time dependent. First of all, it is due to the motion of the free surface, but also due to the motion of the left wall. Consequently, the strategy adopted in this study consists in transforming the problem to a fixed domain $\Q^{\,0}\,$. Obviously, this transformation will be time dependent due to the unsteady character of the problem. In the past; the idea of using conformal mappings from Complex Analysis has become popular in 2D Hydrodynamics \cite{Lavrentev1977}. For the 2D water wave problem it was proposed by L.~\textsc{Ovsyannikov} (1974) \cite{Ovsyannikov1974} and implemented much later numerically by A.~\textsc{Dyachenko} \etal (1996) \cite{Dyachenko1996a}. Strictly speaking, in our developments we do not need the conformal property of the underlying mapping. Consequently we shall relax this assumption. Moreover, it will allow us to have a \emph{fixed} and finite transformed domain $\Q^{\,0}\,$. For instance, we can choose $\Q^{\,0}$ as a unit square
\begin{equation*}
  \Q^{\,0}\ \eqdef\ \bigl\{(q^{\,1},\, q^{\,2})\ |\ 0\ \leq\ q^{\,1},\, q^{\,2}\ \leq\ 1\bigr\}\ \equiv\ [0,\,1]^2
\end{equation*}
and we consider a smooth bijective mapping (\ie~a diffeomorphism) $\x\,:\ \Q^{\,0}\ \mapsto\ \Omega\,(t)\,$:
\begin{equation}\label{eq:map}
  \x\ =\ \x\,(\q,\, t) \qquad \Longleftrightarrow \qquad
  \left\{
  \begin{array}{rl}
    x\ &=\ x\,(q^{\,1},\, q^{\,2},\, t)\,, \\
    y\ &=\ y\,(q^{\,1},\, q^{\,2},\, t)\,.
  \end{array}
  \right.
\end{equation}
\begin{equation}\label{eq:diagram}
\begin{tikzcd}[font=\large, column sep = 6em]
  \Omega\,(t) \arrow[bend left = 25]{r}[font=\small]{\x^{-1}\,(\q,\,t)}
  & \Q^{\,0} \arrow[bend left = 25]{l}[font=\small]{\x\,(\q,\,t)}
\end{tikzcd}
\end{equation}
Additionally we assume that the \textsc{Jacobian} matrix $\J\,(\q,\,t)$ of transformation \eqref{eq:map} is non-singular, \ie
\begin{equation}\label{eq:jac}
  \J\,(\q,\,t)\ \eqdef\ \begin{vmatrix}
    x_{\,q^{\,1}} & x_{\,q^{\,2}} \\
    y_{\,q^{\,1}} & y_{\,q^{\,2}}
  \end{vmatrix}\ \equiv\ 
  x_{\,q^{\,1}}\,y_{\,q^{\,2}}\ -\ x_{\,q^{\,2}}\,y_{\,q^{\,1}}\ \neq\ 0\,.
\end{equation}
This situation is schematically depicted in diagram \eqref{eq:diagram} and in Figure~\ref{fig:map}. For more details on the computation of partial derivatives we refer to Appendix~\ref{app:comp}.

It is assumed that left $\Gamma_\ell$ and right $\Gamma_r$ walls are images of left and right sides of the square $\Q^{\,0}\,$, \ie $q^{\,1}\ =\ 0$ and $q^{\,1}\ =\ 1$ correspondingly ($0\ \leq\ q^{\,2}\ \leq\ 1$). Similarly, the upper ($q^{\,2}\ =\ 1$) and lower ($q^{\,2}\ =\ 0$) sides of the square are respectively mapped on the free surface $\Gamma_s$ and bottom $\Gamma_b\,$. The boundary of the square $\Q^{\,0}$ will be denoted by $\gamma\,$, \ie
\begin{equation*}
  \gamma\ \eqdef\ \partial\,\Q^{\,0}\ =\ \gamma_{\,s}\ \cup\ \gamma_{\,b}\ \cup\ \gamma_{\,r}\ \cup\ \gamma_{\,\ell}\,.
\end{equation*}
The practical construction of the mapping $\x\,(\q,\,t)$ is described in Section~\ref{sec:ell}. At the present stage of the exposition we shall assume that mapping $\x\,(\q,\,t)$ is simply known.

\begin{figure}
  \centering
  \includegraphics[width=0.79\textwidth]{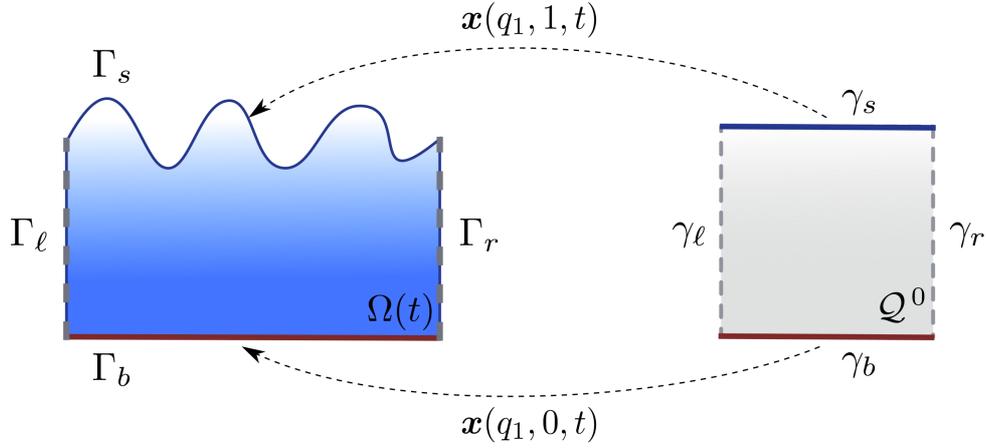}
  \caption{\small\em Schematic illustration of the mapping \eqref{eq:map}. Dashed lines show the correspondence between two boundaries of domains $\Q^{\,0}$ and $\Omega\,(t)$. The same property holds for two other lateral boundaries, but it is not shown for the sake of picture clarity.}
  \label{fig:map}
\end{figure}

In this Section we shall skip the intermediate computations, which are rather standard in the Differential Geometry, for example. We detail them in Appendices~\ref{app:comp} and \ref{app:lapl}. The system of governing equations \eqref{eq:laplace} -- \eqref{eq:pressure} posed on the fixed domain $\Q^{\,0}$ reads:
\begin{align}
  \pd{}{q^{\,\alpha}}\,\Bigl(\K_{\,\alpha\,\beta}\,\pd{\phi}{q^{\,\beta}}\Bigr)\ &=\ 0\,, \qquad \q\ \in\ \Q^{\,0}\,,\quad \alpha,\,\beta\ =\ 1,\,2, \label{eq:laplaceQ} \\
  \eta_{\,t}\ +\ v^{\,1}\cdot\eta_{\,q^{\,1}}\ &=\ v\,, \qquad q^{\,2}\ =\ 1\,, \label{eq:kinfsQ} \\
  \phi_{\,t}\ -\ (u\cdot x_t\ +\ v\cdot y_t)\ +\ \half\, |\grad_{\x}\,\phi|^2\ +\ \eta\ &=\ 0\,, \qquad q^{\,2}\ =\ 1\,, \label{eq:dynfsQ} \\
  \K_{\,2\,1}\,\pd{\phi}{q^{\,1}}\ +\ \K_{\,2\,2}\,\pd{\phi}{q^{\,2}}\ &=\ 0\,, \qquad q^{\,2}\ =\ 0\,, \label{eq:bottomQ} \\
  \K_{\,1\,1}\,\pd{\phi}{q^{\,1}}\ +\ \K_{\,1\,2}\,\pd{\phi}{q^{\,2}}\ &=\ y_{q^{\,2}}\,\dot{s}\,, \qquad q^{\,1}\ =\ 0\,, \quad 0\ \leq\ q^{\,2}\ \leq\ 1\,, \label{eq:lwallQ} \\
  \K_{\,1\,1}\,\pd{\phi}{q^{\,1}}\ +\ \K_{\,1\,2}\,\pd{\phi}{q^{\,2}}\ &=\ 0, \qquad q^{\,1}\ =\ 1, \quad 0\ \leq\ q^{\,2}\ \leq\ 1\,, \label{eq:rwallQ} \\ 
  \F\,(t)\ &=\ \int_{\,0}^{\,1}\, p\,\bigl(0,\, q^{\,2},\, t\bigr)\;\ud q^{\,2}\,, \label{eq:forceQ} \\
  p\, &=\, -\,\phi_{\,t}\, +\, (u\cdot x_t\, +\, v\cdot y_t)\, -\, \half\,|\grad_{\x}\phi|^2\, -\, y\,, \label{eq:pressureQ}
\end{align}
where the implicit convention about the summation over repeated indices is adopted. The transformation of the \textsc{Laplace} equation is detailed in Appendix~\ref{app:lapl}.

\begin{remark}
If the free surface is given in the parametric form \eqref{eq:param}, then we have to add one more kinematic free surface boundary condition:
\begin{equation}\label{eq:kinfsQxi}
  \xi_{\,t}\ +\ v^{\,1}\cdot\xi_{\,q^{\,1}}\ =\ u\,, \qquad \q\ \in\ \gamma_{\,s}\,.
\end{equation}
See also Appendix~\ref{app:bc} for some hints on the derivation of boundary conditions in the transformed domain.
\end{remark}

Notice that equation \eqref{eq:spring} is not affected by transformation \eqref{eq:map}. That is why we do not repeat it here. Above we introduced the following notations. The coefficients in transformed \textsc{Laplacian} are defined as
\begin{equation}\label{eq:K}
  \K_{\,1\,1}\ \eqdef\ \frac{g_{\,2\,2}}{\J}\,, \qquad
  \K_{\,1\,2}\ \equiv\ \K_{\,2\,1}\ \eqdef\ -\frac{g_{\,1\,2}}{\J}\,, \qquad
  \K_{\,2\,2}\ \eqdef\ \frac{g_{\,1\,1}}{\J}\,,
\end{equation}
where $\G = (g_{\alpha\beta})_{\alpha,\beta\, =\, 1,\,2}$ is the familiar metric tensor
\begin{multline}\label{eq:metric}
  g_{\,1\,1}\ \eqdef\ (x_{\,q^{\,1}})^2\ +\ (y_{\,q^{\,1}})^2\,, \qquad
  g_{\,1\,2}\ \equiv\ g_{\,2\,1}\ \eqdef\ x_{\,q^{\,1}} x_{\,q^{\,2}}\ +\ y_{\,q^{\,1}} y_{\,q^{\,2}}\,, \\
  g_{\,2\,2}\ \eqdef\ (x_{\,q^{\,2}})^2\ +\ (y_{\,q^{\,2}})^2\,.
\end{multline}
In fact, tensor $\K\ =\ (\K_{\,\alpha\beta})_{\alpha,\,\beta\, =\, 1,2}$ can be expressed in a compact matrix form through the metric tensor $\G$ as
\begin{equation*}
  \K\ =\ \J\;\G^{-1}\,.
\end{equation*}
The \textsc{Cartesian} $(u,\, v)$ and the contravariant $v^{\,1}$ components of the velocity field are defined as
\begin{align}\label{eq:cartU}
  u\ &\eqdef\ \frac{1}{\J}\,\bigl[\phi_{q^{\,1}}\cdot y_{q^{\,2}}\ -\ \phi_{q^{\,2}}\cdot y_{q^{\,1}}\bigr]\,, \\
  v\ &\eqdef\ \frac{1}{\J}\,\bigl[-\phi_{q^{\,1}}\cdot x_{q^{\,2}}\ +\ \phi_{q^{\,2}}\cdot x_{q^{\,1}}\bigr]\,,\label{eq:cartV} \\
  v^{\,1}\ &\eqdef\ \pd{q^{\,1}}{t}\ +\ u\cdot\pd{q^{\,1}}{x}\ +\ v\cdot\pd{q^{\,1}}{y}\,.\label{eq:cartW}
\end{align}
In order to compute the derivatives of the inverse mapping, we use the following expressions:
\begin{equation}\label{eq:25}
  \pd{q^{\,1}}{t}\ =\ \frac{1}{\J}\,\Bigl[\,\pd{y}{t}\,\pd{x}{q^{\,2}}\ -\ \pd{x}{t}\,\pd{y}{q^{\,2}}\,\Bigr]\,, \qquad
  \pd{q^{\,1}}{x}\ =\ \frac{1}{\J}\,\pd{y}{q^{\,2}}\,, \qquad
  \pd{q^{\,1}}{y}\ =\ -\frac{1}{\J}\,\pd{x}{q^{\,2}}\,.
\end{equation}
\begin{equation}\label{eq:27}
  \pd{q^{\,2}}{t}\ =\ \frac{1}{\J}\,\Bigl[\,\pd{x}{t}\,\pd{y}{q^{\,1}}\ -\ \pd{y}{t}\,\pd{x}{q^{\,1}}\,\Bigr]\,, \qquad
  \pd{q^{\,2}}{x}\ =\ -\frac{1}{\J}\,\pd{y}{q^{\,1}}\,, \qquad
  \pd{q^{\,2}}{y}\ =\ \frac{1}{\J}\,\pd{x}{q^{\,1}}\,.
\end{equation}

The mathematical problem of a potential flow formulation with free surface in curvilinear coordinates consists in determining the velocity potential $\phi\,(q^{\,1},\,q^{\,2},\,t)$ and the free surface profile $\eta\,(q^{\,1},\,t)\,$. The velocity potential satisfies in the domain $\Q^{\,0}$ an elliptic equation \eqref{eq:laplaceQ} with non-constant coefficients along with boundary conditions \eqref{eq:kinfsQ} -- \eqref{eq:rwallQ} on $\partial\Q^{\,0}\,$. In order to obtain a well-posed problem we have to complete the formulation with two initial conditions at $t\ =\ 0$ directly in the fixed domain:
\begin{align*}
  \phi\,(q^{\,1},\,1,\,0)\ &=\ \phi_{\,0}(q^{\,1})\,, \qquad \q\ =\ (q^{\,1},\,1)\ \in\ \Gamma_s\,, \\
  \eta\,(q^{\,1},\,0)\ &=\ \eta_{\,0}\,(q^{\,1})\,, \qquad \q\ =\ (q^{\,1},\,1)\ \in\ \Gamma_s\,.
\end{align*}
Earlier the initial conditions were discussed in Section~\ref{sec:icond}. They can be transposed on the fixed domain $\Q^{\,0}$ using the inverse mapping $\x^{\,-1}\,(\q,\,t)\,$.


\subsection{On some properties of elliptic operators in curvilinear coordinates}
\label{sec:prop}

It can be noticed that the \textsc{Laplace} equation in curvilinear coordinates \eqref{eq:laplaceQ} has a more complex form comparing to the initial \textsc{Cartesian} coordinates $O\,x\,y\,$. In particular, constant coefficients become variable in space and time. Moreover, the mixed derivatives appear as well. However, some properties are important to construct qualitatively correct numerical discretizations. For example, below we shall construct a finite difference scheme for equation \eqref{eq:laplaceQ} with a positive definite difference operator. The proof of this fact relies heavily on the uniform ellipticity property of operator \eqref{eq:laplaceQ}.

\begin{lemma}\label{lem:1}
Partial differential equation \eqref{eq:laplaceQ} is uniformly elliptic.
\end{lemma}

\begin{proof}
By our assumptions the mapping \eqref{eq:map} is differentiable with bounded derivatives in $\Q^{\,0}$ and non-degenerate with a positive \textsc{Jacobian} $\J\ >\ 0\,$, $\forall \q\ \in\ \Q^{\,0}\,$, $\forall t\ >\ 0\,$. Thanks to the identity
\begin{equation*}
  g_{\,1\,1}\cdot g_{\,2\,2}\ \equiv\ \J^{\,2}\ +\ g_{\,1\,2}^{\,2}\,,
\end{equation*}
the metric components $g_{\,1\,1}$ and $g_{\,2\,2}$ take strictly positive values. Thanks to definitions \eqref{eq:K}, it is straightforward to conclude that $\K_{\,1\,1}\,$, $\K_{\,2\,2}$ are also positive. The matrix
\begin{equation*}
  \K\ =\ \begin{pmatrix}
    \K_{\,1\,1} & \K_{\,1\,2} \\
    \K_{\,2\,1} & \K_{\,2\,2}
  \end{pmatrix}
\end{equation*}
is symmetric and its determinant is equal to $\det\K\ =\ 1$ due to \eqref{eq:K}. Consequently, it possesses two orthonormal eigenvectors corresponding to eigenvalues
\begin{align*}
  0\ \leq\ \lambda_{\K}^{(-)}\ \eqdef\ \frac{g_{\,1\,1}\ +\ g_{\,2\,2}\ -\ \sqrt{\Dd}}{2\,\J}\ \leq\ 1\,, \\
  1\ \leq\ \lambda_{\K}^{(+)}\ \eqdef\ \frac{g_{\,1\,1}\ +\ g_{\,2\,2}\ +\ \sqrt{\Dd}}{2\,\J}\,,
\end{align*}
where the discriminant is defined as
\begin{equation*}
  \Dd\ \eqdef\ \bigl(g_{\,1\,1}\ +\ g_{\,2\,2}\bigr)^2\ -\ 4\,\J^{\,2}\ =\ \bigl(g_{\,1\,1}\ -\ g_{\,2\,2}\bigr)^2\ +\ 4\,g_{\,1\,2}^{\,2}\ \geq\ 0\,.
\end{equation*}
Then, for any real numbers $\forall \xi\,$, $\zeta\ \in\ \R$ and for any point $\forall\q\ \in\ \Q^{\,0}$ we have the following inequalities:
\begin{equation}\label{eq:34}
  c_1\,\bigl(\xi^2\ +\ \zeta^2)\ \leq\ \Qf\,(\xi,\,\zeta)\ \leq\ c_2\,\bigl(\xi^2\ +\ \zeta^2)\,,
\end{equation}
where $\Qf$ is the quadratic form defined as
\begin{equation}\label{eq:quadratic}
  \Qf\,(\xi,\,\zeta)\ \eqdef\ \K_{\,1\,1}\,\xi^{\,2}\ +\ 2\,\K_{\,1\,2}\,\xi\,\zeta\ +\ \K_{\,2\,2}\,\zeta^{\,2}\,.
\end{equation}
The constants $c_{\,1,\,2}$ are defined as
\begin{align}\label{eq:c1}
  c_1\ &\eqdef\ \inf_{\q\ \in\ \Q^{\,0}}\;\frac{g_{\,1\,1}\ +\ g_{\,2\,2}\ -\ \sqrt{\Dd}}{2\,\J}\ >\ 0\,, \\
  c_2\ &\eqdef\ \sup_{\q\ \in\ \Q^{\,0}}\;\frac{g_{\,1\,1}\ +\ g_{\,2\,2}\ +\ \sqrt{\Dd}}{2\,\J}\ <\ \infty\,.\label{eq:c2}
\end{align}
The positivity of $c_{\,1}\ >\ 0$ implies the uniform ellipticity property.
\end{proof}

The constants $c_{\,1,\,2}$ defined in the proof above can be used to estimate the convergence speed of iterative methods, which depends directly on the conditioning number \cite{Demmel1997} of the matrix corresponding to the difference operator \cite{Samarskii2001}. The conditioning number depends on the ratio $\dfrac{c_2}{c_1}\,$. High ratio implies a poorly conditioned difference operator and, thus, more iterations are needed to converge to the solution within prescribed accuracy. We underline also that constants $c_{\,1,\,2}$ depend only on the mapping \eqref{eq:map}. Below we provide two examples which illustrate situations where the ratio $\dfrac{c_2}{c_1}\ \gg\ 1\,$.


\subsubsection{Example 1}

Let us take a physical domain in the form of a rectangle:
\begin{equation*}
  \Omega_{\,\blacksquare}\ \eqdef\ \bigl\{\x\ =\ (x,\,y)\;\vert\; 0\ \leq\ x\ \leq\ \ell_{\,1}\,, \quad 0\ \leq\ y\ \leq\ \ell_{\,2}\, \bigr\}
\end{equation*}
with substantially different sizes of the walls, \eg we can assume $0\ <\ \ell_{\,1}\ \ll\ \ell_{\,2}\,$. The mapping \eqref{eq:map} from $\Q^{\,0}\ \mapsto\ \Omega_{\,\blacksquare}$ is given explicitly by formulas
\begin{equation}\label{eq:mapS}
  x\ =\ q^{\,1}\cdot\ell_{\,1}\,, \qquad 
  y\ =\ q^{\,2}\cdot\ell_{\,2}\,.
\end{equation}
Then, we can compute explicitly the metric coefficients
\begin{equation*}
  g_{\,1\,1}\ =\ \ell_{\,1}^{\,2}\,, \qquad
  g_{\,1\,2}\ =\ 0\,, \qquad
  g_{\,2\,2}\ =\ \ell_{\,2}^{\,2}\,,
\end{equation*}
the \textsc{Jacobian} 
\begin{equation*}
  \J\ =\ \ell_{\,1}\cdot\ell_{\,2}\,,
\end{equation*}
and eigenvalues of the matrix $\K\,$:
\begin{equation*}
  c_{\,1}\ =\ \frac{\ell_{\,1}}{\ell_{\,2}}\,, \qquad
  c_{\,2}\ =\ \frac{\ell_{\,2}}{\ell_{\,1}}\,.
\end{equation*}
Thus, the conditioning of the finite difference operator will scale with
\begin{equation*}
  \frac{c_{\,2}}{c_{\,1}}\ \equiv\ \Bigl(\frac{\ell_{\,2}}{\ell_{\,1}}\Bigr)^{\,2}\ \gg\ 1\,.
\end{equation*}
Thus, for highly distorted domains the ratio $\dfrac{c_{\,2}}{c_{\,1}}$ becomes large. If on the reference domain $\Q^{\,0}$ we use a uniform square grid, than mapping \eqref{eq:mapS} generates a uniform grid with distorted cells. The convergence of iterative methods on such grids will be slowed down. Thus, in real computations we should avoid highly distorted cells since they will penalize the convergence of linear solvers.


\subsubsection{Example 2}

Let us take another physical domain $\Omega_{\,\square}$ having the shape of a parallelogram. It can be obtained as an image of the reference domain $\Q^{\,0}$ under the following mapping
\begin{equation}\label{eq:mapR}
  x\ =\ q^{\,1}\ +\ q^{\,2}\cdot\cos\psi\,, \qquad
  y\ =\ q^{\,2}\cdot\sin\psi\,,
\end{equation}
where $0\ <\ \psi\ \ll\ \dfrac{\pi}{2}$ is a small angle. Then, the metric coefficients of this mapping are
\begin{equation*}
  g_{\,1\,1}\ =\ 1\,, \qquad
  g_{\,1\,2}\ =\ \cos\psi\,, \qquad
  g_{\,2\,2}\ =\ 1\,.
\end{equation*}
The \textsc{Jacobian} $\J\ =\ \sin\psi$ and eigenvalues of the matrix $\K$ are
\begin{equation*}
  c_{\,1}\ =\ \tan\frac{\psi}{2}\,, \qquad
  c_{\,2}\ =\ \cot\frac{\psi}{2}\,,
\end{equation*}
and henceforth
\begin{equation*}
  \frac{c_{\,2}}{c_{\,1}}\ \equiv\ \cot^{\,2}\frac{\psi}{2}\ \gg\ 1\,.
\end{equation*}
So, we showed that the ratio $\dfrac{c_{\,2}}{c_{\,1}}$ becomes large for domains $\Omega_{\,\square}$ featuring a small angle. The problem is that the grid generated by mapping \eqref{eq:mapR} is substantially non-orthogonal, since it consists of parallelograms with an acute angle $\psi\,$. On such grids we should expect some reduction of iterative methods convergence speed. Thus, in our computations we should avoid grids featuring small angles.


\section{Finite difference scheme in curvilinear coordinates}
\label{sec:fd}

In the numerical simulation of free surface potential flows of an ideal fluid, an elliptic equation to determine the velocity potential $\phi$ has to be solved at every time step. According to our numerical algorithm (see Section~\ref{sec:alg} for more details), we determine first the velocity potential value on the free surface $\Gamma_s$, then using this value we reconstruct the velocity potential $\phi$ in the whole fluid domain (by taking into account other boundary conditions on $\Gamma_b\,$, $\Gamma_r$ and $\Gamma_\ell$). So, in this Section we assume that we know the values of the velocity potential in nodes which constitute the pre-image of the free surface $\Gamma_s\,$. The goal is to determine the values of the velocity potential in all remaining grid nodes. These values will be determined by solving a system of difference equations constructed below.

Moreover, we assume that the curvilinear grid $\Omega_{\,h}^{\,n}$ on the $n$\up{th} time layer is already constructed. The nodes $\x_{\,\j}^{\,n}$ are images under the mapping \eqref{eq:map} of fixed nodes $\q_{\,\j}\,$, which constitute the uniformly distributed grid $\Q^{\,0}_{\,h}\ \subseteq\ \Q^{\,0}\,$. Here $\j\ \eqdef\ (j_{\,1},\,j_{\,2})$ is a multi-index and $\q_{\,\j}\ \eqdef\ \bigl(q_{\,j_{\,1}}^{\,1},\,q_{\,j_{\,2}}^{\,2}\bigr)\,$. The uniform grid is traditionally defined as
\begin{equation*}
  q_{\,j_{\,\alpha}}^{\,\alpha}\ \eqdef\ j_{\,\alpha}\cdot h_{\,\alpha}\,, \qquad
  j_{\,\alpha}\ =\ 0,\,\ldots,\,N_{\,\alpha}\,, \qquad
  h_{\,\alpha}\ \eqdef\ \frac{1}{N_{\,\alpha}}\,, \qquad
  \alpha\ =\ 1,\,2\,.
\end{equation*}
We make an additional geometrical (and not very restrictive) assumption on the mapping \eqref{eq:map}: boundary components $\gamma_{\,\ell}\,$, $\gamma_{\,b}\,$, $\gamma_{\,r}$ and $\gamma_{\,s}$ are mapped on corresponding boundary components $\Gamma_{\ell}\,$, $\Gamma_{b}\,$, $\Gamma_{r}$ and $\Gamma_{s}$ of the fluid domain $\Omega(t)$ (see Figure~\ref{fig:map} for an illustration). The boundary of the reference domain $\Q^{\,0}$ after the discretization is denoted as
\begin{equation*}
  \gamma^{\,h}\ \eqdef\ \partial\,\Q^{\,0}_{\,h}\ \equiv\ \gamma_{\,\ell}^{\,h}\ \cup\ \gamma_{\,b}^{\,h}\ \cup\ \gamma_{\,r}^{\,h}\ \cup\ \gamma_{\,s}^{\,h}\,.
\end{equation*}
The sketch of the discretized reference domain $\Q^{\,0}_{\,h}$ is depicted in Figure~\ref{fig:mesh}.

\begin{figure}
  \centering
  \includegraphics[width=0.59\textwidth]{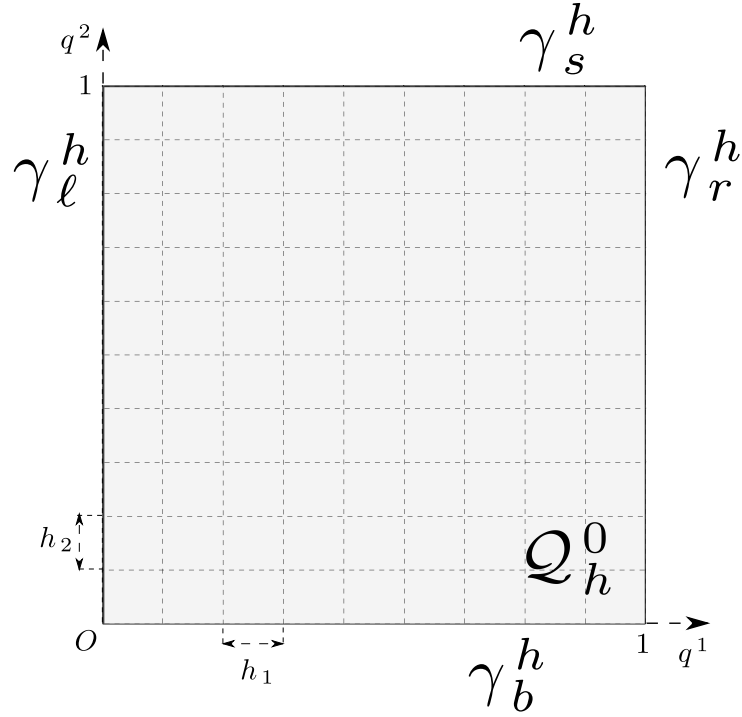}
  \caption{\small\em Sketch of the discretized reference domain $\Q^{\,0}_{\,h}\,$.}
  \label{fig:mesh}
\end{figure}

In order to construct a finite difference approximation we employ the integro-interpolation method using the terminology of \textsc{Tikhonov} \& \textsc{Samarskii} \cite{Tikhonov1963}. In the western literature this method is closer to the finite volume/conservative finite difference methods \cite{Godunov1987, Godunov1999}. The choice for finite differences seems to be natural since we would like to solve an elliptic equation on a simple \textsc{Cartesian} domain \cite{Richtmyer1967}. For this purpose we replace equation~\eqref{eq:laplaceQ} by the following integral relation
\begin{equation}\label{eq:int}
  \sqint_{\partial\,\C}\,\digamma^{\,1}\,\ud q^{\,2}\ -\ \digamma^{\,2}\,\ud q^{\,1}\ =\ 0\,,
\end{equation}
where $\C\ \subseteq\ \Q^{\,0}$ is an arbitrary sub-domain with a piece-wise smooth boundary. We introduced above the notation
\begin{equation*}
  \digamma^{\,1}\ \eqdef\ \K_{\,1\,1}\cdot\pd{\phi}{q^{\,1}}\ +\ \K_{\,1\,2}\cdot\pd{\phi}{q^{\,2}}\,, \qquad
  \digamma^{\,2}\ \eqdef\ \K_{\,2\,1}\cdot\pd{\phi}{q^{\,1}}\ +\ \K_{\,2\,2}\cdot\pd{\phi}{q^{\,2}}\,.
\end{equation*}
Below we construct discrete approximations based on this integral formulation.

\begin{remark}\label{rem:lapl}
Please, notice that the transformed \textsc{Laplace} equation \eqref{eq:laplaceQ} can be compactly rewritten using the new notation:
\begin{equation*}
  \pd{\digamma^{\,1}}{q^{\,1}}\ +\ \pd{\digamma^{\,2}}{q^{\,2}}\ =\ 0\,, \qquad \q\ =\ \bigl(q^{\,1},\,q^{\,2})\ \in\ \Q^{\,0}\,.
\end{equation*}
\end{remark}


\subsection{Approximation in interior nodes}

\begin{figure}
  \centering
  \subfigure[]{\includegraphics[width=0.685\textwidth]{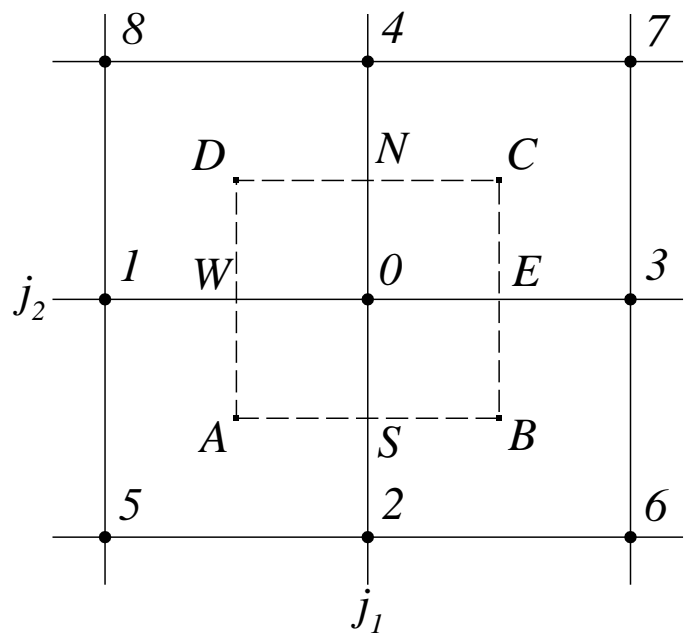}}
  \subfigure[]{\includegraphics[width=0.685\textwidth]{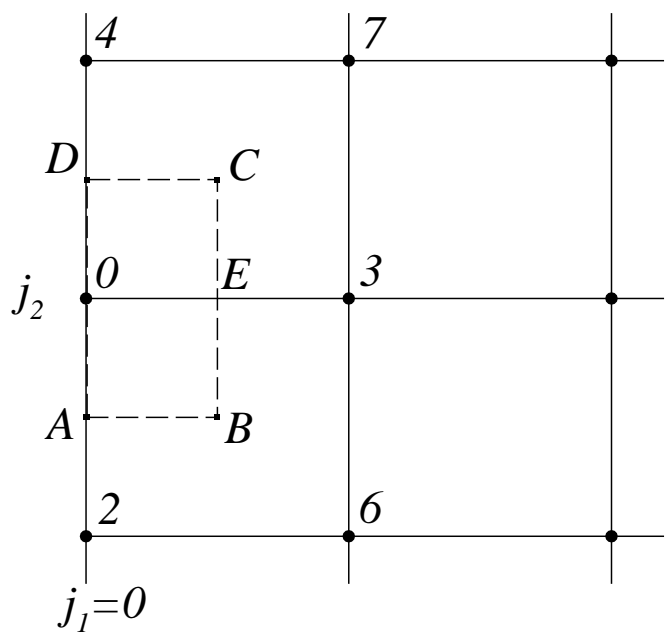}}
  \caption{\small\em Finite difference stencil and the integration contour (shown with a dotted line) in an interior (a) and a left boundary (b) node.}
  \label{fig:stencil}
\end{figure}

Let us consider an interior grid $\Q^{\,0}_{\,h}$ node marked with symbol $0$ in Figure~\ref{fig:stencil}(\textit{a}) along with the integration contour $A\,B\,C\,D\,$. For this contour the integral equation \eqref{eq:int} takes the form
\begin{equation}\label{eq:int1}
  \mint{-}_{\,B\,C}\digamma^{\,1}\,\ud q^{\,2}\ -\ \mint{-}_{\,A\,D}\digamma^{\,1}\,\ud q^{\,2}\ +\ \mint{-}_{\,D\,C}\digamma^{\,2}\,\ud q^{\,1}\ -\ \mint{-}_{\,A\,B}\digamma^{\,2}\,\ud q^{\,1}\ =\ 0\,.
\end{equation}
Now, by applying a quadrature rule to all integrals above, one can obtain a discrete analogue of this integral relation. In the present study we employ the quadrature formula of rectangles. For the integral over the segment $B\,C$ we have
\begin{multline}\label{eq:intBC}
  \mint{-}_{\,B\,C}\digamma^{\,1}\,\ud q^{\,2}\ \approx\ \biggl\{\frac{\K_{\,1\,1}\,(0)\ +\ \K_{\,1\,1}\,(3)}{2}\cdot\frac{\phi_{\,3}\ -\ \phi_{\,0}}{h_{\,1}}\ +\\ \frac{1}{2}\Bigl[\,\K_{\,1\,2}\,(3)\;\frac{\phi_{\,7}\ -\ \phi_{\,6}}{2\,h_{\,2}}\ +\ \K_{\,1\,2}\,(0)\;\frac{\phi_{\,4}\ -\ \phi_{\,2}}{2\,h_{\,2}}\,\Bigr]\biggr\}\,h_{\,2}\ =\\ \frac{1}{2}\;\Bigl\{\bigl(\K_{\,1\,1}\,\phi_{\,q^{\,1}}^{\,\natural}\bigr)_{j_{\,1},\,j_{\,2}}\ +\ \bigl(\K_{\,1\,1}\,\phi_{\,q^{\,1}}^{\,\flat}\bigr)_{j_{\,1}+1,\,j_{\,2}}\ +\\ \frac{1}{2}\;\Bigl[\,\bigl(\K_{\,1\,2}\,\phi_{\,q^{\,2}}^{\,\natural}\bigr)_{j_{\,1}+1,\,j_{\,2}}\ +\ \bigl(\K_{\,1\,2}\,\phi_{\,q^{\,2}}^{\,\flat}\bigr)_{j_{\,1}+1,\,j_{\,2}}\ +\ \bigl(\K_{\,1\,2}\,\phi_{\,q^{\,2}}^{\,\natural}\bigr)_{j_{\,1},\,j_{\,2}}\ +\ \bigl(\K_{\,1\,2}\,\phi_{\,q^{\,2}}^{\,\flat}\bigr)_{j_{\,1},\,j_{\,2}}\,\Bigr]\Bigr\}\,h_{\,2}\,,
\end{multline}
where we used the grid numbering shown in Figure~\ref{fig:stencil}(\textit{a}) and we introduced the following notation
\begin{align*}
  \phi_{\,q^{\,1},\,\j}^{\,\natural}\ &\eqdef\ \frac{\phi_{j_{\,1}+1,\,j_{\,2}}\ -\ \phi_{j_{\,1},\,j_{\,2}}}{h_{\,1}}\,, \qquad \phi_{\,q^{\,1},\,\j}^{\,\flat}\ \eqdef\ \frac{\phi_{j_{\,1},\,j_{\,2}}\ -\ \phi_{j_{\,1}-1,\,j_{\,2}}}{h_{\,1}}\,, \\
  \phi_{\,q^{\,2},\,\j}^{\,\natural}\ &\eqdef\ \frac{\phi_{j_{\,1},\,j_{\,2}+1}\ -\ \phi_{j_{\,1},\,j_{\,2}}}{h_{\,2}}\,, \qquad \phi_{\,q^{\,2},\,\j}^{\,\flat}\ \eqdef\ \frac{\phi_{j_{\,1},\,j_{\,2}}\ -\ \phi_{j_{\,1},\,j_{\,2}-1}}{h_{\,2}}\,.
\end{align*}
for left- and right-sided finite difference operators. By using similar approximations to discretize other integrals over sides $A\,D\,$, $D\,C$ and $A\,B$ we obtain a fully discrete analogue of the integral equation \eqref{eq:int}:
\begin{multline*}
  \biggl\{\frac{1}{2}\;\Bigl[\,\bigl(\K_{\,1\,1}\,\phi_{\,q^{\,1}}^{\,\natural}\bigr)_{\,j_{\,1},\,j_{\,2}}\ +\ \bigl(\K_{\,1\,1}\,\phi_{\,q^{\,1}}^{\,\flat}\bigr)_{\,j_{\,1}+1,\,j_{\,2}}\,\Bigr]\ +\\
  \frac{1}{4}\;\Bigl[\,\bigl(\K_{\,1\,2}\,\phi_{\,q^{\,2}}^{\,\natural}\bigr)_{\,j_{\,1}+1,\,j_{\,2}}\ +\ \bigl(\K_{\,1\,2}\,\phi_{\,q^{\,2}}^{\,\flat}\bigr)_{\,j_{\,1}+1,\,j_{\,2}}\ +\ \bigl(\K_{\,1\,2}\,\phi_{\,q^{\,2}}^{\,\natural}\bigr)_{\,j_{\,1},\,j_{\,2}}\ +\ \bigl(\K_{\,1\,2}\,\phi_{\,q^{\,2}}^{\,\flat}\bigr)_{\,j_{\,1},\,j_{\,2}}\,\Bigr]\biggr\}\,h_{\,2}\\
  -\ \biggl\{\frac{1}{2}\;\Bigl[\,\bigl(\K_{\,1\,1}\,\phi_{\,q^{\,1}}^{\,\natural}\bigr)_{\,j_{\,1}-1,\,j_{\,2}}\ +\ \bigl(\K_{\,1\,1}\,\phi_{\,q^{\,1}}^{\,\flat}\bigr)_{\,j_{\,1},\,j_{\,2}}\,\Bigr]\ +\\
  \frac{1}{4}\;\Bigl[\,\bigl(\K_{\,1\,2}\,\phi_{\,q^{\,2}}^{\,\natural}\bigr)_{\,j_{\,1},\,j_{\,2}}\ +\ \bigl(\K_{\,1\,2}\,\phi_{\,q^{\,2}}^{\,\flat}\bigr)_{\,j_{\,1},\,j_{\,2}}\ +\ \bigl(\K_{\,1\,2}\,\phi_{\,q^{\,2}}^{\,\natural}\bigr)_{\,j_{\,1}-1,\,j_{\,2}}\ +\ \bigl(\K_{\,1\,2}\,\phi_{\,q^{\,2}}^{\,\flat}\bigr)_{\,j_{\,1}-1,\,j_{\,2}}\,\Bigr]\biggr\}\,h_{\,2}\\
  +\ \biggl\{\frac{1}{2}\;\Bigl[\,\bigl(\K_{\,2\,2}\,\phi_{\,q^{\,2}}^{\,\natural}\bigr)_{\,j_{\,1},\,j_{\,2}}\ +\ \bigl(\K_{\,2\,2}\,\phi_{\,q^{\,2}}^{\,\flat}\bigr)_{\,j_{\,1},\,j_{\,2}+1}\,\Bigr]\ +\\
  \frac{1}{4}\;\Bigl[\,\bigl(\K_{\,2\,1}\,\phi_{\,q^{\,1}}^{\,\natural}\bigr)_{\,j_{\,1},\,j_{\,2}+1}\ +\ \bigl(\K_{\,2\,1}\,\phi_{\,q^{\,1}}^{\,\flat}\bigr)_{\,j_{\,1},\,j_{\,2}+1}\ +\ \bigl(\K_{\,2\,1}\,\phi_{\,q^{\,1}}^{\,\natural}\bigr)_{\,j_{\,1},\,j_{\,2}}\ +\ \bigl(\K_{\,2\,1}\,\phi_{\,q^{\,1}}^{\,\flat}\bigr)_{\,j_{\,1},\,j_{\,2}}\,\Bigr]\biggr\}\,h_{\,1}\\
  -\ \biggl\{\frac{1}{2}\;\Bigl[\,\bigl(\K_{\,2\,2}\,\phi_{\,q^{\,2}}^{\,\natural}\bigr)_{\,j_{\,1},\,j_{\,2}-1}\ +\ \bigl(\K_{\,2\,2}\,\phi_{\,q^{\,2}}^{\,\flat}\bigr)_{\,j_{\,1},\,j_{\,2}}\,\Bigr]\ +\\
  \frac{1}{4}\;\Bigl[\,\bigl(\K_{\,2\,1}\,\phi_{\,q^{\,1}}^{\,\natural}\bigr)_{\,j_{\,1},\,j_{\,2}}\ +\ \bigl(\K_{\,2\,1}\,\phi_{\,q^{\,1}}^{\,\flat}\bigr)_{\,j_{\,1},\,j_{\,2}}\ +\ \bigl(\K_{\,2\,1}\,\phi_{\,q^{\,1}}^{\,\natural}\bigr)_{\,j_{\,1},\,j_{\,2}-1}\ +\ \bigl(\K_{\,2\,1}\,\phi_{\,q^{\,1}}^{\,\flat}\bigr)_{\,j_{\,1},\,j_{\,2}-1}\,\Bigr]\biggr\}\,h_{\,1}\ =\ 0\,.
\end{multline*}
By dividing the both parts of the last equation by the area $h_{\,1}\cdot h_{\,2}$ of the rectangle $A\,B\,C\,D\,$, we obtain more compact difference equations in any interior node $\q_{\,\j}\ \in\ \Q_{\,h}^{\,0}\,$:
\begin{multline*}
  \frac{1}{2}\;\Bigl\{\bigl(\K_{\,1\,1}\,\phi_{\,q^{\,1}}^{\,\natural}\bigr)_{\,q^{\,1}}^{\,\flat}\ +\ \bigl(\K_{\,1\,1}\,\phi_{\,q^{\,1}}^{\,\flat}\bigr)_{\,q^{\,1}}^{\,\natural}\Bigr\}_{\,\j}\ +\\ 
  \frac{1}{4}\;\Bigl\{\bigl(\K_{\,1\,2}\,\phi_{\,q^{\,2}}^{\,\natural}\ +\ \K_{\,1\,2}\,\phi_{\,q^{\,2}}^{\,\flat}\bigr)_{\,q^{\,1}}^{\,\natural}\ +\ \bigl(\K_{\,1\,2}\,\phi_{\,q^{\,2}}^{\,\natural}\ +\ \K_{\,1\,2}\,\phi_{\,q^{\,2}}^{\,\flat}\bigr)_{\,q^{\,1}}^{\,\flat}\Bigr\}_{\,\j}\ + \\
  \frac{1}{2}\;\Bigl\{\bigl(\K_{\,2\,2}\,\phi_{\,q^{\,2}}^{\,\natural}\bigr)_{\,q^{\,2}}^{\,\flat}\ +\ \bigl(\K_{\,2\,2}\,\phi_{\,q^{\,2}}^{\,\flat}\bigr)_{\,q^{\,2}}^{\,\natural}\Bigr\}_{\,\j}\ +\\
  \frac{1}{4}\;\Bigl\{\bigl(\K_{\,2\,1}\,\phi_{\,q^{\,1}}^{\,\natural}\ +\ \K_{\,2\,1}\,\phi_{\,q^{\,1}}^{\,\flat}\bigr)_{\,q^{\,2}}^{\,\natural}\ +\ \bigl(\K_{\,2\,1}\,\phi_{\,q^{\,1}}^{\,\natural}\ +\ \K_{\,2\,1}\,\phi_{\,q^{\,1}}^{\,\flat}\bigr)_{\,q^{\,2}}^{\,\flat}\Bigr\}_{\,\j}\ =\ 0\,. 
\end{multline*}
These equations approximate the original differential equation \eqref{eq:laplaceQ} to the order $\O\bigl(h_{\,1}^{\,2}\ +\ h_{\,2}^{\,2}\bigr)$ provided that solutions are sufficiently smooth. This is achieved by using the $9-$point stencil shown in Figure~\ref{fig:stencil}(\textit{a}).


\subsection{Approximation of boundary conditions}
\label{sec:bc}

For the sake of manuscript conciseness we explain the boundary conditions treatment on the example of condition \eqref{eq:lwallQ} imposed on the fixed\footnote{The left wall might be moving in the physical space. However, it becomes fixed in the transformed domain.} boundary $\gamma_{\,\ell}\,$. Let $\gamma_{\,\ell}^{\,h}\ \eqdef\ \bigl\{\q_{\,0,\,j_{\,2}}\ \in\ \gamma_{\,h}\ \vert\ j_{\,2}\ =\ 1,\,2,\,\ldots,\,N_2-1\bigr\}$ be the set of `interior' grid nodes belonging to the left boundary $\gamma_{\,\ell}$ except two angular nodes $\q_{\,0,\,N_2}$ and $\q_{\,0,\,0}\,$, which deserve a special consideration. For a boundary node $\q_{\,\j}\ \in\ \gamma_{\,\ell}^{\,h}$ we choose the integration contour as it is shown in Figure~\ref{fig:stencil}(\textit{b}). The quadrature rule for the integral over side $B\,C$ is given in equation \eqref{eq:intBC} (one has just to substitute $j_{\,1}\ \hookleftarrow\ 0$). On the segment $A\,D$ we have to use the boundary condition \eqref{eq:lwallQ} that we can rewrite in a compact form:
\begin{equation}\label{eq:bc1}
  \digamma^{\,1}\,\bigr\vert_{\q\,\in\,\gamma_{\,\ell}}\ =\ y_{\,q^{\,2}}\,\,\dot{s}\ \defeq\ \mu^{\,0}\,.
\end{equation}
Using this information, we have directly the following approximation:
\begin{equation*}
  \mint{-}_{AD}\digamma^{\,1}\;\ud q^{\,2}\ \approxeq\ \mu_{\,j_{\,2}}^{\,0}\,h_{\,2}\,.
\end{equation*}
For the side $D\,C$ the formula of rectangles yield the following approximation of the integral:
\begin{multline*}
  \mint{-}_{DC}\digamma^{\,2}\;\ud q^{\,1}\ \approxeq\ \biggl\{\frac{\K_{\,2\,2}\,(0)\ +\ \K_{\,2\,2}\,(4)}{2}\cdot\frac{\phi_{\,4}\ -\ \phi_{\,0}}{h_{\,2}}\ +\\
  \frac{1}{2}\;\Bigl[\,\K_{\,2\,1}\,(4)\cdot\frac{\phi_{\,7}\ -\ \phi_{\,4}}{h_{\,1}}\ +\ \K_{\,2\,1}\,(0)\cdot\frac{\phi_{\,3}\ -\ \phi_{\,0}}{h_{\,1}}\,\Bigr]\biggr\}\;\frac{h_{\,1}}{2}\ = \\
  \biggl\{\frac{1}{2}\;\Bigl[\,\bigl(\K_{\,2\,2}\,\phi_{\,q^{\,2}}^{\,\natural}\bigr)_{\,j_{\,1},\,j_{\,2}}\ +\ \bigl(\K_{\,2\,2}\,\phi_{\,q^{\,2}}^{\,\flat}\bigr)_{\,j_{\,1},\,j_{\,2}+1}\,\Bigr]\ +\ \frac{1}{2}\;\Bigl[\,\bigl(\K_{\,2\,1}\,\phi_{\,q^{\,1}}^{\,\natural}\bigr)_{\,j_{\,1},\,j_{\,2}+1}\ +\ \bigl(\K_{\,2\,1}\,\phi_{\,q^{\,1}}^{\,\natural}\bigr)_{\,j_{\,1},\,j_{\,2}}\,\Bigr]\biggr\}\;\frac{h_{\,1}}{2}\,.
\end{multline*}
For the segment $A\,B$ we give only the final result without intermediate calculations:
\begin{multline*}
  \mint{-}_{AB}\digamma^{\,2}\;\ud q^{\,1}\ \approxeq\ \biggl\{\frac{1}{2}\;\Bigl[\,\bigl(\K_{\,2\,2}\,\phi_{\,q^{\,2}}^{\,\natural}\bigr)_{\,j_{\,1},\,j_{\,2}-1}\ +\ \bigl(\K_{\,2\,2}\,\phi_{\,q^{\,2}}^{\,\flat}\bigr)_{\,j_{\,1},\,j_{\,2}}\,\Bigr]\ +\\
  \frac{1}{2}\;\Bigl[\,\bigl(\K_{\,2\,1}\,\phi_{\,q^{\,1}}^{\,\natural}\bigr)_{\,j_{\,1},\,j_{\,2}}\ +\ \bigl(\K_{\,2\,1}\,\phi_{\,q^{\,1}}^{\,\natural}\bigr)_{\,j_{\,1},\,j_{\,2}-1}\,\Bigr]\biggr\}\;\frac{h_{\,1}}{2}\,.
\end{multline*}
Now we substitute all these approximations into the integral equation \eqref{eq:int1} and after dividing its both parts by $h_{\,2}$, we obtain the following difference relation in any node $\gamma_{\,\ell}^{\,h}\,$:
\begin{multline}\label{eq:45}
  \bigl(\K_{\,1\,1}\bigr)_{\,j_{\,1}+\frac{1}{2},\,j_{\,2}}\,\bigl(\phi_{\,q^{\,1}}^{\,\natural}\bigr)_{\,\j}\ +\ \frac{1}{4}\;\Bigl[\,\bigl(\K_{\,1\,2}\,\phi_{\,q^{\,2}}^{\,\natural}\ +\ \K_{\,1\,2}\,\phi_{\,q^{\,2}}^{\,\flat}\bigr)_{\,j_{\,1}+1,\,j_{\,2}}\ +\ \bigl(\K_{\,1\,2}\,\phi_{\,q^{\,2}}^{\,\natural}\ +\ \K_{\,1\,2}\,\phi_{\,q^{\,2}}^{\,\flat}\bigr)_{\,\j}\,\Bigr]\ + \\
  \biggl\{\frac{1}{2}\;\Bigl[\,\bigl(\K_{\,2\,2}\,\phi_{\,q^{\,2}}^{\,\natural}\bigr)_{\,q^{\,2}}^{\,\flat}\ +\ \bigl(\K_{\,2\,2}\,\phi_{\,q^{\,2}}^{\,\flat}\bigr)_{\,q^{\,2}}^{\,\natural}\,\Bigr]\ +\ \frac{1}{2}\;\Bigl[\,\bigl(\K_{\,2\,1}\,\phi_{\,q^{\,1}}^{\,\natural}\bigr)_{\,q^{\,2}}^{\,\natural}\ +\ \bigl(\K_{\,2\,1}\,\phi_{\,q^{\,1}}^{\,\natural}\bigr)_{\,q^{\,2}}^{\,\flat}\,\Bigr]\biggr\}_{\,\j}\;\frac{h_{\,1}}{2}\ =\ \mu_{\,j_{\,2}}^{\,0}\,.
\end{multline}
Here we introduced the following notation:
\begin{equation*}
  \bigl(\,\K_{\,1\,1}\bigr)_{\,j_{\,1}+\frac{1}{2},\,j_{\,2}}\ \eqdef\ \frac{\bigl(\,\K_{\,1\,1}\bigr)_{\,j_{\,1}+1,\,j_{\,2}}\ +\ \bigl(\,\K_{\,1\,1}\bigr)_{\,j_{\,1},\,j_{\,2}}}{2}\,.
\end{equation*}
This difference equation is inhomogeneous because of the wall motion. If the (left) wall is fixed, then $\mu^{\,0}\ \equiv\ 0$ and this relation becomes homogeneous.

\begin{remark}
Please, notice that in all formulas in this Section $j_{\,1}\ \equiv\ 0\,$.
\end{remark}


\subsubsection{Consistency}

In this Section we show that equation \eqref{eq:45} approximates the boundary condition \eqref{eq:lwallQ} to the order $\O(h_{\,1}^{\,2}\ +\ h_{\,2}^{\,2})\,$. Indeed, let us substitute a smooth solution $\phi\,(q^{\,1},\,q^{\,2},\,t)$ into equation \eqref{eq:45} and perform local \textsc{Taylor} expansions. Then, the consistency error $\psi$ can be computed:
\begin{multline*}
  \psi_{\,0,\,j_{\,2}}\ =\ \Bigl(\K_{\,1\,1}\;\pd{\phi}{q^{\,1}}\Bigr)(E)\ +\ \O\bigl(h_{\,1}^{\,2}\bigr)\ +\ \frac{1}{2}\;\biggl\{\Bigl(\K_{\,1\,2}\;\pd{\phi}{q^{\,2}}\Bigr)(3)\ +\ \Bigl(\K_{\,1\,2}\;\pd{\phi}{q^{\,2}}\Bigr)(0)\ +\ \O\bigl(h_{\,2}^{\,2}\bigr)\biggr\}\ +\\ \biggl\{\pd{}{q^{\,2}}\Bigl(\K_{\,2\,2}\;\pd{\phi}{q^{\,2}}\Bigr)(0)\ +\ \O\bigl(h_{\,2}^{\,2}\bigr)\biggr\}\;\frac{h_{\,1}}{2}\ +\\ \biggl\{\K_{\,2\,1}\,\Bigl[\,\pd{\phi}{q^{\,1}}\ +\ h_{\,1}\;\pd{{}^{\,2}\phi}{(q^{\,1})^{\,2}}\ +\ \O\bigl(h_{\,1}^{\,2}\bigr)\,\Bigr]\,(4)\ -\ \K_{\,2\,1}\,\Bigl[\,\pd{\phi}{q^{\,1}}\ +\ h_{\,1}\;\pd{{}^{\,2}\phi}{(q^{\,1})^{\,2}}\ +\ \O\bigl(h_{\,1}^{\,2}\bigr)\,\Bigr]\,(2)\biggr\}\;\frac{h_{\,1}}{4\,h_{\,2}}\ -\ \mu_{\,j_{\,2}}^{\,0}\\ =\ \K_{\,1\,1}\;\pd{\phi}{q^{\,1}}\ +\ \frac{h_{\,1}}{2}\;\pd{}{q^{\,1}}\Bigl(\K_{\,1\,1}\,\pd{\phi}{q^{\,1}}\Bigr)\ +\ \K_{\,1\,2}\;\pd{\phi}{q^{\,2}}\ +\ \frac{h_{\,1}}{2}\;\pd{}{q^{\,1}}\Bigl(\K_{\,1\,2}\;\pd{\phi}{q^{\,2}}\Bigr)\ + \\
  \frac{h_{\,1}}{2}\;\pd{}{q^{\,2}}\Bigl(\K_{\,2\,2}\;\pd{\phi}{q^{\,2}}\Bigr)\ +\ \frac{h_{\,1}}{2}\;\pd{}{q^{\,2}}\Bigl(\K_{\,2\,1}\;\pd{\phi}{q^{\,1}}\Bigr)\ -\ \mu_{\,j_{\,2}}^{\,0}\ =\\ 
  \digamma^{\,1}\ -\ \mu^{\,0}\ +\ \biggl\{\pd{\digamma^{\,1}}{q^{\,1}}\ +\ \pd{\digamma^{\,2}}{q^{\,2}}\biggr\}\;\frac{h_{\,1}}{2}\ +\ \O\bigl(h_{\,1}^{\,2}\ +\ h_{\,2}^{\,2}\bigr)\,.
\end{multline*}
In the very last expression all quantities are evaluated in the same node $\q_{\,0,\,j_{\,2}}\,$, $j_{\,2}\ =\ 1,\,2,\,\ldots,\,N_2-1\,$. Moreover, taking into account the boundary condition \eqref{eq:bc1} and assuming that \textsc{Laplace} equation \eqref{eq:laplaceQ} (see also Remark~\ref{rem:lapl}) is fulfilled up to the boundary $\gamma_{\,\ell}\,$, we obtain that
\begin{equation*}
  \psi_{\,0,\,j_{\,2}}\ =\ \O\bigl(h_{\,1}^{\,2}\ +\ h_{\,2}^{\,2}\bigr)\,.
\end{equation*}
Thus, the boundary condition is approximated to the second order in space. In a similar way we can construct finite difference approximations in boundary nodes $\gamma_{\,b}^{\,h}\ \eqdef\ \bigl\{\q_{\,j_{\,1},\,0}\ \in\ \gamma^{\,h}\ \vert\ j_{\,1}\ =\ 1,\,2,\,\ldots,\,N_1-1\bigr\}$ and $\gamma_{\,r}^{\,h}\ \eqdef\ \bigl\{\q_{\,N_1,\,j_{\,2}}\ \in\ \gamma^{\,h}\ \vert\ j_{\,2}\ =\ 1,\,2,\,\ldots,\,N_2-1\bigr\}\,$. In the derivation of these boundary conditions one has to use the following impermeability conditions (compare with conditions \eqref{eq:bottomQ} and \eqref{eq:rwallQ}):
\begin{equation*}
  \digamma^{\,2}\,\bigl\vert_{\,\q\,\in\,\gamma_{\,b}}\ =\ 0\,, \qquad
  \digamma^{\,1}\,\bigl\vert_{\,\q\,\in\,\gamma_{\,r}}\ =\ 0\,.
\end{equation*}
All these difference equations have six point stencils to achieve the second order accuracy: three nodes lie in the interior of the domain $\Q^{\,0}_{\,h}$ and three on the boundary $\gamma^{\,h}\,$. Finally, there are also two angular nodes $\q_{\,0,\,0}$ and $\q_{\,N_1,\,0}\,$. The stencils contain four nodes: one interior and three on the boundaries.


\subsection{Difference operator}

In order to determine the discrete values of the velocity potential $\bigl\{\phi_{\,\j}\bigr\}\,$, we constructed above a finite difference problem with the uniform second order accuracy. In this difference problem one has to find the values of the velocity potential $\phi_{\,\j}$ in all nodes except those on the free surface, where the \textsc{Dirichlet}-type condition is imposed:
\begin{equation*}
  \phi_{\,j_{\,1},\,N_2}\ =\ \mu_{\,j_{\,1}}^{\,s}\,, \qquad
  j_{\,1}\ =\ 0,\,1,\,\ldots,\,N_1\,.
\end{equation*}
In this way we have to determine \emph{only} $\bigl(N_1\ +\ 1\bigr)\cdot N_2$ the values $\bigl\{\phi_{\,\j}\bigr\}_{\,\j}$ in the nodes $\bigl\{\x_{\,j_{\,1},\,j_{\,2}}\bigr\}$ with $j_{\,1}\ =\ 0,\,1,\,\ldots,\,N_1$ and $j_{\,2}\ =\ 0,\,1,\,\ldots,\,N_2-1\,$. In the resulting problem we have precisely $\bigl(N_1\ +\ 1\bigr)\cdot N_2$ equations.

In order to study theoretically the finite difference problem, it will be convenient to make a change of variables in order to get homogeneous \textsc{Dirichlet}'s boundary conditions in nodes $\q_{\,j_{\,1},\,N_2}\,$, $j_{\,1}\ =\ 0,\,1,\,\ldots,\,N_1\,$. For this purpose we introduce a new grid function
\begin{equation*}
  \phit_{\,j_{\,1},\,j_{\,2}}\ \eqdef\ \begin{dcases}
    \ \mu_{\,j_{\,1}}^{\,s}\,, & \quad j_{\,1}\ =\ 0,\,1,\,\ldots,\,N_1,\, \quad j_{\,2}\ =\ N_2\,, \\
    \ 0\,, & \quad j_{\,1}\ =\ 0,\,1,\,\ldots,\,N_1,\, \quad j_{\,2}\ <\ N_2\,,
  \end{dcases}
\end{equation*}
and a new dependent variable
\begin{equation*}
  \phin\ \eqdef\ \phi\ -\ \phit\,.
\end{equation*}
Then, in the nodes $\q_{\,\j}\ \in\ \gamma^{\,h}_{\,s}$ on the free surface the function $\phin_{\,\j}$ will take zero values. The difference problem for $\phin$ will differ from the one for $\phi$ by inhomogeneous right hand sides in the nodes close to the free surface pre-image $\gamma^{\,h}_{\,s}\,$. For the new problem we introduce the operator notation
\begin{equation}\label{eq:46}
  \Lambda\,\phin\ =\ \nuv\,,
\end{equation}
where $\nuv$ is the vector of the right hand sides. The difference operator $\Lambda$ can be decomposed in sub-operators for our convenience:
\begin{equation*}
  \Lambda\ \equiv\ \Lambda_{\,1}\ +\ \Lambda_{\,2}\,,
\end{equation*}
and on one more level:
\begin{equation*}
  \Lambda_{\,1}\ \equiv\ \Lambda_{\,1\,1}\ +\ \Lambda_{\,1\,2}\,, \qquad
  \Lambda_{\,2}\ \equiv\ \Lambda_{\,2\,1}\ +\ \Lambda_{\,2\,2}\,.
\end{equation*}
Operators $\bigl\{\Lambda_{\,\alpha\,\beta}\bigr\}_{\,1\,\leq\,\alpha,\,\beta\,\leq\,2}$ are defined below as follows.
\begin{equation}\label{eq:lambda11}
  \Lambda_{\,1\,1}\,\phin_{\,\j}\ \eqdef\ \begin{dcases}
  \ \frac{1}{2}\Bigl[\,\bigl(\K_{\,1\,1}\,\phin_{\,q^{\,1}}^{\,\natural}\bigr)_{\,q^{\,1}}^{\,\flat}\ +\ \bigl(\K_{\,1\,1}\,\phin_{\,q^{\,1}}^{\,\flat}\bigr)_{\,q^{\,1}}^{\,\natural}\,\Bigr]_{\,\j}\,, & 0\ <\ j_{\,1}\ <\ N_1\,, \quad 0\ \leq\ j_{\,2}\ <\ N_2\,, \\
  \ \frac{2}{h_{\,1}}\;\bigl(\K_{\,1\,1}\bigr)_{\frac{1}{2},\,j_{\,2}}\,\bigl(\phin_{\,q^{\,1}}^{\,\natural}\bigr)_{\,0,\,j_{\,2}}\,, & j_{\,1}\ =\ 0\,, \quad 0\ \leq\ j_{\,2}\ < N_2\,, \\
  \ -\frac{2}{h_{\,1}}\;\bigl(\K_{\,1\,1}\bigr)_{N_1-\frac{1}{2},\,j_{\,2}}\,\bigl(\phin_{\,q^{\,1}}^{\,\flat}\bigr)_{\,N_1,\,j_{\,2}}\,, & j_{\,1}\ =\ N_1\,, \quad 0\ \leq\ j_{\,2}\ < N_2\,.
  \end{dcases}
\end{equation}
\begin{equation}\label{eq:lambda12}
  \Lambda_{\,1\,2}\,\phin_{\,\j}\ \eqdef\ \begin{dcases}
  \ \frac{1}{4}\;\Bigl[\,\bigl(\K_{\,1\,2}\,\phin_{\,q^{\,2}}^{\,\natural}\ +\ \K_{\,1\,2}\,\phin_{\,q^{\,2}}^{\,\flat}\bigr)_{\,q^{\,1}}^{\,\natural}\ +\ \bigl(\K_{\,1\,2}\,\phin_{\,q^{\,2}}^{\,\natural}\ +\ \K_{\,1\,2}\,\phin_{\,q^{\,2}}^{\,\flat}\bigr)_{\,q^{\,1}}^{\,\flat}\,\Bigr]_{\,\j}\,, & \\
  \qquad\qquad\qquad 0\ <\ j_{\,1}\ <\ N_1\,, \quad 0\ <\ j_{\,2}\ <\ N_2\,, & \\
  \ \frac{1}{2\,h_{\,1}}\;\Bigl[\,\bigl(\K_{\,1\,2}\,\phin_{\,q^{\,2}}^{\,\natural}\ +\ \K_{\,1\,2}\,\phin_{\,q^{\,2}}^{\,\flat}\bigr)_{\,1,\,j_{\,2}}\ +\ \bigl(\K_{\,1\,2}\,\phin_{\,q^{\,2}}^{\,\natural}\ +\ \K_{\,1\,2}\,\phin_{\,q^{\,2}}^{\,\flat}\bigr)_{\,0,\,j_{\,2}}\,\Bigr]\,, & \\
  \qquad\qquad\qquad\quad j_{\,1}\ =\ 0\,, \quad 0\ <\ j_{\,2}\ <\ N_2\,, & \\
  \ -\frac{1}{2\,h_{\,1}}\;\Bigl[\,\bigl(\K_{\,1\,2}\,\phin_{\,q^{\,2}}^{\,\natural}\ +\ \K_{\,1\,2}\,\phin_{\,q^{\,2}}^{\,\flat}\bigr)_{\,N_1-1,\,j_{\,2}}\ +\ \bigl(\K_{\,1\,2}\,\phin_{\,q^{\,2}}^{\,\natural}\ +\ \K_{\,1\,2}\,\phin_{\,q^{\,2}}^{\,\flat}\bigr)_{\,N_1,\,j_{\,2}}\,\Bigr]\,, & \\
  \qquad\qquad\qquad\quad j_{\,1}\ =\ N_1\,, \quad 0\ <\ j_{\,2}\ <\ N_2\,, & \\
  \ \frac{1}{2}\;\Bigl[\,\bigl(\K_{\,1\,2}\,\phin_{\,q^{\,2}}^{\,\natural}\bigr)_{\,q^{\,1}}^{\,\natural}\ +\ \bigl(\K_{\,1\,2}\,\phin_{\,q^{\,2}}^{\,\natural}\bigr)_{\,q^{\,1}}^{\,\flat}\,\Bigr]_{\,j_{\,1},\,0}\,, \quad 0\ <\ j_{\,1}\ <\ N_1\,, \quad j_{\,2}\ =\ 0\,,& \\
  \ \frac{1}{h_{\,1}}\;\Bigl[\,\bigl(\K_{\,1\,2}\,\phin_{\,q^{\,2}}^{\,\natural}\bigr)_{\,1,\,0}\ +\ \bigl(\K_{\,1\,2}\,\phin_{\,q^{\,2}}^{\,\natural}\bigr)_{\,0,\,0}\,\Bigr]\,, \qquad j_{\,1}\ =\ 0\,, \quad j_{\,2}\ =\ 0\,,& \\
  \ -\frac{1}{h_{\,1}}\;\Bigl[\,\bigl(\K_{\,1\,2}\,\phin_{\,q^{\,2}}^{\,\natural}\bigr)_{\,N_1-1,\,0}\ +\ \bigl(\K_{\,1\,2}\,\phin_{\,q^{\,2}}^{\,\natural}\bigr)_{\,N_1,\,0}\,\Bigr]\,, \qquad j_{\,1}\ =\ N_1\,, \quad j_{\,2}\ =\ 0\,,&
  \end{dcases}
\end{equation}
\begin{equation}\label{eq:lambda21}
  \Lambda_{\,2\,1}\,\phin_{\,\j}\ \eqdef\ \begin{dcases}
  \ \frac{1}{4}\;\Bigl[\,\bigl(\K_{\,2\,1}\,\phin_{\,q^{\,1}}^{\,\natural}\ +\ \K_{\,2\,1}\,\phin_{\,q^{\,1}}^{\,\flat}\bigr)_{\,q^{\,2}}^{\,\natural}\ +\ \bigl(\K_{\,2\,1}\,\phin_{\,q^{\,1}}^{\,\natural}\ +\ \K_{\,2\,1}\,\phin_{\,q^{\,1}}^{\,\flat}\bigr)_{\,q^{\,2}}^{\,\flat}\,\Bigr]_{\,\j}\,, & \\
  \qquad\qquad\qquad 0\ <\ j_{\,1}\ <\ N_1\,, \quad 0\ <\ j_{\,2}\ <\ N_2\,, & \\
  \ \frac{1}{2}\;\Bigl[\,\bigl(\K_{\,2\,1}\,\phin_{\,q^{\,1}}^{\,\natural}\bigr)_{\,q^{\,2}}^{\,\natural}\ +\ \bigl(\K_{\,2\,1}\,\phin_{\,q^{\,1}}^{\,\natural}\bigr)_{\,q^{\,2}}^{\,\flat}\,\Bigr]_{\,0,\,j_{\,2}}\,, \quad j_{\,1}\ =\ 0\,, \quad 0\ <\ j_{\,2}\ <\ N_2\,, & \\
  \ \frac{1}{2}\;\Bigl[\,\bigl(\K_{\,2\,1}\,\phin_{\,q^{\,1}}^{\,\flat}\bigr)_{\,q^{\,2}}^{\,\natural}\ +\ \bigl(\K_{\,2\,1}\,\phin_{\,q^{\,1}}^{\,\flat}\bigr)_{\,q^{\,2}}^{\,\flat}\,\Bigr]_{\,N_1,\,j_{\,2}}\,, \quad j_{\,1}\ =\ N_1\,, \quad 0\ <\ j_{\,2}\ <\ N_2\,, & \\
  \ \frac{1}{2\,h_{\,2}}\;\Bigl[\,\bigl(\K_{\,2\,1}\,\phin_{\,q^{\,1}}^{\,\natural}\ +\ \K_{\,2\,1}\,\phin_{\,q^{\,1}}^{\,\flat}\bigr)_{\,j_{\,1},\,0}\ +\ \bigl(\K_{\,2\,1}\,\phin_{\,q^{\,1}}^{\,\natural}\ +\ \K_{\,2\,1}\,\phin_{\,q^{\,1}}^{\,\flat}\bigr)_{\,j_{\,1},\,1}\,\Bigr]\,, & \\
  \qquad\qquad\qquad\quad 0\ <\ j_{\,1}\ <\ N_1\,, \quad j_{\,2}\ =\ 0\,, & \\
  \ \frac{1}{h_{\,2}}\;\Bigl[\,\bigl(\K_{\,2\,1}\,\phin_{\,q^{\,1}}^{\,\natural}\bigr)_{\,0,\,1}\ +\ \bigl(\K_{\,2\,1}\,\phin_{\,q^{\,1}}^{\,\natural}\bigr)_{\,0,\,0}\,\Bigr]\,, \qquad j_{\,1}\ =\ 0\,, \quad j_{\,2}\ =\ 0\,, &\\
  \ \frac{1}{h_{\,2}}\;\Bigl[\,\bigl(\K_{\,2\,1}\,\phin_{\,q^{\,1}}^{\,\flat}\bigr)_{\,N_1,\,0}\ +\ \bigl(\K_{\,2\,1}\,\phin_{\,q^{\,1}}^{\,\flat}\bigr)_{\,N_1,\,1}\,\Bigr]\,, \qquad j_{\,1}\ =\ N_1\,, \quad j_{\,2}\ =\ 0\,,
  \end{dcases}
\end{equation}
\begin{equation}\label{eq:lambda22}
  \Lambda_{\,2\,2}\,\phin_{\,\j}\ \eqdef\ \begin{dcases}
  \ \frac{1}{2}\;\Bigl[\,\bigl(\K_{\,2\,2}\,\phin_{\,q^{\,2}}^{\,\natural}\bigr)_{\,q^{\,2}}^{\,\flat}\ +\ \bigl(\K_{\,2\,2}\,\phin_{\,q^{\,2}}^{\,\flat}\bigr)_{\,q^{\,2}}^{\,\natural}\,\Bigr]_{\,\j}\,, & 0\ \leq\ j_{\,1}\ \leq\ N_1\,, \quad 0\ <\ j_{\,2}\ <\ N_2\,, \\
  \ \frac{2}{h_{\,2}}\;\bigl(\K_{\,2\,2}\bigr)_{j_{\,1},\,\frac{1}{2}}\,\bigl(\phi_{\,q^{\,2}}^{\,\natural}\bigr)_{\,j_{\,1},\,0}\,, & 0\ \leq\ j_{\,1}\ \leq\ N_1\,, \quad j_{\,2}\ =\ 0\,.
  \end{dcases}
\end{equation}
Above we used the notation
\begin{align}\label{eq:54a}
  \bigl(\K_{\,\alpha\,\alpha}\bigr)_{j_{\,1}+\frac{1}{2},\,j_{\,2}}\ &\eqdef\ \frac{\bigl(\K_{\,\alpha\,\alpha}\bigr)_{j_{\,1},\,j_{\,2}}\ +\ \bigl(\K_{\,\alpha\,\alpha}\bigr)_{j_{\,1}+1,\,j_{\,2}}}{2}\ \equiv\ \K_{\,\alpha\,\alpha}\,\bigr\vert_{\,\mathrm{E}}\,, \\
  \bigl(\K_{\,\beta\,\beta}\bigr)_{j_{\,1},\,j_{\,2}+\frac{1}{2}}\ &\eqdef\ \frac{\bigl(\K_{\,\beta\,\beta}\bigr)_{j_{\,1},\,j_{\,2}}\ +\ \bigl(\K_{\,\beta\,\beta}\bigr)_{j_{\,1},\,j_{\,2}+1}}{2}\ \equiv\ \K_{\,\beta\,\beta}\,\bigr\vert_{\,\mathrm{N}}\,,\label{eq:54b} \\
  \bigl(\K_{\,\alpha\,\alpha}\bigr)_{j_{\,1}-\frac{1}{2},\,j_{\,2}}\ &\eqdef\ \frac{\bigl(\K_{\,\alpha\,\alpha}\bigr)_{j_{\,1},\,j_{\,2}}\ +\ \bigl(\K_{\,\alpha\,\alpha}\bigr)_{j_{\,1}-1,\,j_{\,2}}}{2}\ \equiv\ \K_{\,\alpha\,\alpha}\,\bigr\vert_{\,\mathrm{W}}\,,\label{eq:54c} \\
  \bigl(\K_{\,\beta\,\beta}\bigr)_{j_{\,1},\,j_{\,2}-\frac{1}{2}}\ &\eqdef\ \frac{\bigl(\K_{\,\beta\,\beta}\bigr)_{j_{\,1},\,j_{\,2}}\ +\ \bigl(\K_{\,\beta\,\beta}\bigr)_{j_{\,1},\,j_{\,2}-1}}{2}\ \equiv\ \K_{\,\beta\,\beta}\,\bigr\vert_{\,\mathrm{S}}\,,\label{eq:54d}
\end{align}
with $\alpha,\,\beta\ =\ 1,\,2\,$. Moreover, all the operators vanish in the nodes of the upper boundary $\gamma_{\,s}^{\,h}\,$, \ie
\begin{equation*}
  \Lambda_{\,\alpha\,\beta}\,\phin_{\,\j}\ \equiv\ 0\,, \qquad \forall\q_{\,\j}\ \in\ \gamma_{\,s}^{\,h}\,, \qquad \alpha,\,\beta\ =\ 1,\,2\,.
\end{equation*}
Below we study the properties of the operator $\Lambda$ to show that the problem \eqref{eq:46} is well-posed.


\subsection{Mathematical study of the finite difference problem for the velocity potential}
\label{sec:study}

Let us define by $\H$ the set of functions defined on the grid $\Q^{\,0}_{\,h}$ and which take zero values on $\gamma_{\,s}^{\,h}\,$, \ie
\begin{equation*}
  \H\ \eqdef\ \bigl\{\phin: \Q^{\,0}_{\,h}\ \mapsto\ \R\ \vert\ \phin\,\bigr\vert_{\q_{\,\j}\ \in\ \gamma_{\,s}^{\,h}}\ \equiv\ 0\,\bigr\}\,.
\end{equation*}
We introduce also on space $\H$ the scalar product
\begin{multline*}
  \scalh{u}{v}\ \eqdef\ \sum_{j_{\,1}\,=\,1}^{N_1\,-\,1}\;\sum_{j_{\,2}\,=\,1}^{N_2\,-\,1} u_{\,j_{\,1},\,j_{\,2}}\,v_{\,j_{\,1},\,j_{\,2}}\,h_1\,h_2\ +\ \sum_{j_{\,2}\,=\,1}^{N_2\,-\,1} u_{\,0,\,j_{\,2}}\,v_{\,0,\,j_{\,2}}\;\frac{h_1\,h_2}{2}\ +\\ 
  \sum_{j_{\,1}\,=\,1}^{N_1\,-\,1} u_{\,j_{\,1},\,0}\,v_{\,j_{\,1},\,0}\;\frac{h_1\,h_2}{2}\ +\ \sum_{j_{\,2}\,=\,1}^{N_2\,-\,1} u_{\,N_1,\,j_{\,2}}\,v_{\,N_1,\,j_{\,2}}\;\frac{h_1\,h_2}{2}\ +\\ 
  u_{\,0,\,0}\,v_{\,0,\,0}\;\frac{h_1\,h_2}{4}\ +\ u_{\,N_1,\,0}\,v_{\,N_1,\,0}\;\frac{h_1\,h_2}{4}\,.
\end{multline*}
This scalar product induces a natural norm
\begin{equation*}
  \normh{u}^{\,2}\ \eqdef\ \scalh{u}{u}\,.
\end{equation*}
Let us define a new operator $\Aa\ \eqdef\ -\Lambda$ and study its properties. The operator $\Lambda$ was introduced in \eqref{eq:46} and rigorously defined earlier.

For any fixed value of index $j_{\,2}\,$, a function $u\ \in\ \H$ can be considered as a function being defined on a 1D grid:
\begin{equation*}
  \Q_{\,h}^{\,1}\ \eqdef\ \bigl\{q_{\,j_{\,1}}^{\,1}\ \in\ [\,0,\,1\,]\ \vert\ j_{\,1}\ =\ 0,\,1,\,\ldots,\, N_1\bigr\}\,.
\end{equation*}
The class of all grid functions defined on $\Q_{\,h}^{\,1}$ will be denoted by $\H_{\,1}\,$. For two arbitrary functions $u,\,v\ \in\ \H$ we define the scalar product on the layer $j_{\,2}$ as
\begin{equation}\label{eq:57}
  \scalhh{u}{v}\ \eqdef\ \sum_{j_{\,1}\,=\,1}^{N_1\,-\,1} u_{\,j_{\,1},\,j_{\,2}}\,v_{\,j_{\,1},\,j_{\,2}}\,h_{\,1}\ +\ u_{\,0,\,j_{\,2}}\,v_{\,0,\,j_{\,2}}\;\frac{h_{\,1}}{2}\ +\ u_{\,N_{\,1},\,j_{\,2}}\,v_{\,N_{\,1},\,j_{\,2}}\;\frac{h_{\,1}}{2}\,.
\end{equation}
Then, it is not difficult to see that
\begin{equation}\label{eq:56}
  \scalh{u}{v}\ \equiv\ \sum_{j_{\,2}\,=\,1}^{N_{\,2}\,-\,1}\scalhh{u}{v}\,h_{\,2}\ +\ \langle\, u,\,v\,\rangle_{\,\H_{\,1},\,0}\;\frac{h_{\,2}}{2}\,.
\end{equation}

However, we can proceed in a different way by considering instead `vertical' slices. Let $\H_{\,2}$ be the set of 1D functions defined on the grid
\begin{equation*}
  \Q_{\,h}^{\,2}\ \eqdef\ \bigl\{q_{\,j_{\,2}}^{\,2}\ \in\ [\,0,\,1\,]\ \vert\ j_{\,2}\ =\ 0,\,1,\,\ldots,\, N_2\bigr\}\,.
\end{equation*}
and taking zero value for $j_{\,2}\ =\ N_{\,2}\,$, \ie
\begin{equation*}
  \H_{\,2}\ \eqdef\ \bigl\{\phin: \Q^{\,2}_{\,h}\ \mapsto\ \R\ \vert\ \phin\,\bigr\vert_{j_{\,2}\,=\,N_{\,2}}\ \equiv\ 0\,\bigr\}\,.
\end{equation*}
Similarly, in $\H$ we can introduce a scalar product on the layer $j_{\,1}$ as follows:
\begin{equation}\label{eq:59}
  \langle\, u,\,v\,\rangle_{\,\H_{\,2},\,j_{\,1}}\ \eqdef\ \sum_{j_{\,2}\,=\,1}^{N_{\,2}\,-\,1} u_{\,j_{\,1},\,j_{\,2}}\,v_{\,j_{\,1},\,j_{\,2}}\,h_{\,2}\ +\ u_{\,j_{\,1},\,0}\,v_{\,j_{\,1},\,0}\;\frac{h_{\,2}}{2}\,,
\end{equation}
and the scalar product in $\H$ can be similarly expressed as
\begin{equation}\label{eq:58}
  \scalh{u}{v}\ \equiv\ \sum_{j_{\,1}\,=\,1}^{N_{\,1}\,-\,1}\langle\, u,\,v\,\rangle_{\,\H_{\,2},\,j_{\,1}}\,h_{\,1}\ +\ \langle\, u,\,v\,\rangle_{\,\H_{\,2},\,0}\;\frac{h_{\,1}}{2}\ +\ \langle\, u,\,v\,\rangle_{\,\H_{\,2},\,N_{\,1}}\;\frac{h_{\,1}}{2}\,.
\end{equation}
It is not difficult to show that formulas \eqref{eq:56} and \eqref{eq:58} provide the same value for $\scalh{u}{v}$ and, thus, they are nothing else than two different ways to express the scalar product in $\H\,$. Intuitively, we can understand it in the following way. Imagine that one wants to compute the sum of all elements of a matrix. To do it, one can move along the lines and then along the columns, which corresponds to \eqref{eq:58}, or vice versa \eqref{eq:56}. Since the addition is commutative, the resulting sum will be the same.


\subsubsection{Some discrete identities}
\label{sec:discr}

Let us remind that for any discrete functions $\bigl\{\alphat_{\,j}\bigr\}_{\,j\,=\,0}^{\,N}\,$, $\bigl\{u_{\,j}\bigr\}_{\,j\,=\,0}^{\,N}$ and $\bigl\{v_{\,j}\bigr\}_{\,j\,=\,0}^{\,N}$ defined on a 1D grid $\Q_{\,h}^{\,1}\ =\ \bigl\{\,q_{\,j}\ \in\ [\,0,\,1\,]\,\bigr\}_{\,j\,=\,0}^{\,N}\,$, the first discrete \textsc{Green} identity can be shown:
\begin{equation}\label{eq:green1}
  \Bigl\langle \bigl(\alphat\,u_{\,q}^{\,\flat}\bigr)_{\,q}^{\,\natural},\,v \Bigr\rangle\ \equiv\ -\,\Bigl\langle \alphat\,u_{\,q}^{\,\flat},\, v_{\,q}^{\,\flat}\,\Bigr]\ +\ \bigl\{ \alphat\,u_{\,q}^{\,\flat}\,v \bigr\}_{\,N}\ -\ \alphat_{\,1}\bigl\{u_{\,q}^{\,\natural}\,v\bigr\}_{\,0}\,,
\end{equation}
and also
\begin{equation}\label{eq:green2}
  \Bigl\langle \bigl(\alphat\,u_{\,q}^{\,\natural}\bigr)_{\,q}^{\,\flat},\,v \Bigr\rangle\ \equiv\ -\,\Bigl[\,\alphat\,u_{\,q}^{\,\natural},\, v_{\,q}^{\,\natural}\,\Bigr\rangle\ +\ \alphat_{\,N-1}\bigl\{\,u_{\,q}^{\,\flat}\,v \bigr\}_{\,N}\ -\ \bigl\{\alphat\,u_{\,q}^{\,\natural}\,v\bigr\}_{\,0}\,,
\end{equation}
where for the sake of simplicity we introduced the following notations
\begin{equation*}
  \bigl\{ \alphat\, \betat \bigr\}_{\,j}\ \eqdef\ \alphat_{\,j}\,\betat_{\,j}\,, \qquad \bigl\{ \alphat\, \betat\, \gammat \bigr\}_{\,j}\ \eqdef\ \alphat_{\,j}\,\betat_{\,j}\,\gammat_{\,j}\,,
\end{equation*}
\begin{equation*}
  \Bigl\langle \alphat,\,\betat \Bigr\rangle\ \eqdef\ h\,\sum_{j\,=\,1}^{N-1} \alphat_{\,j}\,\betat_{\,j}\,, \qquad
  \Bigl\langle \alphat,\, \betat\,\Bigr]\ =\ h\,\sum_{j\,=\,1}^{N}\alphat_{\,j}\,\betat_{\,j}\,, \qquad
  \Bigl[\, \alphat,\, \betat\,\Bigr\rangle\ =\ h\,\sum_{j\,=\,0}^{N-1}\alphat_{\,j}\,\betat_{\,j}\,.
\end{equation*}
Moreover, the following summation by parts formulas hold as well
\begin{align}\label{eq:sum1}
  \bigl\langle u_{\,q}^{\,\natural},\,v\bigr\rangle\ &\equiv\ -\,\bigl\langle u,\,v_{\,q}^{\,\flat}\,\bigr]\ +\ \bigl\{u\,v\bigr\}_{\,N}\ -\ u_{\,1}\,v_{\,0}\,, \\
  \bigl\langle u_{\,q}^{\,\flat},\,v\bigr\rangle\ &\equiv\ -\,\bigl[\,u,\,v_{\,q}^{\,\natural}\,\bigr\rangle\ +\ u_{\,N-1}\,v_{\,N}\ -\ \bigl\{u\,v\bigr\}_{\,0}\,.\label{eq:sum2}
\end{align}

All these properties mentioned above show that the discretization proposed in the present study belongs to the class of operational finite difference schemes \cite{Samarskii1981}, which became later known as \emph{mimetic} methods \cite{Lipnikov2014}.


\subsection{The main result}

\begin{theorem}\label{thm:1}
The operator $\Aa\ \equiv\ -\Lambda$ is self-adjoint in the space $\H\,$.
\end{theorem}

\begin{proof}
For two arbitrary grid functions $u,\,v\ \in\ \H$ we consider their scalar product on a layer $0\ \leq\ j_{\,2}\ <\ N_2\,$:
\begin{equation*}
  \bigl\langle \Lambda_{\,1\,1}\, u,\,v\bigr\rangle_{\,\H_{\,1},\,j_{\,2}}\ =\ h_{\,1}\,\sum_{j_{\,1}\,=\,1}^{N_1-1}\bigl\{v\,\Lambda_{\,1\,1}\, u\bigr\}_{\,j_{\,1},\,j_{\,2}}\ +\ \bigl\{v\,\Lambda_{\,1\,1}\, u\bigr\}_{\,0,\,j_{\,2}}\;\frac{h_{\,1}}{2}\ +\ \bigl\{v\,\Lambda_{\,1\,1}\, u\bigr\}_{\,N_{\,1},\,j_{\,2}}\;\frac{h_{\,1}}{2}\,.
\end{equation*}
Taking into account definition \eqref{eq:lambda11} of the operator $\Lambda_{\,1\,1}$ and \textsc{Green}'s identity \eqref{eq:green1} we obtain
\begin{multline*}
  \bigl\langle \Lambda_{\,1\,1}\, u,\,v\bigr\rangle_{\,\H_{\,1},\,j_{\,2}}\ =\ -\frac{1}{2}\;\sum_{j_{\,1}\,=\,0}^{N_1\,-\,1}\bigl\{\K_{\,1\,1}\,u_{\,q^{\,1}}^{\,\natural}\,v_{\,q^{\,1}}^{\,\natural}\bigr\}_{\,j_{\,1},\,j_{\,2}}\,h_{\,1}\ +\ \frac{1}{2}\;\bigl(\K_{\,1\,1}\bigr)_{\,N_1-1,\,j_{\,2}}\,\bigl\{u_{\,q^{\,1}}^{\,\flat}\,v\bigr\}_{\,N_1,\,j_{\,2}}\\
  -\ \frac{1}{2}\;\bigl\{\K_{\,1\,1}\,u_{\,q^{\,1}}^{\,\natural}\,v\bigr\}_{\,0,\,j_{\,2}}\ -\ \frac{1}{2}\;\sum_{j_{\,1}\,=\,1}^{N_1}\bigl\{\K_{\,1\,1}\,u_{\,q^{\,1}}^{\,\flat}\,v_{\,q^{\,1}}^{\,\flat}\bigr\}_{\,j_{\,1},\,j_{\,2}}\,h_{\,1}\ +\ \frac{1}{2}\;\bigl\{\K_{\,1\,1}\,u_{\,q^{\,1}}^{\,\flat}\,v\bigr\}_{\,N_1,\,j_{\,2}}\\
  -\ \frac{1}{2}\;\bigl(\K_{\,1\,1}\bigr)_{\,1,\,j_{\,2}}\,\bigl\{u_{\,q^{\,1}}^{\,\natural}\,v\bigr\}_{\,0,\,j_{\,2}}\ +\ \frac{2}{h_{\,1}}\;\bigl(\K_{\,1\,1}\bigr)_{\,1/2,\,j_{\,2}}\,\bigl\{u_{\,q^{\,1}}^{\,\natural}\,v\bigr\}_{\,0,\,j_{\,2}}\;\frac{h_{\,1}}{2}\\ -\ \frac{2}{h_{\,1}}\;\bigl(\K_{\,1\,1}\bigr)_{\,N_1\,-\,1/2,\,j_{\,2}}\,\bigl\{u_{\,q^{\,1}}^{\,\flat}\,v\bigr\}_{\,N_1,\,j_{\,2}}\;\frac{h_{\,1}}{2}\,,
\end{multline*}
or simply
\begin{equation}\label{eq:65}
  \bigl\langle \Lambda_{\,1\,1}\, u,\,v\bigr\rangle_{\,\H_{\,1},\,j_{\,2}}\ =\ -\frac{1}{2}\;\sum_{j_{\,1}\,=\,0}^{N_1\,-\,1}\bigl\{\K_{\,1\,1}\,u_{\,q^{\,1}}^{\,\natural}\,v_{\,q^{\,1}}^{\,\natural}\bigr\}_{\,j_{\,1},\,j_{\,2}}\,h_{\,1}\ -\frac{1}{2}\;\sum_{j_{\,1}\,=\,1}^{N_1}\bigl\{\K_{\,1\,1}\,u_{\,q^{\,1}}^{\,\flat}\,v_{\,q^{\,1}}^{\,\flat}\bigr\}_{\,j_{\,1},\,j_{\,2}}\,h_{\,1}\,.
\end{equation}
By substituting the last expression into formula \eqref{eq:56}, we obtain the equality
\begin{multline*}
  \scalh{\Lambda_{\,1\,1}\, u}{v}\ =\ -\frac{1}{2}\;\sum_{\,j_{\,2}\,=\,1}^{N_2\,-\,1}\biggl\{\sum_{j_{\,1}\,=\,0}^{\,N_1\,-\,1}\K_{\,1\,1}\,u_{\,q^{\,1}}^{\,\natural}\,v_{\,q^{\,1}}^{\,\natural}\ +\ \sum_{j_{\,1}\,=\,1}^{\,N_1}\K_{\,1\,1}\,u_{\,q^{\,1}}^{\,\flat}\,v_{\,q^{\,1}}^{\,\flat}\biggr\}_{\,j_{\,1},\,j_{\,2}}\,h_{\,1}\,h_{\,2} \\
  - \frac{1}{2}\;\sum_{j_{\,1}\,=\,0}^{N_1\,-\,1}\bigl\{\K_{\,1\,1}\,u_{\,q^{\,1}}^{\,\natural}\,v_{\,q^{\,1}}^{\,\natural}\bigr\}_{\,j_{\,1},\,0}\;\frac{h_{\,1}\,h_{\,2}}{2}\ -\frac{1}{2}\;\sum_{j_{\,1}\,=\,1}^{N_1}\bigl\{\K_{\,1\,1}\,u_{\,q^{\,1}}^{\,\flat}\,v_{\,q^{\,1}}^{\,\flat}\bigr\}_{\,j_{\,1},\,j_{\,2}}\;\frac{h_{\,1}\,h_{\,2}}{2}\,.
\end{multline*}
The last identity is absolutely symmetric with respect to the functions $u$ and $v\,$, \ie
\begin{equation*}
  \scalh{\Lambda_{\,1\,1}\, u}{v}\ \equiv\ \scalh{u}{\Lambda_{\,1\,1}\, v}\,.
\end{equation*}
Consequently, the operator $\Lambda_{\,1\,1}$ is self-adjoint in space $\H$, \ie
\begin{equation*}
  \Lambda_{\,1\,1}^{\,\ast}\ \equiv\ \Lambda_{\,1\,1}\,.
\end{equation*}

Using the summation by parts formulas \eqref{eq:sum1}, \eqref{eq:sum2} along with the definition \eqref{eq:lambda12} of the operator $\Lambda_{\,1\,2}$ in interior and boundary nodes, we obtain the following expression for the scalar product
\begin{multline}\label{eq:66}
  \scalh{\Lambda_{\,1\,2}\, u}{v}\ =\ -\frac{1}{4}\;\sum_{j_{\,2}\,=\,1}^{N_2\,-\,1}\biggl[\,\sum_{j_{\,1}\,=\,1}^{N_1}\,\K_{\,1\,2}\,\bigl(u_{\,q^{\,2}}^{\,\natural}\ +\ u_{\,q^{\,2}}^{\,\flat}\bigr)\,v_{\,q^{\,1}}^{\,\flat}\ +\ \sum_{j_{\,1}\,=\,0}^{N_1\,-\,1}\,\K_{\,1\,2}\,\bigl(u_{\,q^{\,2}}^{\,\natural}\ +\ u_{\,q^{\,2}}^{\,\flat}\bigr)\,v_{\,q^{\,1}}^{\,\natural}\,\biggr]_{\,j_{\,1},\,j_{\,2}}\!\!\! h_{\,1}\,h_{\,2}\\ 
  -\ \frac{1}{2}\;\sum_{j_{\,1}\,=\,1}^{N_1}\bigl\{\K_{\,1\,2}\,u_{\,q^{\,2}}^{\,\natural}\,v_{\,q^{\,1}}^{\,\flat}\bigr\}_{\,j_{\,1},\,0}\;\frac{h_{\,1}\,h_{\,2}}{2}\ -\ \frac{1}{2}\;\sum_{j_{\,1}\,=\,0}^{N_1\,-\,1}\bigl\{\K_{\,1\,2}\,u_{\,q^{\,2}}^{\,\natural}\,v_{\,q^{\,1}}^{\,\natural}\bigr\}_{\,j_{\,1},\,0}\;\frac{h_{\,1}\,h_{\,2}}{2}\,.
\end{multline}
It is obvious that
\begin{equation*}
  \scalh{\Lambda_{\,1\,2}\, u}{v}\ \neq\ \scalh{u}{\Lambda_{\,1\,2}\, v}\,.
\end{equation*}
Thus, the operator $\Lambda_{\,1\,2}$ is not self-adjoint. Nevertheless, we continue the proof without any deception. Consider right now the operator $\Lambda_{\,2\,1}$ defined in \eqref{eq:lambda21}. Using similar techniques we obtain the following expression for the scalar product of $\Lambda_{\,2\,1}\,u$ with $v\,$:
\begin{multline}\label{eq:68}
  \scalh{\Lambda_{\,2\,1}\, u}{v}\ =\ -\frac{1}{4}\;\sum_{j_{\,1}\,=\,1}^{N_1\,-\,1}\biggl[\,\sum_{j_{\,2}\,=\,1}^{N_2}\,\K_{\,2\,1}\,\bigl(u_{\,q^{\,1}}^{\,\natural}\ +\ u_{\,q^{\,1}}^{\,\flat}\bigr)\,v_{\,q^{\,2}}^{\,\flat}\ +\ \sum_{j_{\,2}\,=\,0}^{N_2\,-\,1}\,\K_{\,2\,1}\,\bigl(u_{\,q^{\,1}}^{\,\natural}\ +\ u_{\,q^{\,1}}^{\,\flat}\bigr)\,v_{\,q^{\,2}}^{\,\natural}\,\biggr]_{\,j_{\,1},\,j_{\,2}}\!\!\! h_{\,1}\,h_{\,2}\\ 
  -\ \frac{1}{2}\;\sum_{j_{\,2}\,=\,1}^{N_2}\bigl\{\K_{\,2\,1}\,u_{\,q^{\,1}}^{\,\natural}\,v_{\,q^{\,2}}^{\,\flat}\bigr\}_{\,0,\,j_{\,2}}\;\frac{h_{\,1}\,h_{\,2}}{2}\ -\ \frac{1}{2}\;\sum_{j_{\,2}\,=\,0}^{N_2\,-\,1}\bigl\{\K_{\,2\,1}\,u_{\,q^{\,1}}^{\,\natural}\,v_{\,q^{\,2}}^{\,\natural}\bigr\}_{\,0,\,j_{\,2}}\;\frac{h_{\,1}\,h_{\,2}}{2}\\
  -\ \frac{1}{2}\;\sum_{j_{\,2}\,=\,1}^{N_2}\bigl\{\K_{\,2\,1}\,u_{\,q^{\,1}}^{\,\flat}\,v_{\,q^{\,2}}^{\,\flat}\bigr\}_{\,N_1,\,j_{\,2}}\;\frac{h_{\,1}\,h_{\,2}}{2}\ -\ \frac{1}{2}\;\sum_{j_{\,2}\,=\,0}^{N_2\,-\,1}\bigl\{\K_{\,2\,1}\,u_{\,q^{\,1}}^{\,\flat}\,v_{\,q^{\,2}}^{\,\natural}\bigr\}_{\,N_1,\,j_{\,2}}\;\frac{h_{\,1}\,h_{\,2}}{2}\,.
\end{multline}
From the last formula it can be clearly seen that
\begin{equation*}
  \scalh{\Lambda_{\,2\,1}\, u}{v}\ \neq\ \scalh{u}{\Lambda_{\,2\,1}\, v}\,.
\end{equation*}
Thus, the operator $\Lambda_{\,2\,1}$ as well as $\Lambda_{\,1\,2}$ is not self-adjoint. However, we can study the sum $\Lambda_{\,2\,1}\ +\ \Lambda_{\,1\,2}\,$. By regrouping the terms in \eqref{eq:66} and \eqref{eq:68} we obtain
\begin{multline}\label{eq:69}
  \scalh{\Lambda_{\,1\,2}\, u\ +\ \Lambda_{\,2\,1}\, u}{v}\ =\\
  -\frac{h_{\,1}\,h_{\,2}}{4}\;\biggl\{\sum_{j_{\,1}\,=\,1}^{N_1\,-\,1}\sum_{j_{\,2}\,=\,1}^{N_2\,-\,1}\Bigl[\,\K_{\,1\,2}\,\bigl(u_{\,q^{\,1}}^{\,\natural}\ +\ u_{\,q^{\,1}}^{\,\flat}\bigr)\cdot\bigl(v_{\,q^{\,2}}^{\,\natural}\ +\ v_{\,q^{\,2}}^{\,\flat}\bigr)\ +\ \K_{\,1\,2}\,\bigl(v_{\,q^{\,1}}^{\,\natural}\ +\ v_{\,q^{\,1}}^{\,\flat}\bigr)\cdot\bigl(u_{\,q^{\,2}}^{\,\natural}\ +\ u_{\,q^{\,2}}^{\,\flat}\bigr)\,\Bigr]_{j_{\,1},\,j_{\,2}}\\
  +\ \sum_{j_{\,2}\,=\,1}^{N_2\,-\,1}\Bigl[\,\K_{\,1\,2}\,u_{\,q^{\,1}}^{\,\natural}\,\bigl(v_{\,q^{\,2}}^{\,\natural}\ +\ v_{\,q^{\,2}}^{\,\flat}\bigr)\ +\ \K_{\,1\,2}\,v_{\,q^{\,1}}^{\,\natural}\,\bigl(u_{\,q^{\,2}}^{\,\natural}\ +\ u_{\,q^{\,2}}^{\,\flat}\bigr)\,\Bigr]_{0,\,j_{\,2}} \\
  +\ \sum_{j_{\,2}\,=\,1}^{N_2\,-\,1}\Bigl[\,\K_{\,1\,2}\,u_{\,q^{\,1}}^{\,\flat}\,\bigl(v_{\,q^{\,2}}^{\,\natural}\ +\ v_{\,q^{\,2}}^{\,\flat}\bigr)\ +\ \K_{\,1\,2}\,v_{\,q^{\,1}}^{\,\flat}\,\bigl(u_{\,q^{\,2}}^{\,\natural}\ +\ u_{\,q^{\,2}}^{\,\flat}\bigr)\,\Bigr]_{N_1,\,j_{\,2}}\\
  +\ \sum_{j_{\,1}\,=\,1}^{N_1\,-\,1}\Bigl[\,\K_{\,1\,2}\,v_{\,q^{\,2}}^{\,\natural}\,\bigl(u_{\,q^{\,1}}^{\,\natural}\ +\ u_{\,q^{\,1}}^{\,\flat}\bigr)\ +\ \K_{\,1\,2}\,u_{\,q^{\,2}}^{\,\natural}\,\bigl(v_{\,q^{\,1}}^{\,\natural}\ +\ v_{\,q^{\,1}}^{\,\flat}\bigr)\,\Bigr]_{j_{\,1},\,0} \\
  \ +\ \Bigl[\,\K_{\,1\,2}\,u_{\,q^{\,1}}^{\,\natural}\,v_{\,q^{\,2}}^{\,\natural}\ +\ \K_{\,1\,2}\,u_{\,q^{\,2}}^{\,\natural}\,v_{\,q^{\,1}}^{\,\natural}\,\Bigr]_{\,0,\,0}\ +\ \Bigl[\,\K_{\,1\,2}\,u_{\,q^{\,1}}^{\,\flat}\,v_{\,q^{\,2}}^{\,\natural}\ +\ \K_{\,1\,2}\,u_{\,q^{\,2}}^{\,\natural}\,v_{\,q^{\,1}}^{\,\flat}\,\Bigr]_{\,N_1,\,0}\biggr\}\,.
\end{multline}
Above we used also the fact that $\K_{\,1\,2}\ \equiv\ \K_{\,2\,1}\,$. From the last formula \eqref{eq:69} we readily obtain that
\begin{equation*}
  \scalh{\Lambda_{\,1\,2}\, u\ +\ \Lambda_{\,2\,1}\, u}{v}\ \equiv\ \scalh{u}{\Lambda_{\,1\,2}\, v\ +\ \Lambda_{\,2\,1}\, v}\,.
\end{equation*}
Thus, the operator $\Lambda_{\,2\,1}\ +\ \Lambda_{\,1\,2}$ is self-adjoint in the space $\H\,$, \ie
\begin{equation*}
   \bigl(\Lambda_{\,1\,2}\ +\ \Lambda_{\,2\,1}\bigr)^{\,\ast}\ \equiv\ \Lambda_{\,1\,2}\ +\ \Lambda_{\,2\,1}\,,
\end{equation*}
even if constitutive operators $\Lambda_{\,1\,2}$ and $\Lambda_{\,2\,1}$ are not self-adjoint if taken separately.

Finally, consider the operator $\Lambda_{\,2\,2}\,$, which was defined in \eqref{eq:lambda22}. We employ first the discrete \textsc{Green} identities \eqref{eq:green1}, \eqref{eq:green2} to evaluate scalar products on layers $j_{\,1}\ \in\ \bigl\{0,\,1,\,\ldots,\,N_1\bigr\}\,$:
\begin{equation*}
  \bigl\langle \Lambda_{\,2\,2}\, u,\,v\bigr\rangle_{\,\H_{\,2},\,j_{\,1}}\ =\ -\frac{1}{2}\;\sum_{j_{\,2}\,=\,0}^{N_2\,-\,1}\bigl\{\K_{\,2\,2}\,u_{\,q^{\,2}}^{\,\natural}\,v_{\,q^{\,2}}^{\,\natural}\bigr\}_{\,j_{\,1},\,j_{\,2}}\;h_{\,2}\ -\ \frac{1}{2}\;\sum_{j_{\,2}\,=\,1}^{N_2}\bigl\{\K_{\,2\,2}\,u_{\,q^{\,2}}^{\,\flat}\,v_{\,q^{\,2}}^{\,\flat}\bigr\}_{\,j_{\,1},\,j_{\,2}}\;h_{\,2}\,.
\end{equation*}
Using this partial result we can now compute the scalar product in the $\H$ space as well:
\begin{multline}\label{eq:67}
  \scalh{\Lambda_{\,2\,2}\,u}{v}\ =\ -\frac{1}{2}\;\sum_{j_{\,1}\,=\,1}^{N_1\,-\,1}\biggl[\,\sum_{j_{\,2}\,=\,0}^{N_2\,-\,1}\,\K_{\,2\,2}\,u_{\,q^{\,2}}^{\,\natural}\,v_{\,q^{\,2}}^{\,\natural}\ +\ \sum_{j_{\,2}\,=\,1}^{N_2}\,\K_{\,2\,2}\,u_{\,q^{\,2}}^{\,\flat}\,v_{\,q^{\,2}}^{\,\flat}\,\biggr]_{\,j_{\,1},\,j_{\,2}}\,h_{\,1}\,h_{\,2}\\
  -\ \frac{1}{2}\;\sum_{j_{\,2}\,=\,0}^{N_2\,-\,1}\bigl\{\K_{\,2\,2}\,u_{\,q^{\,2}}^{\,\natural}\,v_{\,q^{\,2}}^{\,\natural}\bigr\}_{\,0,\,j_{\,2}}\;\frac{h_{\,1}\,h_{\,2}}{2}\ -\ \frac{1}{2}\;\sum_{j_{\,2}\,=\,1}^{N_2}\bigl\{\K_{\,2\,2}\,u_{\,q^{\,2}}^{\,\flat}\,v_{\,q^{\,2}}^{\,\flat}\bigr\}_{\,0,\,j_{\,2}}\;\frac{h_{\,1}\,h_{\,2}}{2} \\
  -\ \frac{1}{2}\;\sum_{j_{\,2}\,=\,0}^{N_2\,-\,1}\bigl\{\K_{\,2\,2}\,u_{\,q^{\,2}}^{\,\natural}\,v_{\,q^{\,2}}^{\,\natural}\bigr\}_{\,N_1,\,j_{\,2}}\;\frac{h_{\,1}\,h_{\,2}}{2}\ -\ \frac{1}{2}\;\sum_{j_{\,2}\,=\,1}^{N_2}\bigl\{\K_{\,2\,2}\,u_{\,q^{\,2}}^{\,\flat}\,v_{\,q^{\,2}}^{\,\flat}\bigr\}_{\,N_1,\,j_{\,2}}\;\frac{h_{\,1}\,h_{\,2}}{2}\,.
\end{multline}
By changing symmetrically grid functions $u$ and $v\,$, we obtain
\begin{equation*}
  \scalh{\Lambda_{\,2\,2}\, u}{v}\ \equiv\ \scalh{u}{\Lambda_{\,2\,2}\, v}\,.
\end{equation*}
Consequently, the operator $\Lambda_{\,2\,2}$ is self-adjoint in the space $\H$, \ie
\begin{equation*}
  \Lambda_{\,2\,2}^{\,\ast}\ \equiv\ \Lambda_{\,2\,2}\,.
\end{equation*}
Since $\Lambda\ \equiv\ \Lambda_{\,1\,1}\ +\ \bigl(\Lambda_{\,1\,2}\ +\ \Lambda_{\,2\,1}\bigr)\ +\ \Lambda_{\,2\,2}$ and each term is an self-adjoint operator, the operator $\Lambda$ is self-adjoint as well. Hence, the same holds for $\Aa\ \equiv\ -\Lambda\,$.
\end{proof}

Below we are going to show also that the operator $\Aa$ is positive definite. The proof will be based on the uniform ellipticity property of equation \eqref{eq:laplaceQ}, which was established in Lemma~\ref{lem:1}. Another key ingredient consists in the properties of the finite difference approximation to the following mixed Boundary Value Problem (BVP) for the \textsc{Laplace} equation posed on $\Q_{\,0}\,$:
\begin{equation}\label{eq:70}
  \Delta\,\phi\ =\ 0\,, \qquad \q\ =\ \bigl(q^{\,1},\,q^{\,2}\bigr)\ \in\ \Q_{\,0}\,,
\end{equation}
with boundary conditions
\begin{align}\label{eq:71a}
  \pd{\phi}{q^{\,1}}\Bigr\vert_{q^{\,1}\,=\,1}\ &=\ 0\,, \qquad\quad\ \ 
  \pd{\phi}{q^{\,1}}\Bigr\vert_{q^{\,1}\,=\,0}\ =\ \mu_{\,\ell}\,\bigl(q^{\,2}\bigr)\,, \\
  \phi\bigr\vert_{q^{\,2}\,=\,1}\ &=\ \tilde{\mu}\,\bigl(q^{\,1}\bigr)\,, \qquad
  \pd{\phi}{q^{\,2}}\Bigr\vert_{q^{\,2}\,=\,0}\ =\ 0\,.\label{eq:71b}
\end{align}
This problem is a particular case of the BVP considered earlier for elliptic equation \eqref{eq:laplaceQ} with the following coefficients:
\begin{equation*}
  \K_{\,1\,1}\ =\ \K_{\,2\,2}\ \equiv\ 1\,, \qquad
  \K_{\,1\,2}\ =\ \K_{\,2\,1}\ \equiv\ 0\,.
\end{equation*}
In this case the finite difference operator $\Lambda$ admits a much simpler form since finite difference operators corresponding to mixed derivatives vanish, \ie
\begin{equation*}
  \Lambda_{\,1\,2}\ =\ \Lambda_{\,2\,1}\ \equiv\ 0\,.
\end{equation*}
The finite difference scheme for the BVP \eqref{eq:70} -- \eqref{eq:71b} can be written as
\begin{equation*}
  \Delta_{\,h}\,\phi\ =\ \nuv\,, \qquad \Delta_{\,h}\ \eqdef\ \Delta_{\,h}^{\,(1)}\ +\ \Delta_{\,h}^{\,(2)}\,,
\end{equation*}
with operators $\Delta_{\,h}^{\,(1,\,2)}$ defined as
\begin{equation*}
  \Delta_{\,h}^{\,(1)}\,\phi_{\,\j}\ \eqdef\ \begin{dcases}
  \ \bigl\{\phi_{\,q^{\,1}\,q^{\,1}}^{\,\flat,\ \natural}\bigr\}_{\,j_{\,1},\,j_{\,2}}\,, & \qquad 0\ <\ j_{\,1}\ <\ N_{1}\ -\ 1\,, \qquad 0\ \leq\ j_{\,2}\ <\ N_{2}\,, \\
  \ \frac{2}{h_{\,1}}\;\bigl\{\phi_{\,q^{\,1}}^{\,\natural}\bigr\}_{\,0,\,j_{\,2}}\,, & \qquad j_{\,1}\ =\ 0\,, \qquad 0\ \leq\ j_{\,2}\ <\ N_{2}\,, \\
  \ -\,\frac{2}{h_{\,1}}\;\bigl\{\phi_{\,q^{\,1}}^{\,\flat}\bigr\}_{\,N_1,\,j_{\,2}}\,, & \qquad j_{\,1}\ =\ N_1\,, \qquad 0\ \leq\ j_{\,2}\ <\ N_{2}\,, \\
  \ 0\,, & \qquad 0\ \leq\ j_{\,1}\ \leq\ N_1\,, \qquad j_{\,2}\ =\ N_2\,,
  \end{dcases}
\end{equation*}
\begin{equation*}
  \Delta_{\,h}^{\,(2)}\,\phi_{\,\j}\ \eqdef\ \begin{dcases}
    \ \bigl\{\phi_{\,q^{\,2}\,q^{\,2}}^{\,\flat,\ \natural}\bigr\}_{\,j_{\,1},\,j_{\,2}}\,, & \qquad 0\ \leq\ j_{\,1}\ \leq\ N_{1}\,, \qquad 0\ <\ j_{\,2}\ <\ N_{2}\,, \\
    \ \frac{2}{h_{\,2}}\;\bigl\{\phi_{\,q^{\,2}}^{\,\natural}\bigr\}_{\,j_{\,1},\,0}\,, & \qquad 0\ \leq\ j_{\,1}\ \leq\ N_1\,, \qquad j_{\,2}\ =\ 0\,, \\
    \ 0\,, & \qquad 0\ \leq\ j_{\,1}\ \leq\ N_1\,, \qquad j_{\,2}\ =\ N_2\,.
  \end{dcases}
\end{equation*}
Below we introduce also the operators $\Ba_{\,1}\ \eqdef\ -\Delta_{\,h}^{\,(1)}\,$, $\Ba_{\,2}\ \eqdef\ -\Delta_{\,h}^{\,(2)}$ and $\Ba\ \eqdef\ \Ba_{\,1}\ +\ \Ba_{\,2}\,$, \ie $\Ba\ \eqdef\ -\,\Delta_{\,h}\,$. By applying Theorem~\ref{thm:1} we conclude that the operator $\Ba\,:\ \H\ \mapsto\ \H$ is self-adjoint in the space $\H\,$.

It is not difficult to see that the family of grid functions
\begin{equation}\label{eq:eigf}
  \psi^{\,(k,\,l)}_{\,j_{\,1},\,j_{\,2}}\ \eqdef\ \alpha_{\,k}\cdot\cos\bigl(\pi\,k\,q_{\,j_{\,1}}^{\,1}\bigr)\cdot\cos\bigl(\pi\,(l\,+\,\half)\,q_{\,j_{\,2}}^{\,2}\bigr)\,, \quad k\ =\ 0,\,1,\,\ldots,\,N_{1}\,, \quad l\ =\ 0,\,1,\,\ldots\,N_2\,-\,1
\end{equation}
satisfies the following identities:
\begin{equation}\label{eq:eig}
  \Ba\cdot\psi^{\,(k,\,l)}\ \equiv\ \lambda^{\,(k,\,l)}\cdot\psi^{\,(k,\,l)}\,.
\end{equation}
Henceforth, the functions $\{\psi^{\,(k,\,l)}\}_{\,k,\,l}$ are eigenfunctions\footnote{Please, notice that each eigenfunction $\psi^{\,(k,\,l)}$ is a grid function taking in each node $\j\ =\ \bigl(j_{\,1},\,j_{\,2}\bigr)$ the value $\psi^{\,(k,\,l)}_{\,j_{\,1},\,j_{\,2}}$ given in equation \eqref{eq:eigf}.} of the operator $\Ba$ and numbers $\{\lambda^{\,(k,\,l)}\}_{\,k,\,l}\ \in\ \R$ are its eigenvalues:
\begin{equation*}
  \lambda^{\,(k,\,l)}\ \eqdef\ \frac{4}{h_{\,1}^{\,2}}\;\sin^{\,2}\Bigl(\frac{\pi\,k\,h_{\,1}}{2}\Bigr)\ +\ \frac{4}{h_{\,2}^{\,2}}\;\sin^{\,2}\biggl[\,\frac{(l\ +\ \frac{1}{2})\,\pi\,h_{\,2}}{2}\,\biggr]\,.
\end{equation*}
If we choose the normalization factor $\alpha_{\,k}$ as
\begin{equation*}
  \alpha_{\,k}\ =\ \begin{dcases}
  \ \sqrt{2}\,, & \qquad k\ =\ 0\ \wedge\ N_1\,, \\
  \ 2\,, & k\ =\ 1,\,2,\,\ldots,\,N_1\,-\,1\,,
  \end{dcases}
\end{equation*}
then the eigenfunctions $\{\psi^{\,(k,\,l)}\}_{\,k,\,l}$ of the operator $\Ba$ form an orthonormal basis in $\H\,$, \ie
\begin{equation*}
  \scalh{\psi^{\,(k,\,l)}}{\psi^{\,(r,\,s)}}\ =\ \delta^{\,k\,r}\cdot\delta^{\,l\,s}\ =\ \begin{dcases}
  \ 1\,, & \qquad k\ =\ r\ \wedge\ l\ =\ s\,, \\
  \ 0\,, & \qquad k\ \neq\ r\ \vee\ l\ \neq\ s\,.
  \end{dcases}
\end{equation*}
In other words, an arbitrary function $u\ \in\ \H$ can be represented as a sum of a finite \textsc{Fourier} series:
\begin{equation}\label{eq:repr}
  u\ =\ \sum_{k\,=\,0}^{N_1}\,\sum_{l\,=\,0}^{N_2\,-\,1} c_{\,k\,l}\,\psi^{\,(k,\,l)}\,,
\end{equation}
where the numbers $\bigl\{c_{\,k\,l}\ \equiv\ \scalh{u}{\psi^{\,(k,\,l)}}\bigr\}_{\,k,\,l}$ are called the \textsc{Fourier} coefficients. The \textsc{Parseval} theorem holds in the finite dimensional setting as well:
\begin{equation}\label{eq:Pars}
  \norm{u}^{\,2}_{\,\H}\ \equiv\ \scalh{u}{u}\ =\ \sum_{k\,=\,0}^{N_1}\,\sum_{l\,=\,0}^{N_2\,-\,1} c_{\,k\,l}^{\,2}\,.
\end{equation}
Since all eigenvalues of the operator $\Ba$ are positive, the following Lemma holds:
\begin{lemma}\label{lem:2}
The finite difference operator $\Ba$ is positive definite in space $\H$ and for any function $u\ \in\ \H$ we have the following estimations
\begin{equation}\label{eq:est}
  \lambda_{\,\min}\,\norm{u}^{\,2}_{\,\H}\ \leq\ \scalh{\Ba\,u}{u}\ \leq\ \lambda_{\,\max}\,\norm{u}^{\,2}_{\,\H}\,,
\end{equation}
with
\begin{equation}\label{eq:lambdas}
  \lambda_{\,\min}\ \eqdef\ \frac{4}{h_{\,2}^{\,2}}\;\sin^{\,2}\Bigl(\frac{\pi\,h_{\,2}}{4}\Bigr)\,, \qquad
  \lambda_{\,\max}\ \eqdef\ \frac{4}{h_{\,1}^{\,2}}\ +\ \frac{4}{h_{\,2}^{\,2}}\;\cos^{\,2}\Bigl(\frac{\pi\,h_{\,2}}{4}\Bigr)\,.
\end{equation}
\end{lemma}

\begin{proof}
Let us take an arbitrary function $u\ \in\ \H$ and we expand it \eqref{eq:repr} on the basis of orthonormal eigenfunctions of the operator $\Ba\,$. Then we have
\begin{equation*}
  \scalh{\Ba\,u}{u}\ =\ \biggl\langle\, \sum_{k\,=\,0}^{N_1}\,\sum_{l\,=\,0}^{N_2\,-\,1} c_{\,k\,l}\,\lambda^{\,(k,\,l)}\,\psi^{\,(k,\,l)}\,,\ \sum_{r\,=\,0}^{N_1}\,\sum_{s\,=\,0}^{N_2\,-\,1} c_{\,r\,s}\,\psi^{\,(r,\,s)}\,\biggr\rangle_{\,\H}\,,
\end{equation*}
where we used the fact that $\psi^{\,(k,\,l)}$ is an eigenfunction of $\Ba$ \eqref{eq:eig}. Then, by taking into account the orthonormality property we can greatly simplify the scalar product:
\begin{equation*}
  \scalh{\Ba\,u}{u}\ =\ \sum_{k\,=\,0}^{N_1}\,\sum_{l\,=\,0}^{N_2\,-\,1} c_{\,k\,l}^{\,2}\,\lambda^{\,(k,\,l)}\,.
\end{equation*}
Obviously, among a finite number of eigenvalues $\{\lambda^{\,(k,\,l)}\}_{\,k,\,l}$ we can always choose the minimal and maximal ones. Thus,
\begin{equation*}
  \forall\, (k,\,l)\,:\ \quad \lambda_{\,\min}\ \leq\ \lambda^{\,(k,\,l)}\ \leq\ \lambda_{\,\max}\,.
\end{equation*}
Taking into account the last estimation along with the \textsc{Parseval} identity \eqref{eq:Pars}, we obtain immediately the estimations \eqref{eq:est}.
\end{proof}

Thanks to the property of self-adjointness and positive definiteness of the operator $\Ba\,$, we can define another scalar product $\scalb{u}{v}\ \eqdef\ \scalh{\Ba\,u}{v}\,$, which implies the corresponding energy norm:
\begin{equation*}
  \norm{u}_{\,\Ba}\ \eqdef\ \sqrt{\scalb{u}{u}}\,,
\end{equation*}
generated by the operator $\Ba\,$. Below we use an equivalent representation of the energy norm, which is more suitable for our purposes:
\begin{equation*}
  \norm{u}_{\,\Ba}^{\,2}\ \equiv\ \scalh{\Ba_{\,1}\,u}{u}\ +\ \scalh{\Ba_{\,2}\,u}{u}\,.
\end{equation*}

\begin{lemma}\label{lem:3}
For any function $u\ \in\ \H$ the following identities hold:
\begin{equation}\label{eq:85}
  \scalh{\Ba_{\,1}\,u}{u}\ =\ \sum_{j_{\,1}\,=\,1}^{N_1}\,\biggl[\,\sum_{j_{\,2}\,=\,1}^{N_2\,-\,1}\Bigl\{\bigl(u_{\,q^{\,1}}^{\,\flat}\bigr)^{\,2}\Bigr\}_{\,j_{\,1},\,j_{\,2}}\ +\ \frac{1}{2}\;\Bigl\{\bigl(u_{\,q^{\,1}}^{\,\flat}\bigr)^{\,2}\Bigr\}_{\,j_{\,1},\,0}\,\biggr]\,h_{\,1}\,h_{\,2}\,,
\end{equation}
\begin{equation}\label{eq:86}
  \scalh{\Ba_{\,2}\,u}{u}\ =\ \sum_{j_{\,2}\,=\,1}^{N_2}\,\biggl[\,\sum_{j_{\,1}\,=\,1}^{N_1\,-\,1}\Bigl\{\bigl(u_{\,q^{\,2}}^{\,\flat}\bigr)^{\,2}\Bigr\}_{\,j_{\,1},\,j_{\,2}}\ +\ \frac{1}{2}\;\Bigl\{\bigl(u_{\,q^{\,2}}^{\,\flat}\bigr)^{\,2}\Bigr\}_{\,0,\,j_{\,2}}\ +\ \frac{1}{2}\;\Bigl\{\bigl(u_{\,q^{\,2}}^{\,\flat}\bigr)^{\,2}\Bigr\}_{\,N_1,\,j_{\,2}}\,\biggr]\,h_{\,1}\,h_{\,2}\,.
\end{equation}
\end{lemma}

\begin{proof}
Taking into account equation \eqref{eq:57} and discrete identities shown in Section~\ref{sec:discr}, we obtain the equality
\begin{equation*}
  \bigl\langle \Ba_{\,1}\,u,\,u\bigr\rangle_{\,\H_{\,1},\,j_{\,2}}\ =\ \sum_{j_{\,1}\,=\,1}^{N_1}\,\Bigl\{\bigl(u_{\,q^{\,1}}^{\,\flat}\bigr)^{\,2}\Bigr\}_{\,j_{\,1},\,j_{\,2}}\,h_{\,1}\,.
\end{equation*}
Then, from the last result and relation \eqref{eq:56} follows immediately the requested identity \eqref{eq:85}. Similarly, the discrete identities along with definition \eqref{eq:59} yield
\begin{equation*}
  \bigl\langle \Ba_{\,2}\,u,\,u\bigr\rangle_{\,\H_{\,2},\,j_{\,1}}\ =\ \sum_{j_{\,2}\,=\,1}^{N_2}\,\Bigl\{\bigl(u_{\,q^{\,2}}^{\,\flat}\bigr)^{\,2}\Bigr\}_{\,j_{\,1},\,j_{\,2}}\,h_{\,2}\,.
\end{equation*}
Then, using relation \eqref{eq:58} we obtain the second requested identity \eqref{eq:86}.
\end{proof}

Now we can come back to the original finite difference operator $\Aa$ and prove the following
\begin{theorem}\label{thm:2}
The operator $\Aa\ \equiv\ -\,\Lambda$ is positive definite in the space $\H\,$, and we have the following estimations:
\begin{equation}\label{eq:87}
  c_{\,1}\,\lambda_{\,\min}\,\norm{u}^{\,2}_{\,\H}\ \leq\ \scalh{\Aa\,u}{u}\ \leq\ c_{\,2}\,\lambda_{\,\max}\,\norm{u}^{\,2}_{\,\H}\,,
\end{equation}
where the constants $c_{\,1,\,2}$ were defined in \eqref{eq:c1}, \eqref{eq:c2} and $\lambda_{\,\min,\,\max}$ --- in \eqref{eq:lambdas}.
\end{theorem}

\begin{proof}
As in the first step of the proof we take $v\ \equiv\ u\ \in\ \H$ and substitute it into equations \eqref{eq:65}. After regrouping interior and boundary nodes, we obtain:
\begin{multline*}
  \scalh{\Lambda_{\,1\,1}\,u}{u}\ =\ -\,\frac{h_{\,1}\,h_{\,2}}{4}\;\biggl\{\,\sum_{j_{\,1}\,=\,1}^{N_1\,-\,1}\sum_{j_{\,2}\,=\,1}^{N_2\,-\,1}\Bigl[\,2\,\K_{\,1\,1}\,\bigl(u_{\,q^{\,1}}^{\,\natural}\bigr)^{\,2}\ +\ 2\,\K_{\,1\,1}\,\bigl(u_{\,q^{\,1}}^{\,\flat}\bigr)^{\,2}\,\Bigr]_{j_{\,1},\,j_{\,2}}\ +\\
  \sum_{j_{\,2}\,=\,1}^{N_2\,-\,1}\Bigl[\,2\,\K_{\,1\,1}\,\bigl(u_{\,q^{\,1}}^{\,\natural}\bigr)^{\,2}\,\Bigr]_{0,\,j_{\,2}}\ +\ \sum_{j_{\,2}\,=\,1}^{N_2\,-\,1}\Bigl[\,2\,\K_{\,1\,1}\,\bigl(u_{\,q^{\,1}}^{\,\flat}\bigr)^{\,2}\,\Bigr]_{N_1,\,j_{\,2}}\ +\\
  \sum_{j_{\,1}\,=\,1}^{N_1\,-\,1}\Bigl[\,\K_{\,1\,1}\,\bigl(u_{\,q^{\,1}}^{\,\natural}\bigr)^{\,2}\ +\ \K_{\,1\,1}\,\bigl(u_{\,q^{\,1}}^{\,\flat}\bigr)^{\,2}\,\Bigr]_{\,j_{\,1},\,0}\ +\ \Bigl[\,\K_{\,1\,1}\,\bigl(u_{\,q^{\,1}}^{\,\natural}\bigr)^{\,2}\,\Bigr]_{\,0,\,0}\ +\ \Bigl[\,\K_{\,1\,1}\,\bigl(u_{\,q^{\,1}}^{\,\flat}\bigr)^{\,2}\,\Bigr]_{\,N_1,\,0}\,\biggr\}\,.
\end{multline*}
We perform the same operation with equation \eqref{eq:67} as well:
\begin{multline*}
  \scalh{\Lambda_{\,2\,2}\,u}{u}\ =\ -\,\frac{h_{\,1}\,h_{\,2}}{4}\;\biggl\{\,\sum_{j_{\,1}\,=\,1}^{N_1\,-\,1}\sum_{j_{\,2}\,=\,1}^{N_2\,-\,1}\Bigl[\,2\,\K_{\,2\,2}\,\bigl(u_{\,q^{\,2}}^{\,\natural}\bigr)^{\,2}\ +\ 2\,\K_{\,2\,2}\,\bigl(u_{\,q^{\,2}}^{\,\flat}\bigr)^{\,2}\,\Bigr]_{j_{\,1},\,j_{\,2}}\ +\\
  \sum_{j_{\,2}\,=\,1}^{N_2\,-\,1}\Bigl[\,\K_{\,2\,2}\,\bigl(u_{\,q^{\,2}}^{\,\natural}\bigr)^{\,2}\ +\ \K_{\,2\,2}\,\bigl(u_{\,q^{\,2}}^{\,\flat}\bigr)^{\,2}\,\Bigr]_{\,0,\,j_{\,2}}\ +\ \sum_{j_{\,2}\,=\,1}^{N_2\,-\,1}\Bigl[\,\K_{\,2\,2}\,\bigl(u_{\,q^{\,2}}^{\,\natural}\bigr)^{\,2}\ +\ \K_{\,2\,2}\,\bigl(u_{\,q^{\,2}}^{\,\flat}\bigr)^{\,2}\,\Bigr]_{\,N_1,\,j_{\,2}}\ +\\
  \sum_{j_{\,1}\,=\,1}^{N_1\,-\,1}\Bigl[\,2\,\K_{\,2\,2}\,\bigl(u_{\,q^{\,2}}^{\,\natural}\bigr)^{\,2}\,\Bigr]_{\,j_{\,1},\,0}\ +\ \Bigl[\,\K_{\,2\,2}\,\bigl(u_{\,q^{\,2}}^{\,\natural}\bigr)^{\,2}\,\Bigr]_{\,0,\,0}\ +\ \Bigl[\,\K_{\,2\,2}\,\bigl(u_{\,q^{\,2}}^{\,\natural}\bigr)^{\,2}\,\Bigr]_{\,N_1,\,0}\ + \\
  \sum_{j_{\,1}\,=\,1}^{N_1\,-\,1}\Bigl[\,2\,\K_{\,2\,2}\,\bigl(u_{\,q^{\,2}}^{\,\flat}\bigr)^{\,2}\,\Bigr]_{\,j_{\,1},\,N_2}\ +\ \Bigl[\,\K_{\,2\,2}\,\bigl(u_{\,q^{\,2}}^{\,\flat}\bigr)^{\,2}\,\Bigr]_{\,0,\,N_2}\ +\ \Bigl[\,\K_{\,2\,2}\,\bigl(u_{\,q^{\,2}}^{\,\flat}\bigr)^{\,2}\,\Bigr]_{\,N_1,\,N_2}\,.
\end{multline*}
Similarly, we set $v\ \equiv\ u$ in equation \eqref{eq:69} and using the last two formulas we have
\begin{multline*}
  \scalh{\Aa\,u}{u}\ =\\ \frac{h_{\,1}\,h_{\,2}}{4}\;\biggl\{\,\sum_{j_{\,1}\,=\,1}^{N_1\,-\,1}\sum_{j_{\,2}\,=\,1}^{N_2\,-\,1}\Bigl[\,\Qf\,\bigl(u_{\,q^{\,1}}^{\,\natural},\,u_{\,q^{\,2}}^{\,\natural}\bigr)\ +\ \Qf\,\bigl(u_{\,q^{\,1}}^{\,\natural},\,u_{\,q^{\,2}}^{\,\flat}\bigr)\ +\ \Qf\,\bigl(u_{\,q^{\,1}}^{\,\flat},\,u_{\,q^{\,2}}^{\,\natural}\bigr)\ +\ \Qf\,\bigl(u_{\,q^{\,1}}^{\,\flat},\,u_{\,q^{\,2}}^{\,\flat}\bigr)\,\Bigr]_{\,j_{\,1},\,j_{\,2}}\\
  +\ \sum_{j_{\,2}\,=\,1}^{N_2\,-\,1}\Bigl[\,\Qf\,\bigl(u_{\,q^{\,1}}^{\,\natural},\,u_{\,q^{\,2}}^{\,\natural}\bigr)\ +\ \Qf\,\bigl(u_{\,q^{\,1}}^{\,\natural},\,u_{\,q^{\,2}}^{\,\flat}\bigr)\,\Bigr]_{\,0,\,j_{\,2}}\ +\ \sum_{j_{\,2}\,=\,1}^{N_2\,-\,1}\Bigl[\,\Qf\,\bigl(u_{\,q^{\,1}}^{\,\flat},\,u_{\,q^{\,2}}^{\,\natural}\bigr)\ +\ \Qf\,\bigl(u_{\,q^{\,1}}^{\,\flat},\,u_{\,q^{\,2}}^{\,\flat}\bigr)\,\Bigr]_{\,N_1,\,j_{\,2}}\\
  +\ \sum_{j_{\,1}\,=\,1}^{N_1\,-\,1}\Bigl[\,\Qf\,\bigl(u_{\,q^{\,1}}^{\,\natural},\,u_{\,q^{\,2}}^{\,\natural}\bigr)\ +\ \Qf\,\bigl(u_{\,q^{\,1}}^{\,\flat},\,u_{\,q^{\,2}}^{\,\natural}\bigr)\,\Bigr]_{\,j_{\,1},\,0}\ +\ \Bigl[\,\Qf\,\bigl(u_{\,q^{\,1}}^{\,\natural},\,u_{\,q^{\,2}}^{\,\natural}\bigr)\,\Bigr]_{\,0,\,0}\ +\ \Bigl[\,\Qf\,\bigl(u_{\,q^{\,1}}^{\,\flat},\,u_{\,q^{\,2}}^{\,\natural}\bigr)\,\Bigr]_{\,N_1,\,0}\\
  +\ \sum_{j_{\,1}\,=\,1}^{N_1\,-\,1}\Bigl[\,2\,\Qf\,\bigl(0,\,u_{\,q^{\,2}}^{\,\flat}\bigr)\,\Bigr]_{\,j_{\,1},\,N_2}\ +\ \Bigl[\,\Qf\,\bigl(0,\,u_{\,q^{\,2}}^{\,\flat}\bigr)\,\Bigr]_{\,0,\,N_2}\ +\ \Bigl[\,\Qf\,\bigl(0,\,u_{\,q^{\,2}}^{\,\flat}\bigr)\,\Bigr]_{\,N_1,\,N_2}\,\biggr\}\,.
\end{multline*}
where we used the quadratic form $\Qf(\cdot,\,\cdot)$ defined earlier in \eqref{eq:quadratic}. By taking into account the uniform ellipticity property \eqref{eq:34}, we can estimate the quantity $\scalh{\Aa\,u}{u}$ from below:
\begin{multline}\label{eq:90}
  \scalh{\Aa\,u}{u}\ \geq\ c_{\,1}\;\frac{h_{\,1}\,h_{\,2}}{4}\cdot\biggl\{\,\sum_{j_{\,1}\,=\,1}^{N_1\,-\,1}\sum_{j_{\,2}\,=\,1}^{N_2\,-\,1}\,2\,\Bigl[\,\bigl(u_{\,q^{\,1}}^{\,\natural}\bigr)^{\,2}\ +\ \bigl(u_{\,q^{\,1}}^{\,\flat}\bigr)^{\,2}\ +\ \bigl(u_{\,q^{\,2}}^{\,\natural}\bigr)^{\,2}\ +\ \bigl(u_{\,q^{\,2}}^{\,\flat}\bigr)^{\,2}\,\Bigr]_{\,j_{\,1},\,j_{\,2}}\ +\\
  \sum_{j_{\,2}\,=\,1}^{N_2\,-\,1}\Bigl[\,2\,\bigl(u_{\,q^{\,1}}^{\,\natural}\bigr)^{\,2}\ +\ \bigl(u_{\,q^{\,2}}^{\,\natural}\bigr)^{\,2}\ +\ \bigl(u_{\,q^{\,2}}^{\,\flat}\bigr)^{\,2}\,\Bigr]_{\,0,\,j_{\,2}}\ +\ \sum_{j_{\,2}\,=\,1}^{N_2\,-\,1}\Bigl[\,2\,\bigl(u_{\,q^{\,1}}^{\,\flat}\bigr)^{\,2}\ +\ \bigl(u_{\,q^{\,2}}^{\,\natural}\bigr)^{\,2}\ +\ \bigl(u_{\,q^{\,2}}^{\,\flat}\bigr)^{\,2}\,\Bigr]_{\,N_1,\,j_{\,2}}\ + \\
  \sum_{j_{\,1}\,=\,1}^{N_1\,-\,1}\Bigl[\,\bigl(u_{\,q^{\,1}}^{\,\natural}\bigr)^{\,2}\ +\ \bigl(u_{\,q^{\,1}}^{\,\flat}\bigr)^{\,2}\ +\ 2\,\bigl(u_{\,q^{\,2}}^{\,\natural}\bigr)^{\,2}\,\Bigr]_{\,j_{\,1},\,0}\ +\ \Bigl[\,\bigl(u_{\,q^{\,1}}^{\,\natural}\bigr)^{\,2}\ +\ \bigl(u_{\,q^{\,2}}^{\,\natural}\bigr)^{\,2}\,\Bigr]_{\,0,\,0}\ +\ \Bigl[\,\bigl(u_{\,q^{\,1}}^{\,\flat}\bigr)^{\,2}\ +\ \bigl(u_{\,q^{\,2}}^{\,\natural}\bigr)^{\,2}\,\Bigr]_{\,N_1,\,0}\\
  +\ 2\,\sum_{j_{\,1}\,=\,1}^{N_1\,-\,1}\bigl(u_{\,q^{\,2}}^{\,\flat}\bigr)^{\,2}_{\,j_{\,1},\,N_2}\ +\ \bigl(u_{\,q^{\,2}}^{\,\flat}\bigr)^{\,2}_{\,0,\,N_2}\ +\ \bigl(u_{\,q^{\,2}}^{\,\flat}\bigr)^{\,2}_{\,N_{1},\,N_2}\,\biggr\}\ \defeq \Bf_{\,\mathrm{low}}\,.
\end{multline}
Then, in the expression above for $\Bf_{\,\mathrm{low}}$ we replace\footnote{We can do it since the right finite difference can be considered as the left one for the subsequent value of the index. Thus, this operation can be seen as a change of index by one in all summations.} all right derivatives $u_{\,q^{\,\alpha}}^{\,\natural}$ by their left counterparts $u_{\,q^{\,\alpha}}^{\,\flat}\,$:
\begin{multline*}
  \Bf_{\,\mathrm{low}}\ =\ c_{\,1}\,\biggl\{\,\sum_{j_{\,1}\,=\,1}^{N_1}\,\Bigl[\,\sum_{j_{\,2}\,=\,1}^{N_2\,-\,1}\,\bigl(u_{\,q^{\,1}}^{\,\flat}\bigr)^{\,2}_{\,j_{\,1},\,j_{\,2}}\ +\ \frac{1}{2}\;\bigl(u_{\,q^{\,1}}^{\,\flat}\bigr)^{\,2}_{\,j_{\,1},\,0}\,\Bigr]\ +\\
  \sum_{j_{\,2}\,=\,1}^{N_2}\,\Bigl[\,\sum_{j_{\,1}\,=\,1}^{N_1\,-\,1}\,\bigl(u_{\,q^{\,2}}^{\,\flat}\bigr)^{\,2}_{\,j_{\,1},\,j_{\,2}}\ +\ \frac{1}{2}\;\bigl(u_{\,q^{\,2}}^{\,\flat}\bigr)^{\,2}_{\,0,\,j_{\,2}}\ +\ \frac{1}{2}\;\bigl(u_{\,q^{\,2}}^{\,\flat}\bigr)^{\,2}_{\,N_1,\,j_{\,2}}\,\Bigr]\,\biggr\}\,h_{\,1}\,h_{\,2}\,.
\end{multline*}
Now we can make use of Lemma~\ref{lem:3}, where we showed identities \eqref{eq:85} and \eqref{eq:86}, which yield:
\begin{equation*}
  \Bf_{\,\mathrm{low}}\ \equiv\ c_{\,1}\,\Bigl\{\scalh{\Ba_{\,1}\,u}{u}\ +\ \scalh{\Ba_{\,2}\,u}{u}\Bigr\}\,.
\end{equation*}
In other words, we just showed the following estimation from below:
\begin{equation*}
  \scalh{\Aa\,u}{u}\ \geq\ c_{\,1}\,\scalh{\Ba\,u}{u}\,.
\end{equation*}
Finally, by applying Lemma~\ref{lem:2} we obtain that
\begin{equation*}
  \scalh{\Ba\,u}{u}\ \geq\ \lambda_{\,\min}\,\norm{u}^{\,2}_{\,\H}\,.
\end{equation*}
The two last inequalities correspond precisely to the required lower bound in \eqref{eq:87}. By departing again from equation \eqref{eq:90} and using the same techniques we can show the upper bound in \eqref{eq:87}. This completes the proof of this Theorem.
\end{proof}

\begin{remark}
During the proof of Theorem~\ref{thm:2} we established also the following estimates:
\begin{equation*}
  c_{\,1}\,\scalh{\Ba\,u}{u}\ \leq\ \scalh{\Aa\,u}{u}\ \leq\ c_{\,2}\,\scalh{\Ba\,u}{u}\,, \qquad \forall u\ \in\ \H\,,
\end{equation*}
or equivalently
\begin{equation*}
  c_{\,1}\,\Ba\ \leq\ \Aa\ \leq\ c_{\,2}\,\Ba\,.
\end{equation*}
Such operators $\Aa$ and $\Ba$ satisfying the inequalities above are usually referred to as \emph{energetically equivalent} \cite{Samarskii2001}.
\end{remark}

\subsubsection{Final remarks}

From the positive definiteness property of the operator $\Aa$ follows the unique solvability property of the finite difference problem \eqref{eq:46}. To solve this problem numerically one can use almost any iterative technique. It is natural to use the conjugate gradient method whose convergence is guaranteed if the operator is positive definite and self-adjoint \cite{Polyak1969}. The speed of convergence depends essentially on the conditioning number $\kappa$ of the difference operator $\Aa\,$, which can be estimated as
\begin{equation*}
  \kappa\ \leq\ \frac{c_{\,2}}{c_{\,1}}\cdot\frac{\lambda_{\,\max}}{\lambda_{\,\min}}\ =\ \frac{c_{\,2}}{c_{\,1}}\cdot\underbrace{\frac{\dfrac{4}{h_{\,1}^{\,2}}\ +\ \dfrac{4}{h_{\,2}^{\,2}}\;\cos^{\,2}\Bigl(\dfrac{\pi\,h_{\,2}}{4}\Bigr)}{\dfrac{4}{h_{\,2}^{\,2}}\;\sin^{\,2}\Bigl(\dfrac{\pi\,h_{\,2}}{4}\Bigr)}}_{\displaystyle{(\star)}}\,.
\end{equation*}
For larger values of $\kappa$ the convergence will be slower. As we showed earlier in Section~\ref{sec:prop}, the first factor $\dfrac{c_{\,2}}{c_{\,1}}$ is completely determined by the properties of the geometric mapping \eqref{eq:map}. This ratio becomes large for the meshes with highly elongated or highly distorted elements. The second factor $(\star)$ depends on the number of nodes $N_{1,\,2}$ chosen in the directions $O\,q^{1,\,2}$ correspondingly and $(\star)$ becomes larger when we refine the grid (\ie when $N_{1,\,2}$ increase).


\section{Numerical algorithm}
\label{sec:alg}

In this Section we describe briefly how the numerical code is organized and we provide the details on the treatment of kinematic and dynamic boundary conditions on the free surface and on the lateral fixed and moving walls.

Assume that the position of the wall $x\ =\ s^{\,n}$ on the time layer $t\ =\ t^{\,n}$ is known along with all other fields. The grid $\Omega_{\,h}^{\,n}$ with nodes $\x_{\,\j}^{\,n}\,$, ordered using a multi-index $\j\ =\ \bigl(j_{\,1},\,j_{\,2}\bigr)$ in the physical space is constructed. On this grid we know the values of grid functions $\bigl\{\xi_{\,j_{\,1}}^{\,n}\bigr\}_{\,j_{\,1}}\,$, $\bigl\{\eta_{\,j_{\,1}}^{\,n}\bigr\}_{\,j_{\,1}}$ and $\bigl\{\phi_{\,\j}^{\,n}\bigr\}_{\,\j}\,$. The time marching requires to compute these quantities on the following time layer $t\ =\ t^{\,n+1}\,$. It is done in several stages:
\begin{enumerate}
  \item First, we compute the values of the velocity potential $\bigl\{\phi_{\,j_{\,1},\,N_2}^{\,n+1}\bigr\}_{\,j_{\,1}}$ at the free surface node $\q_{\,j_{\,1},\,N_2}\,$.
  \item These values are used as the \textsc{Dirichlet} data in the linear problem \eqref{eq:46} to determine the velocity potential $\bigl\{\phi_{\,j_{\,1},\,j_{\,2}}^{\,n+1}\bigr\}_{\,j_{\,1},\,j_{\,2}}$ in all other nodes $\q_{\,\j}$ with $0\ \leq\ j_{\,1}\ \leq\ N_1$ and $0\ \leq\ j_{\,2}\ <\ N_2\,$.
  \item Then, we find new positions $\bigl\{\xi_{\,j_{\,1}}^{\,n+1}\bigr\}_{\,j_{\,1}}\,$, $\bigl\{\eta_{\,j_{\,1}}^{\,n+1}\bigr\}_{\,j_{\,1}}$ of the free surface by using the finite difference approximation of the kinematic boundary conditions \eqref{eq:kinfsQ} and \eqref{eq:kinfsQxi}.
  \item The new position $x\ =\ s^{\,n+1}$ of the moving wall is found from the discretized version of the Equation~\eqref{eq:sping}.
  \item Finally, the new grid $\bar{\Omega}_{\,h}^{\,n+1}$ is constructed on the following time layer $t\ =\ t^{\,n+1}\,$.
\end{enumerate}
Every step above involves grid functions on various layers and to achieve a better understanding of all the stages of our algorithm, below we provide a detailed description of every item.


\subsection{Approximation of the dynamic boundary condition (1)}

The finite difference approximation to the dynamic boundary condition \eqref{eq:dynfsQ} takes the following form:
\begin{equation}\label{eq:94}
  \frac{\phi_{\,\j}^{\,n+1}\ -\ \phi_{\,\j}^{\,n}}{\tau_{\,n}}\ -\ \bigl(u_{\,\j}^{\,n}\cdot x_{\,t,\,\j}\ +\ v_{\,\j}^{\,n}\cdot y_{\,t,\,\j}\bigr)\ +\ \half\,\abs{\u_{\,\j}^{\,n}}^{\,2}\ +\ \eta_{\,j_{\,1}}^{\,n}\ =\ 0\,,
\end{equation}
where $\j\ =\ \bigl(j_{\,1},\,j_{\,2}\bigr)$ is a multi-index with $j_{\,1}\ =\ 1,\,2,\,\ldots,\,N_1\,-\,1$ and $j_{\,2}\ =\ N_2\,$. In order to compute the speeds of grid nodes we employ finite differences applied to two last grids $\bar{\Omega}_{\,h}^{\,n}$ and $\bar{\Omega}_{\,h}^{\,n-1}\,$:
\begin{equation}\label{eq:95}
  \x_{\,t,\,\j}\ =\ \frac{\x_{\,\j}^{\,n}\ -\ \x_{\,\j}^{\,n-1}}{\tau_{\,n-1}}\,,
\end{equation}
where $\tau_{\,n-1}\ \equiv\ t^{\,n}\ -\ t^{\,n-1}$ is the local time step, $\abs{\u_{\,\j}^{\,n}}^{\,2}\ \equiv\ \bigl(u_{\,\j}^{\,n}\bigr)^{\,2}\ +\ \bigl(v_{\,\j}^{\,n}\bigr)^{\,2}\,$. In formula \eqref{eq:95} we use only the first order approximation in time. This is done for the sake of algorithm memory efficiency. Otherwise, to achieve at least the second order accuracy in $\tau\,$, we would have to keep in memory the grids for three time layers. In the current implementation we keep track of the preceding grid only.

The \textsc{Cartesian} components $u_{\,\j}^{\,n}$ and $v_{\,\j}^{\,n}$ of the velocity vector $\u_{\,\j}^{\,n}$ are computed after approximating formulas \eqref{eq:cartU} and \eqref{eq:cartV}:
\begin{equation}\label{eq:96}
  u_{\,\j}^{\,n}\ =\ \biggl[\,\frac{\phi_{q^{\,1}}\cdot y_{q^{\,2}}\ -\ \phi_{q^{\,2}}\cdot y_{q^{\,1}}}{\J}\,\biggr]_{\,\j}^{\,n}\,, \qquad
  v_{\,\j}^{\,n}\ =\ \biggl[\,\frac{-\phi_{q^{\,1}}\cdot x_{q^{\,2}}\ +\ \phi_{q^{\,2}}\cdot x_{q^{\,1}}}{\J}\,\biggr]_{\,\j}^{\,n}\,.
\end{equation}
Above one has to approximate also the partial derivatives $\phi_{\,q^{\,1}}\,$, $\x_{\,q^{\,1}}$ over the independent variable $q^{\,1}$ along the upper side of the square $\Q^{\,0}$. These derivatives are computed with standard central finite differences up to the second order accuracy. The derivatives $\phi_{\,q^{\,2}}\,$, $\x_{\,q^{\,2}}$ are computed with one-sided finite differences of the second order as well (in order to have the uniform second order accuracy in space).

We use equation \eqref{eq:94} in order to compute the updated values of the velocity potential only in interior nodes of the grid $\gamma_{\,s}^{\,h}\,$. For two nodes located in the upper corners $\q_{\,0,\,N_2}\ =\ (0,\,1)$ and $\q_{\,N_1,\,N_2}\ =\ (1,\,1)$ we use modified approximation formulas which take into account the lateral walls $\gamma_{\,\ell,\,r}$ impermeability. In the computational domain the left wall has the fixed coordinate $q^{\,1}\ =\ 0\,$. The intersection point of the free surface with the left boundary is a \emph{triple point} and it has the coordinate $\x_{\,0,\,N_2}\,$. This point is permanently moving up and down the left wall.

In the coordinate system $O\,q^{\,1}\,q^{\,2}$ the \textsc{Cartesian} components $\u\ =\ \odd{\x}{t}$ become the so-called contravariant components of the velocity $\v\ =\ \odd{\q}{t}$ with
\begin{equation}\label{eq:97}
  v^{\,\alpha}\ =\ \od{q^{\,\alpha}}{t}\ =\ \pd{q^{\,\alpha}}{t}\ +\ u\;\pd{q^{\,\alpha}}{x}\ +\ v\;\pd{q^{\,\alpha}}{y}\,, \qquad \alpha\ =\ 1,\,2\,.
\end{equation}
Consequently, if a fluid particle moves such that its coordinate $q^{\,\alpha}$ does not change, then the corresponding contravariant velocity component $v^{\,\alpha}$ has to vanish. For instance, the fluid particles, which constitute the free surface, have the coordinate $q^{\,2}\ =\ 1\,$, consequently
\begin{equation*}
  v^{\,2}\bigr\vert_{\,q^{\,2}\,=\,1}\ =\ 0\,.
\end{equation*}
Similarly, for fluid particles sliding along the left vertical wall we always have $q^{\,1}\ =\ 0\,$, which yields
\begin{equation*}
  v^{\,1}\bigr\vert_{\,q^{\,1}\,=\,0}\ =\ 0\,.
\end{equation*}
The point $\x_{\,0,\,N_2}\ \in\ \gamma_{\,s}^{\,h}\ \cap\ \gamma_{\,\ell}^{\,h}\,$, thus
\begin{equation}\label{eq:100}
  v^{\,1}\,\bigl(\q_{\,0,\,N_2}\bigr)\ =\ v^{\,2}\,\bigl(\q_{\,0,\,N_2}\bigr)\ \equiv\ 0\,.
\end{equation}
From formula \eqref{eq:97} and using also equations \eqref{eq:25}, \eqref{eq:27} we obtain the following expressions for contravariant components of the velocity:
\begin{equation*}
  v^{\,1}\ =\ \frac{\bigl(u\ -\ x_{\,t}\bigr)\cdot y_{\,q^{\,2}}\ -\ \bigl(v\ -\ y_{\,t}\bigr)\cdot x_{\,q^{\,2}}}{\J}\,, \qquad
  v^{\,2}\ =\ \frac{-\bigl(u\ -\ x_{\,t}\bigr)\cdot y_{\,q^{\,1}}\ +\ \bigl(v\ -\ y_{\,t}\bigr)\cdot x_{\,q^{\,1}}}{\J}\,.
\end{equation*}
From these expressions and using boundary condition\footnote{We implicitly use also the fact that the transformation \eqref{eq:map} is non-degenerate.} \eqref{eq:100} for the \emph{triple point} $\x_{\,0,\,N_2}$ we obtain that
\begin{equation*}
  x_{\,t}\bigr\vert_{\,\q_{\,0,\,N_2}}\ =\ u\,, \qquad
  y_{\,t}\bigr\vert_{\,\q_{\,0,\,N_2}}\ =\ v\,.
\end{equation*}
Henceforth, $\abs{\u}^{\,2}\ \equiv\ u^{\,2}\ +\ v^{\,2}\ =\ u\,x_{\,t}\ +\ v\,y_{\,t}$ and the free surface dynamic boundary condition \eqref{eq:dynfsQ} in the corner $\q_{\,0,\,N_2}$ becomes:
\begin{equation*}
  \phi_{\,t}\ -\ \half\;\bigl(u\cdot x_{\,t}\ +\ v\cdot y_{\,t}\bigr)\ +\ \eta\ =\ 0\,, \qquad \x\ =\ \x_{\,0,\,N_2}\,.
\end{equation*}
The difference equation \eqref{eq:94} correspondingly takes the form:
\begin{equation}\label{eq:103}
  \frac{\phi_{\,\j}^{\,n+1}\ -\ \phi_{\,\j}^{\,n}}{\tau_{n}}\ -\ \half\;\bigl(u_{\,\j}^{\,n}\cdot x_{\,t,\,\j}\ +\ v_{\,\j}^{\,n}\cdot y_{\,t,\,\j}\bigr)\ +\ \eta_{\,j_{\,1}}^{\,n}\ =\ 0\,, \qquad \j\ =\ \bigl(j_{\,1},\,j_{\,2}\bigr)\ \equiv\ (0,\,N_2)\,.
\end{equation}
In a similar way one can derive the finite difference discretization of the boundary condition in the right triple point. We give only the final result. It coincides with equation \eqref{eq:103} provided that we make a substitution $j_{\,1}\ =\ N_1\,$.


\subsection{Computation of the velocity potential (2)}

On the second stage of the numerical algorithm we solve the finite difference problem \eqref{eq:46} in order to find the values of the velocity potential $\bigl\{\phi_{\,\j}^{\,n+1}\bigr\}_{\,\j}$ in the nodes $\q_{\,\j}\ \in\ \bar{\Q}_{\,h}^{\,0}\ \setminus\ \gamma_{\,s}^{\,h}$. In other words, in the nodes $\j\ =\ \bigl(j_{\,1},\,j_{\,2})$ with $0\ \leq\ j_{\,1}\ \leq\ N_1$ and $0\ \leq\ j_{\,2}\ <\ N_2\,$.

In order to invert the linear system \eqref{eq:46} we tested several iterative methods and no method clearly outperformed the others. In the final implementation of the code we used the Successive Over Relaxation (SOR) method \cite{Young1950} that we remind briefly here. The choice of the SOR method can be explained by two main reasons:
\begin{enumerate}
  \item We would like to have a sufficiently simple to implement and sufficiently efficient method to find approximate solutions to relatively large linear systems
  \item The method does not have to rely heavily on the matrix structure (such as \textsc{Thomas}'s algorithm, for example). For instance, if we include obstacles in the fluid domain, its topology might change, which implies some drastic changes in the pattern of non-zero elements of the matrix.
\end{enumerate}
Thus, to our opinion SOR method represents the best trade-off among the efficiency, generality and ease of implementation.

We use a nine-point stencil in the finite difference scheme. Thus, equation $\j$ has the following general form:
\begin{equation}\label{eq:104}
  \sum_{k\,=\,0}^{8} \A_{\,k,\,\j}\,\phi_{\,k,\,\j}\ =\ \B_{\,\j}\,, \qquad
  \q_{\,\j}\ \in\ \bar{\Q}_{\,h}^{\,0}\ \setminus\ \gamma_{\,s}^{\,h}\,.
\end{equation}
The last condensed form is obtained from \eqref{eq:46} by regrouping all the terms in front of unknowns $\bigl\{\phi_{\,\j}^{\,n+1}\bigr\}_{\,\j}\,$. The corresponding coefficient is denoted by $\A_{\,k,\,\j}\,$. In equation \eqref{eq:104} we used a local numeration of indices depicted in Figure~\ref{fig:stencil}(\textit{a}). The right hand side is given by
\begin{equation}\label{eq:105}
  \B_{\,\j}\ \eqdef\ \begin{dcases}
    \ \mu_{\,0,\,j_{\,2}}\,h_{\,2}\,, & \qquad \q_{\,\j}\ \in\ \gamma_{\,\ell}^{\,h}\,, \\
    \ \half\;\mu_{\,0,\,0}\,h_{\,2}\,, & \qquad \q_{\,\j}\ \equiv\ \q_{\,0,\,0}\,, \\
    \ 0\,, & \qquad \mathrm{otherwise}\,.
  \end{dcases}
\end{equation}
Then, one step of the SOR method reads:
\begin{equation*}
  \tilde{\phi}_{\,0,\,\j}\ \eqdef\ -\,\frac{1}{\A_{\,0,\,\j}}\;\biggl\{\,\sum_{k\,\in\,\{1,\,2,\,5,\,6\}}\A_{\,k,\,\j}\,\phi_{\,k,\,\j}^{\,(m+1)}\ +\ \sum_{k\,\in\,\{3,\,4,\,7,\,8\}}\A_{\,k,\,\j}\,\phi_{\,k,\,\j}^{\,(m)}\ -\ \B_{\,0}\,\biggr\}\,,
\end{equation*}
\begin{equation*}
  \phi_{\,0,\,\j}^{\,(m+1)}\ \eqdef\ \theta\,\tilde{\phi}_{\,0,\,\j}\ +\ (1\ -\ \theta)\,\phi_{\,0,\,\j}^{\,(m)}\,,
\end{equation*}
where $(m)$ is the iteration number and $\theta$ is the relaxation parameter. The value of this parameter weakly influences the iterative process speed of convergence. The question of finding the \emph{optimal} value of the relaxation parameter $\theta$ is the key to the efficiency of the SOR method. This question can be studied theoretically for the \textsc{Dirichlet} problem of the \textsc{Poisson} equation in a square domain $[\,0,\,\ell\,]^{\,2}$ with uniform grids. We know also that the optimal value $\theta^{\,\star}$ always belongs to the interval:
\begin{equation*}
  \theta^{\,\star}\ \in\ (1,\,2)\,.
\end{equation*}
For the \textsc{Poisson} equation on a square domain with the uniform discretization $h_{\,1}\ =\ h_{\,2}\ \equiv\ h\,$, one can show
\begin{equation*}
  \theta^{\,\star}\ =\ \frac{2}{1\ +\ \sin\Bigl(\dfrac{\pi\,h}{\ell}\Bigr)}\,.
\end{equation*}
In the last formula for $\ell\ =\ 1$ and $h\ =\ \frac{1}{30}$ we obtain the value $\theta^{\,\star}\ \approx\ 1.81\,$. In our discrete problem for the velocity potential $\phi_{\,\j}$ we cannot determine theoretically the optimal value $\theta^{\,\star}\,$. This value was determined \emph{experimentally} for each problem under consideration. In numerical computations below we always took $\theta^{\,\star}\ \in\ [\,1.85,\,1.95\,]\,$, depending on the discretization parameters $h_{\,1,\,2}$ as well.

In boundary nodes the stencil contains six points and in corner points only four. However, it does not change anything for the SOR scheme since extra coefficients $\A_{\,k,\,\j}$ can be set to zero for our convenience. The expressions of coefficients $\A_{\,k,\,\j}$ are given in Table~\ref{tab:aj}. In this Table we use the following notation:
\begin{equation*}
  \rhot_{\,1}\ \eqdef\ \frac{h_{\,2}}{h_{\,1}}\;\K_{\,1\,1}\,\bigr\vert_{\,\mathrm{W}}\,, \qquad
  \rhot_{\,2}\ \eqdef\ \frac{h_{\,1}}{h_{\,2}}\;\K_{\,2\,2}\,\bigr\vert_{\,\mathrm{S}}\,, \qquad
  \rhot_{\,3}\ \eqdef\ \frac{h_{\,2}}{h_{\,1}}\;\K_{\,1\,1}\,\bigr\vert_{\,\mathrm{E}}\,, \qquad
  \rhot_{\,4}\ \eqdef\ \frac{h_{\,1}}{h_{\,2}}\;\K_{\,2\,2}\,\bigr\vert_{\,\mathrm{N}}\,,
\end{equation*}
\begin{equation*}
  \rhot_{\,5}\ \eqdef\ \frac{\K_{\,1\,2}\,\bigr\vert_{\,1}\ +\ \K_{\,1\,2}\,\bigr\vert_{\,2}}{4}\,, \qquad
  \rhot_{\,6}\ \eqdef\ -\,\frac{\K_{\,1\,2}\,\bigr\vert_{\,2}\ +\ \K_{\,1\,2}\,\bigr\vert_{\,3}}{4}\,,
\end{equation*}
\begin{equation*}
  \rhot_{\,7}\ \eqdef\ \frac{\K_{\,1\,2}\,\bigr\vert_{\,3}\ +\ \K_{\,1\,2}\,\bigr\vert_{\,4}}{4}\,, \qquad
  \rhot_{\,8}\ \eqdef\ -\,\frac{\K_{\,1\,2}\,\bigr\vert_{\,4}\ +\ \K_{\,1\,2}\,\bigr\vert_{\,1}}{4}\,,
\end{equation*}
\begin{equation*}
  \sigma_{\,\alpha}\ \eqdef\ \frac{\K_{\,1\,2}\,\bigr\vert_{\,0}\ -\ \K_{\,1\,2}\,\bigr\vert_{\,\alpha}}{4}\,, \quad \alpha\ =\ 1,\,2\,, \qquad
  \sigma_{\,\beta}\ \eqdef\ \frac{\K_{\,1\,2}\,\bigr\vert_{\,\beta}\ -\ \K_{\,1\,2}\,\bigr\vert_{\,0}}{4}\,, \quad \beta\ =\ 3,\,4\,.
\end{equation*}
In the definition of quantities $\rhot_{\,k}$ we used the notation introduced in equations \eqref{eq:54a} -- \eqref{eq:54d}. Finally, the coefficient $\A_{\,0,\,\j}$ is determined as
\begin{equation*}
  \A_{\,0,\,\j}\ \eqdef\ -\,\sum_{k\,=\,1}^{4}\,\A_{\,k,\,\j}\,.
\end{equation*}
This concludes the description of our velocity potential solver.

\begin{table}
  \centering
  \begin{tabular}{l|c|c|c|c|c|c|c|c}
  \hline\hline
  \textit{Node type} & $\A_{\,1,\,\j}$ & $\A_{\,2,\,\j}$ & $\A_{\,3,\,\j}$ & $\A_{\,4,\,\j}$ & $\A_{\,5,\,\j}$ & $\A_{\,6,\,\j}$ & $\A_{\,7,\,\j}$ & $\A_{\,8,\,\j}$ \\
  \hline\hline
  $\q_{\,\j}\ \in\ \Q^{\,h}_{\,0}$ & $\rhot_{\,1}$ & $\rhot_{\,2}$ & $\rhot_{\,3}$ & $\rhot_{\,4}$ & $\rhot_{\,5}$ & $\rhot_{\,6}$ & $\rhot_{\,7}$ & $\rhot_{\,8}$ \\
  $\q_{\,\j}\ \in\ \gamma^{\,h}_{\,\ell}$ & $0$ & $\half\,\rho_{\,2}\ -\ \sigma_{\,2}$ & $\rho_{\,3}$ & $\half\,\rho_{\,4}\ -\ \sigma_{\,4}$ & $0$ & $\rho_{\,6}$ & $\rho_{\,7}$ & $0$ \\
  $\q_{\,\j}\ \in\ \gamma^{\,h}_{\,b}$ & $\half\,\rho_{\,1}\ -\ \sigma_{\,1}$ & $0$ & $\half\,\rho_{\,3}\ -\ \sigma_{\,3}$ & $\rho_{\,4}$ & $0$ & $0$ & $\rho_{\,7}$ & $\rho_{\,8}$ \\
  $\q_{\,\j}\ \in\ \gamma^{\,h}_{\,r}$ & $\rho_{\,1}$ & $\half\,\rho_{\,2}\ +\ \sigma_{\,2}$ & $0$ & $\half\,\rho_{\,4}\ +\ \sigma_{\,4}$ & $\rho_{\,5}$ & $0$ & $0$ & $\rho_{\,8}$ \\
  $\q_{\,0,\,0}$ & $0$ & $0$ & $\half\,\rho_{\,3}\ -\ \sigma_{\,3}$ & $\half\,\rho_{\,4}\ -\ \sigma_{\,4}$ & $0$ & $0$ & $\rho_{\,7}$ & $0$ \\
  $\q_{\,N_1,\,0}$ & $\half\,\rho_{\,1}\ -\ \sigma_{\,1}$ & $0$ & $0$ & $\half\,\rho_{\,4}\ +\ \sigma_{\,4}$ & $0$ & $0$ & $0$ & $\rho_{\,8}$ \smallskip \\
  \hline\hline
  \end{tabular}
  \bigskip
  \caption{\small\em Matrix elements of the difference equation \eqref{eq:104}.}
  \label{tab:aj}
\end{table}


\subsection{Free surface motion (3)}

The free surface position on the following time layer $t\ =\ t^{\,n+1}$ is found by integrating in time kinematic conditions \eqref{eq:kin1}, \eqref{eq:kin2}. Here we consider the simplest approximation of these equations using an explicit upwind scheme:
\begin{equation}\label{eq:106}
  \frac{\etab_{\,j_{\,1}}^{\,n+1}\ -\ \etab_{\,j_{\,1}}^{\,n}}{\tau_{\,n}}\ +\ \frac{v_{\,\j}^{\,1,\,n}\ +\ \abs{v_{\,\j}^{\,1,\,n}}}{2}\;\etab_{\,q^{\,1},\,j_{\,1}}^{\,\flat,\,n}\ +\ \frac{v_{\,\j}^{\,1,\,n}\ -\ \abs{v_{\,\j}^{\,1,\,n}}}{2}\;\etab_{\,q^{\,1},\,j_{\,1}}^{\,\natural,\,n}\ -\ \u_{\,\j}^{\,n}\ =\ 0\,, \quad \j\ =\ \bigl(j_{\,1},\,N_2\bigr)\,,
\end{equation}
where we introduced the vector $\etab\ \eqdef\ \bigl(\eta,\,\xi\bigr)$ and index $j_{\,1}\ =\ 1,\,2,\,\ldots,\,N_1\,-\,1\,$. The contravariant velocity component $v_{\,\j}^{\,1}$ is approximated as
\begin{equation}\label{eq:107}
  v_{\,\j}^{\,1,\,n}\ =\ \frac{\bigl(u_{\,\j}^{\,n}\ -\ x_{\,t,\,\j}^{\,n}\bigr)\cdot y_{\,q^{\,2},\,\j}^{\,\flat,\,n}\ -\ \bigl(v_{\,\j}^{\,n}\ -\ y_{\,t,\,\j}^{\,n}\bigr)\cdot x_{\,q^{\,2},\,\j}^{\,\flat,\,n}}{\J_{\,\j}^{\,n}}\,.
\end{equation}
The components of the (\textsc{Cartesian}) velocity vector $\u_{\,\j}^{\,n}$ are computed using formulas \eqref{eq:96} except the fact that we use the values of the velocity potential $\bigl\{\phi_{\,\j}^{\,n+1}\bigr\}_{\,\j}$ on the subsequent time layer $t\ =\ t^{\,n+1}\,$.

We employ formula \eqref{eq:106} in order to compute the free surface position in `interior' nodes $\q_{\,j_{\,1},\,N_1}$ with $j_{\,1}\ =\ 1,\,2,\,\ldots,\,N_1\,-\,1\,$, which lie on the upper side of the square $\Q_{\,h}^{\,0}\,$. Thanks to impermeability conditions \eqref{eq:100}, kinematic conditions \eqref{eq:kin1}, \eqref{eq:kin2} take a very simple form:
\begin{equation*}
  \etab_{\,t}\ -\ \u\ =\ 0\,, \qquad \q\ \in\ \bigl\{\q_{\,0,\,N_2},\,\q_{\,N_1,\,N_2}\bigr\}\,.
\end{equation*}
Taking into account this information, the finite difference formula \eqref{eq:106} takes a much simpler form as well:
\begin{equation*}
  \frac{\etab_{\,j_{\,1}}^{\,n+1}\ -\ \etab_{\,j_{\,1}}^{\,n}}{\tau_{\,n}}\ =\ \u_{\,j_{\,1},\,N_2}^{\,n}\,, \qquad j_{\,1}\ \in\ \bigl\{0,\,N_1\bigr\}\,.
\end{equation*}
The fully discrete scheme to compute the free surface transport is described.


\subsection{Moving wall displacement (4)}

The movable wall position is given by the function $x\ =\ s\,(t)\,$, which is a solution to equation \eqref{eq:piston}. The fluid pressure $p$ is determined from the \textsc{Cauchy}--\textsc{Lagrange} integral given in dimensionless variables in equation \eqref{eq:pressure}. The force acting on the vertical wall is given in dimensionless variables in equation \eqref{eq:force}. In curvilinear coordinates the pressure is given in equation \eqref{eq:pressureQ} and it can be used to compute the force in curvilinear coordinates as well \eqref{eq:forceQ}. The advantage of this formulation is that the wall has a fixed position $q^{\,1}\ =\ 0$ in the transformed domain. The integral in equation \eqref{eq:forceQ} is discretized using the trapezoidal quadrature formula:
\begin{equation*}
  \F\,(t)\ \approx\ \sum_{j_{\,2}\,=\,0}^{N_{\,2}\,-\,1}\;\frac{p_{\,0,\,j_{\,2}}\ +\ p_{\,0,\,j_{\,2}\,+\,1}}{2}\;h_{\,2}\,.
\end{equation*}
In order to compute the time derivative of the velocity potential $\phi_{\,t}$ appearing in the pressure term $\bigl\{p_{\,0,\,j_{\,2}}\bigr\}_{j_{\,2}\,=\,0}^{\,N_{\,2}}\,$, we use the values of the velocity potential on three time layers $\bigl\{\phi_{\,\j}^{\,n-1}\bigr\}_{\,\j}\,$, $\bigl\{\phi_{\,\j}^{\,n}\bigr\}_{\,\j}$ and $\bigl\{\phi_{\,\j}^{\,n+1}\bigr\}_{\,\j}\,$. So, the time derivative of the velocity potential is approximated as:
\begin{equation*}
  \bigl(\phi_{\,t}\bigr)_{0,\,j_{\,2}}\ \approx\ \frac{3\,\phi_{\,0,\,j_{\,2}}^{\,n+1}\ -\ 4\,\phi_{\,0,\,j_{\,2}}^{\,n}\ +\ \phi_{\,0,\,j_{\,2}}^{\,n-1}}{2\,\Delta t}\,.
\end{equation*}
However, due to the change of variables in \eqref{eq:forceQ} one has to compute also the velocities of grid nodes. This is done as specified in \eqref{eq:95}. \textsc{Cartesian} velocities computed as specified in \eqref{eq:96} with the only difference --- we use the values of velocity potential at the new time level $\bigl\{\phi_{\,\j}^{\,n+1}\bigr\}_{\,\j}\,$.

In order to integrate numerically the nonlinear differential equation~\eqref{eq:spring}, we rewrite this equation as a system of two first order differential equations:
\begin{align*}
  \dot{s}\ &=\ \upsilon\,, \\
  m\,\dot{\upsilon}\ +\ k\,s\ &=\ -\,\bigl[\,\F\,(t)\ -\ \F\,(0)\,\bigr]\,,
\end{align*}
together with initial conditions:
\begin{equation*}
  s\,(0)\ =\ 0\,, \qquad \upsilon\,(0)\ =\ 0\,.
\end{equation*}
The last system of equations can be rewritten in the vectorial form for our convenience:
\begin{equation}\label{eq:cauchy}
  \dot{\upupsilon}\ =\ \Upxi\scal\upupsilon\ +\ \Upsigma\,(t)\,, \qquad
  \upupsilon\,(0)\ =\ \vO\,,
\end{equation}
where we introduced the following notations:
\begin{equation*}
  \upupsilon\,(t)\ \eqdef\ \begin{pmatrix}
    s\,(t) \\
    \upsilon\,(t)
  \end{pmatrix}\,, \quad
  \Upxi\ \eqdef\ \begin{pmatrix}
    0 & 1 \\
    -\,\frac{k}{m} & 0
  \end{pmatrix}\,, \quad
  \Upsigma\,(t)\ \eqdef\ \begin{pmatrix}
    0 \\
    -\,\frac{1}{m}\;\bigl[\,\F\,(t)\ -\ \F\,(0)\,\bigr]
  \end{pmatrix}\,.
\end{equation*}
The \textsc{Cauchy} problem \eqref{eq:cauchy} is solved using the so-called modified \textsc{Euler} method\footnote{This scheme is also called the second order explicit \textsc{Runge}--\textsc{Kutta} scheme \cite{Butcher2016}.}:
\begin{align*}
  \frac{\upupsilon^{\,\star}\ -\ \upupsilon^{\,n}}{\Delta t}\ &=\ \Upxi\scal\upupsilon^{\,n}\ +\ \Upsigma^{\,n}\,, \\
  \frac{\upupsilon^{\,n+1}\ -\ \upupsilon^{\,n}}{\Delta t}\ &=\ \frac{1}{2}\;\bigl[\,\Upxi\scal\upupsilon^{\,n}\ +\ \Upxi\scal\upupsilon^{\,\star}\,\bigr]\ +\ \Upsigma^{\,n}\,.
\end{align*}
If we exclude the intermediate variable $\upupsilon^{\,\star}$ from last equations and rewrite the obtained equations in the component-wise form, we obtain the following scheme:
\begin{align*}
  \frac{s^{\,n+1}\ -\ s^{\,n}}{\Delta t}\ &=\ \upsilon^{\,n}\ -\ \frac{\Delta t}{2}\;\frac{k}{m}\;s^{\,n}\ +\ \frac{\Delta}{2\,m}\;\bigl[\,\F\,(t)\ -\ \F\,(0)\,\bigr]\,, \\
  \frac{\upsilon^{\,n+1}\ -\ \upsilon^{\,n}}{\Delta t}\ &=\ -\,\frac{k}{m}\;s^{\,n}\ +\ \frac{1}{m}\;\bigl[\,\F\,(t)\ -\ \F\,(0)\,\bigr]\ -\ \frac{\Delta t}{2}\;\frac{k}{m}\;\upsilon^{\,n}\,.
\end{align*}
The last scheme is implemented in our numerical code.


\subsection{Elliptic mapping construction (5)}
\label{sec:ell}

In previous sections we assumed that mapping \eqref{eq:map} was known. In order to complete our numerical method description we have to describe how to construct this mapping in practice. Namely, we shall use the so-called \emph{equidistribution} method \cite{Khakimzyanov2001}. A typical grid generated using this method is shown in Figure~\ref{fig:grid}. Recently this method was successfully applied in 1D to the simulation of conservation laws \cite{Khakimzyanov2015a}. Below we provide a detailed description of this method to two spatial dimensions.

\begin{figure}
  \centering
  \includegraphics[width=0.79\textwidth]{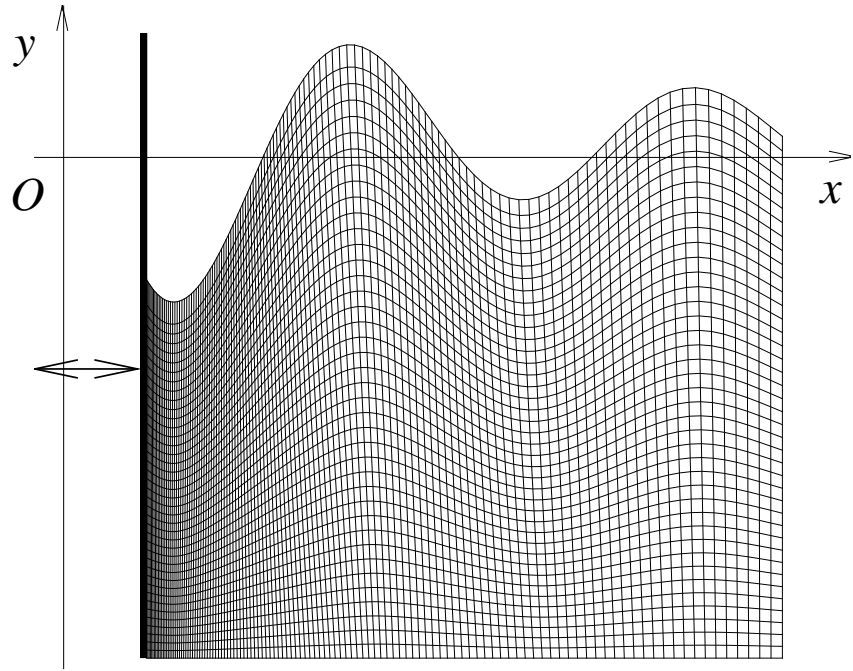}
  \caption{\small\em A typical grid generated using the equidistribution method. On this picture we show a zoom on the area near the moving wall. We underline that the free surface profile is chosen to illustrate the adaptive features of the grid and it does not come from any physical simulation.}
  \label{fig:grid}
\end{figure}

According to the equidistribution method, the coordinates of grid nodes $\bigl\{\x_{\,\j}^{\,n+1}\bigr\}_{\,\j}\ \subseteq\ \Omega\,(t)$ in the physical space on the subsequent time layer $t\ =\ t^{\,n+1}$ are determined by solving the following vectorial parabolic equation:
\begin{equation}\label{eq:movingmesh}
  \vsigma\;\pd{\x}{t}\ =\ \pd{}{q^{\,1}}\;\biggl[\,g_{\,2\,2}\,\varpi\,\pd{\x}{q^{\,1}}\,\biggr]\ +\ \pd{}{q^{\,2}}\;\biggl[\,g_{\,1\,1}\,\varpi\,\pd{\x}{q^{\,2}}\,\biggr]\,,
\end{equation}
where $\vsigma\ >\ 0$ is a positive smoothing parameter whose value is chosen to minimize the oscillations in grid nodes trajectories. Above $\varpi\ =\ \varpi\,(\x,\,t)$ is the so-called \emph{monitoring function}, which determines the local density of the grid nodes and $\bigl\{g_{\,\alpha\,\alpha}\bigr\}_{\,\alpha\,=\,1}^{\,2}$ are metric tensor components defined in \eqref{eq:metric}.

The equidistribution method based on the solution of parabolic equations for grid nodes coordinates was proposed in \cite{Shokin1982}. A few years later this method was applied to the simulation of 1D (compressible) gas dynamics problems in \cite{Bell1983, Shokin1985}. This method has been especially designed for non-stationary problems and its application allows to avoid (or simply reduce) abrupt changes in the nodes positions when we move from one time layer to the following. In this method one constructs the new mesh by taking into account the position of nodes at the last time layer. The new coordinates of grid nodes are determined by solving a discretized version of the parabolic problem \eqref{eq:movingmesh}. This modification of the classical equidistribution method (proposed in \cite{DeBoor1974}) of the construction of dynamically adapted moving grids has been used to solve numerous 1D problems \cite{Darin1988, Darin1988a, Darin1988b, Degtyarev1993, Dorfi1987, Zegeling1992, Mazhukin1993}, this list is not being exhaustive.

The system of nonlinear equations \eqref{eq:movingmesh} is solved using an iterative method of alternating directions. The monitoring function $\varpi$ was computed at the known solution on the time layer $t\ =\ t^{\,n}\,$. As the initial guess for this iterative process we take the known mesh $\bigl\{\x_{\,\j}^{\,n}\bigr\}_{\,\j}$ from the previous time layer as well. We mention that at the initial moment of time $t\ =\ 0$ we employ the traditional (steady) equidistribution method as it was described in \cite{Khakimzyanov1999}. It is equivalent to solve the system \eqref{eq:movingmesh} with $\vsigma\ \equiv\ 0\,$.

Since the position of boundary nodes varies\footnote{Indeed, we are dealing with free surface flows in the presence of a moving wall.} from one time step to another, before solving equations \eqref{eq:movingmesh} we have to determine the grid $\bigl\{\x_{\,\j}^{\,n+1}\bigr\}_{\,\j}$ along the boundary $\partial\,\Q_{\,h}^{\,0}\,$. Let us assume that a portion $\Gammao\ \in\ \bigl\{\Gamma_{\,\ell},\,\Gamma_{\,r},\,\Gamma_{\,b},\,\Gamma_{\,s}\bigr\}\ \equiv\ \partial\,\Q^{\,0}$ of the boundary $\Gamma^{\,n+1}\ \equiv\ \Gamma(t^{\,n+1})$ is given in the following parametric form:
\begin{equation*}
  x\ =\ x\,(\zetat)\,, \qquad y\ =\ y\,(\zetat)\,,
\end{equation*}
where $\zetat$ is a real parameter. Then, the coordinates of grid nodes $\bigl\{\x_{\,\j}^{\,n+1}\bigr\}_{\,\j}\ \subseteq\ \Gammao^{\,h}$ on the boundary can be determined by formulas:
\begin{equation*}
  x_{\,\j}^{\,n+1}\ =\ x\,(\zetat_{\,j}^{\,n+1})\,, \qquad y_{\,\j}^{\,n+1}\ =\ y\,(\zetat_{\,j}^{\,n+1})\,.
\end{equation*}
We assume that the nodes on the boundary under consideration are re-ordered with a single scalar index $j$ for the sake of notation compactness. So, the problem is to determine the discrete scalar function $\bigl\{\zetat_{\,j}^{\,n+1}\bigr\}_{\,j}$ in the parameter space. To do it, we solve numerically the following scalar nonlinear difference equation:
\begin{equation}\label{eq:114}
  \vsigma\;\frac{\zetat_{\,j+1}^{\,n+1}\ -\ \zetat_{\,j}^{\,n+1}}{\tau_{\,n}}\ =\ \frac{1}{h}\;\biggl[\,\bigl\{\varpi^{\,n}\,\tilde{\J}^{\,n+1}\bigr\}_{\,j+\frac{1}{2}}\;\frac{\zetat_{\,j+1}^{\,n+1}\ -\ \zetat_{\,j}^{\,n+1}}{h}\ -\ \bigl\{\varpi^{\,n}\,\tilde{\J}^{\,n+1}\bigr\}_{\,j-\frac{1}{2}}\;\frac{\zetat_{\,j}^{\,n+1}\ -\ \zetat_{\,j-1}^{\,n+1}}{h}\,\biggr]\,,
\end{equation}
where $h\ \in\ \bigl\{h_{\,1},\,h_{\,2}\bigr\}$ is the discrete step in space, which depends on the boundary under consideration. The function $\tilde{\J}_{\,j}^{\,n+1}\ \eqdef\ \Bigl\{\sqrt{x_{\,\zeta}^{\,2}\ +\ y_{\,\zeta}^{\,2}}\Bigr\}\Bigr\vert_{\,j}^{\,n+1}$ is the trace of the transformation \textsc{Jacobian} at this boundary. Please, notice that we compute the monitoring function $\varpi^{\,n}$ on the known position of the nodes. The difference equation \eqref{eq:114} is solved on four sides (left and right walls, bottom and free surface) of the discretized rectangle $\Q_{\,h}^{\,0}\,$. The distribution of boundary nodes obtained in this way is used as the \textsc{Dirichlet}-type boundary condition for the 2D grid generation, which is achieved by solving numerically equation \eqref{eq:movingmesh}. By the end of this step we know the grid $\bar{\Omega}_{\,h}^{\,n+1}$ nodes coordinates $\bigl\{\x_{\,\j}^{\,n+1}\bigr\}_{\,\j}\,$.


\subsubsection{The re-computation}

At the next step we recompute quantities $\etab^{\,n+1}$ and $\bigl\{\phi^{\,n+1}_{\,\j}\bigr\}_{\,\j}$ to take into account the new mesh. For this, we solve again equation \eqref{eq:94} (or an analogue of \eqref{eq:103} on the boundaries) in which the \textsc{Cartesian} components of the velocity \eqref{eq:96} are still computed on the previous grid $\bar{\Omega}_{\,h}^{\,n}$ and using the velocity potential values from the previous time layer $t\ =\ t^{\,n}\,$. However, in formulas \eqref{eq:96} we take the \emph{new} velocity of grid nodes. In other words, formula \eqref{eq:95} is replaced by
\begin{equation*}
  \x_{\,t,\,\j}\ =\ \frac{\x_{\,\j}^{\,n+1}\ -\ \x_{\,\j}^{\,n}}{\tau_{\,n}}\,.
\end{equation*}
During the second solution of system \eqref{eq:104} the right hand side member \eqref{eq:105} is modified as well to take into account the new velocity of the left wall and, obviously, the new location of mesh nodes on this boundary. Another modification concerns the velocity components $u\,$, $v$ when they are used in the computation of the contravariant component $v_{\,\j}^{\,1,\,n}$ whose discretized formula is given in \eqref{eq:107}. The quantity $v_{\,\j}^{\,1,\,n}$ is needed to determine the new position of the free surface elevation $\etab^{\,n+1}\,$. During the present step, the components $u\,$, $v$ are computed as
\begin{equation*}
  u_{\,\j}^{\,n+1}\ =\ \biggl[\,\frac{\phi_{q^{\,1}}\cdot y_{q^{\,2}}\ -\ \phi_{q^{\,2}}\cdot y_{q^{\,1}}}{\J}\,\biggr]_{\,\j}^{\,n+1}\,, \qquad
  v_{\,\j}^{\,n+1}\ =\ \biggl[\,\frac{-\phi_{q^{\,1}}\cdot x_{q^{\,2}}\ +\ \phi_{q^{\,2}}\cdot x_{q^{\,1}}}{\J}\,\biggr]_{\,\j}^{\,n+1}\,.
\end{equation*}
The last formulas replace equations \eqref{eq:96} during this stage.


\subsection{Stability of the scheme}

In order to study the stability of the proposed scheme, we consider the linearized governing equations. Moreover, we consider only the Initial Value Problem\footnote{In fact, we have the so-called \textsc{Cauchy}--\textsc{Poisson} problem, which is of mixed type between IVP and BVP. We say that we have an IVP since there are no boundary conditions on lateral boundaries, \ie the domain is unbounded in the horizontal directions. However, there are still boundary conditions to be satisfied on the impermeable bottom and on the (linearized) free surface.} (IVP) in order to remove the complexity of (lateral) boundary conditions treatment. The linear water wave problem is known as the \textsc{Cauchy}--\textsc{Poisson} problem \cite{Stoker1957, Mei1994, ddk, Kervella2007} since the pioneering works of Augustin Louis~\textsc{Cauchy} \cite{Cauchy1827} and Sim\'{e}on Denis \textsc{Poisson} \cite{Poisson1818}. Let us formulate this problem in precise mathematical terms. Consider a fluid layer of infinite horizontal extent ($-\infty\ <\ x\ <\ +\infty$) over a solid bottom of uniform depth $y\ =\ -h_{\,0}\ =\ \const\,$. The fluid domain is bounded from above by the free surface, whose location is assumed to be at $y\ =\ 0$ after the linearization. The \textsc{Cauchy}--\textsc{Poisson} problem consists in finding the free surface elevation $\eta$ and the velocity potential $\phi$ by solving the \textsc{Laplace} equation\footnote{This equation is to be compared with equation \eqref{eq:laplace}.} in a strip:
\begin{equation*}
  \Delta\,\phi\ =\ 0\,, \qquad -\infty\ <\ x\ <\ +\infty\,, \qquad -\,h_{\,0}\ \leq\ y\ \leq\ 0\,.
\end{equation*}
The last equation has to be completed by the following free surface boundary conditions:
\begin{align}\label{eq:118a}
  \eta_{\,t}\ -\ \phi_{\,y}\ &=\ 0\,, \qquad y\ =\ 0\,, \\
  \phi_{\,t}\ +\ g\,\eta\ &=\ 0\,, \qquad y\ =\ 0\,,\label{eq:118b}
\end{align}
and by one bottom impermeability condition:
\begin{equation*}
  \phi_{\,y}\ =\ 0\,, \qquad y\ =\ -h_{\,0}\,.
\end{equation*}
It is not difficult to see that free surface boundary conditions given above are nothing else but linearized versions of equations \eqref{eq:kinfs} and \eqref{eq:dynfs}. A particular solution to the \textsc{Cauchy}--\textsc{Poisson} problem can be easily obtained using some elementary \textsc{Fourier} analysis \cite{Stoker1957, Mei1994, ddk, Kervella2007}:
\begin{equation}\label{eq:120}
  \eta\,(x,\,t)\ =\ \eta_{\,0}\,\ue^{\ui\,(k\,x\ -\ \omega\,t)}\,, \qquad
  \phi\,(x,\,y,\,t)\ =\ \phi_{\,0}\,\ue^{\ui\,(k\,x\ -\ \omega\,t)}\,\cosh\bigl[\,k\,(y\ +\ h_{\,0})\,\bigr]\,,
\end{equation}
where $k\ \eqdef\ \dfrac{2\,\pi}{\lambda}$ is the wave number, $\lambda$ is the wave length, $\omega\ \eqdef\ \dfrac{2\,\pi}{T}$ is the wave frequency and $T$ is its period. Wave amplitudes $\eta_{\,0}\ \in\ \R$ and $\phi_{\,0}\ \in\ \R$ are some real numbers. The wave frequency $\omega$ is related to the wave number $k$ through the so-called \emph{dispersion relation} of gravity waves:
\begin{equation}\label{eq:121}
  \omega\,(k)\ =\ \pm\,\sqrt{g\,k\,\tanh(k\,h_{\,0})}\,.
\end{equation}
For the sake of simplicity we consider only waves moving rightwards. It fixes the branch $+$ in the relation above. We reiterate the fact that the dispersion relation \eqref{eq:121} is a necessary condition for the existence of solutions \eqref{eq:120}.

The numerical scheme considered above can be applied to the \textsc{Cauchy}--\textsc{Poisson} problem as well. The semi-discretization in time of free surface boundary conditions \eqref{eq:118a}, \eqref{eq:118b} reads:
\begin{align}\label{eq:semi1}
  \frac{\eta^{\,n+1}\,(x)\ -\ \eta^{\,n}\,(x)}{\tau_{\,n}}\ -\ \phi_{\,y}^{\,n+1}\ &=\ 0\,, \\
  \frac{\phi^{\,n+1}\,(x,\,0)\ -\ \phi^{\,n}\,(x,\,0)}{\tau_{\,n}}\ +\ g\,\eta^{\,n}\,(x)\ &=\ 0\,.\label{eq:semi2}
\end{align}
From now on, for the sake of simplicity we take the time step $\tau_{\,n}\ \equiv\ \tau\ >\ 0$ to be constant. The elementary solution \eqref{eq:120} can be semi-discretized as well:
\begin{equation*}
  \eta^{\,n}\,(x)\ =\ \eta_{\,0}\,\rhou^{\,n}\,\ue^{\ui\,k\,x}\,, \qquad
  \phi\,(x,\,y,\,t)\ =\ \phi_{\,0}\,\rhou^{\,n}\,\ue^{\ui\,k\,x}\,\cosh\bigl[\,k\,(y\ +\ h_{\,0})\,\bigr]\,,
\end{equation*}
where we introduced the notation $\rhou\ \eqdef\ \ue^{-\,\ui\,\omega\,\tau}\,$. By substituting this semi-discrete solution ansatz into relations \eqref{eq:semi1}, \eqref{eq:semi2} we obtain the following linear system of equations with respect to wave amplitudes $\eta_{\,0}$ and $\phi_{\,0}\,$:
\begin{align*}
  g\,\eta_{\,0}\ +\ \phi_{\,0}\;\frac{\rho\ -\ 1}{\tau}\;\cosh\bigl[\,k\,h_{\,0}\,\bigr]\ &=\ 0\,, \\
  \frac{\rho\ -\ 1}{\tau}\;\eta_{\,0}\ -\ \phi_{\,0}\,\rhou\,k\,\sinh\bigl[\,k\,h_{\,0}\,\bigr]\ &=\ 0\,.
\end{align*}
In order to have non-trivial solutions, the determinant of this system has to vanish. It gives us the following quadratic equation with respect to the transfer coefficient $\rhou\,$:
\begin{equation}\label{eq:quad}
  \rhou^{\,2}\ -\ 2\,\biggl(1\ -\ \frac{\tau^{\,2}\,\omega^{\,2}}{2}\biggr)\,\rhou\ +\ 1\ =\ 0\,.
\end{equation}
To have the linear stability property of our scheme, it is necessary that both roots of this equation verify the inequality $\abs{\rho^{\,\pm}}\ \leq\ 1\,$. To meet this requirement, it is sufficient to ask that the discriminant of the quadratic equation \eqref{eq:quad} is not positive, \ie
\begin{equation*}
  \tau^{\,2}\,\omega^{\,2}\,\Bigl(\frac{\tau^{\,2}\,\omega^{\,2}}{4}\ -\ 1\Bigr)\ \leq\ 0\,,
\end{equation*}
or equivalently
\begin{equation}\label{eq:122}
  \tau\ \leq\ \frac{2}{\omega}\,,
\end{equation}
where the wave frequency was defined in equation \eqref{eq:121} (with the sign $+$ by our convention).

Let $\Delta\,x$ be the discretization step size in the horizontal direction. The minimal wave length $\lambda_{\,\inf}\,$, which can be represented on the grid with spacing $\Delta\,x$ is $\lambda_{\,\inf}\ \equiv\ 2\,\Delta\,x\,$. All other waves satisfy the inequality $\lambda\ \geq\ \lambda_{\,\inf}\,$. These considerations on the wave length can be translated into the language of wave numbers, \ie
\begin{equation*}
  k\ \leq\ k_{\,\sup}\ \eqdef\ \frac{\pi}{\Delta\,x}\,.
\end{equation*}
Henceforth, we can derive the following estimation for the wave frequency using the dispersion relation \eqref{eq:121}:
\begin{equation*}
  \omega\,(k)\ =\ \sqrt{g\,h_{\,0}}\cdot\sqrt{\frac{k}{h_{\,0}}\;\tanh\bigl[\,k\,h_{\,0}\,\bigr]}\ =\ k\,\sqrt{g\,h_{\,0}}\cdot\sqrt{\frac{\tanh\bigl[\,k\,h_{\,0}\,\bigr]}{k\,h_{\,0}}}\ \leq\ \sqrt{g\,h_{\,0}}\;\frac{\pi}{\Delta\,x}\,.
\end{equation*}
Consequently, in order to satisfy the stability condition \eqref{eq:122}, it is sufficient to impose the following restriction on the time step $\tau$ magnitude:
\begin{equation}\label{eq:cfl}
  \tau\ \leq\ \frac{2\,\Delta\,x}{\pi\,\sqrt{g\,h_{\,0}}}\,.
\end{equation}
An analogue of this stability condition will be used below during the simulation of fully nonlinear problems. The most important conclusion of this Section is that we have a hyperbolic-type stability \textsc{Courant}--\textsc{Friedrichs}--\textsc{Lewy} condition \cite{Courant1928} --- the time step is a linear function of the mesh spacing.


\subsubsection{Practical choice of the time step}

In the previous Section we studied the scheme stability in the linear case. However, in numerical simulations presented below, we solve the nonlinear problem. Henceforth, one may ask the question how to choose the time step in practical nonlinear simulations. Below we explain our approach to this problem.

In order to cover the whole family of problems, we choose to work in scaled variables. The CFL condition \eqref{eq:cfl} can be rewritten in dimensionless variables as
\begin{equation*}
  \tau^{\,\ast}\ \leq\ \frac{2}{\pi}\;h_{\,1}^{\,\ast}\ \approx\ 0.64\,h_{\,1}^{\,\ast}\,,
\end{equation*}
where $h_{\,1}^{\,\ast}\ \eqdef\ \dfrac{h_{\,1}}{h_{\,0}}\,$. As grid spacing $h_{\,1}\,$, we take the smallest horizontal spacing along the free surface (since most important variations take place there):
\begin{equation*}
  h_{\,1}\ \sim\ \Delta x_{\,\min}^{\,n}\ \eqdef\ \min_{0\ \leq\ j_{\,1}\ <\ N_{\,1}}\abs{x_{\,j_{\,1} + 1,\,N_{\,2}}\ -\ x_{\,j_{\,1},\,N_{\,2}}}\,.
\end{equation*}
Consequently, the dimensionless CFL condition can be rewritten as
\begin{equation*}
  \tau_{\,n}^{\,\ast}\ \leq\ 0.64\,\Delta x_{\,\min}^{\,n,\,\ast}\,.
\end{equation*}
However, our numerical simulations show that this estimation is too pessimistic. Consequently, in all simulations presented below we took the time step according to the following less restrictive formula:
\begin{equation*}
  \tau_{\,n}^{\,\ast}\ \leq\ \kappa\,\Delta x_{\,\min}^{\,n,\,\ast}\,,
\end{equation*}
with $\kappa\ =\ 0.95\,$. Most probably, the last condition is pessimistic as well. However, it guaranteed the stability of our \emph{nonlinear} computations.


\section{Numerical results}
\label{sec:res}

Above we described the proposed finite difference scheme and our resolution algorithm on a fixed reference domain. Below we present several validation tests and numerical experiments which show the performance and abilities of our numerical approach.


\subsection{Solitary wave run-up on a fixed wall}
\label{sec:fix}

In order to illustrate the applicability of our numerical algorithm, we consider the classical problem of the solitary wave/fixed wall interaction. Due to symmetry considerations, this set-up is equivalent to the head-on collision of two equal solitary waves. This problem is well-studied in the literature \cite{Maxworthy1976, Mirie1982, Su1980, Cooker1997, Chambarel2009} and it can serve as the first validation test.

Consider a solitary wave of amplitude $\alphau$ moving in the leftward direction. The channel has a constant depth and the wall is assumed to be fixed in order to be able to perform comparisons with previous investigations. On the right the channel is also bounded by a fixed vertical wall. The total length of the channel is equal to $\ell\ =\ 20\,$. In this Section we provide all values in dimensionless variables as it was explained earlier in Section~\ref{sec:dimless}. Moreover, we assume in these computations that waves do not overturn. In other words, the free surface is traditionally given as the graph of a function $y\ =\ \eta\,(x,\,t)\,$. The initial condition for the free surface elevation $\eta_{\,0}\,(x)\ \equiv\ \eta\,(x,\,0)$ and velocity field $\u_{\,0}\,(\x)\ \equiv\ \u\,(\x,\,0)$ are given by the following approximate formulas \cite{Ovsyannikov1983}:
\begin{equation}\label{eq:ic1}
  \eta_{\,0}\,(x)\ =\ \alphau\,\sech^{\,2}\Bigl[\,\underbrace{\frac{\kappa}{2}\;(x\ -\ x_{\,0})}_{\displaystyle{\defeq\ \Uptheta}}\,\Bigr]\,, \qquad \kappa\ \eqdef\ \sqrt{\frac{3\,\alphau}{\alphau\ +\ 1}}\,,
\end{equation}
\begin{multline}\label{eq:ic2}
  u_{\,0}\,(x,\,y)\ =\ -\,\sqrt{1\ +\ \alphau}\;\frac{\eta_{\,0}\,(x)}{1\ +\ \eta_{\,0}\,(x)}\ +\ \frac{\alphau^{\,2}}{\sqrt{1\ +\ \alphau}}\;\biggl[\,\frac{1}{4}\ -\ \frac{3}{4}\;\Bigl(\frac{y\ +\ 1}{\eta_{\,0}\,(x)\ +\ 1}\Bigr)^{\,2}\,\biggr]\cdot \\
  \cdot\biggl[\,2\;\frac{\eta_{\,0}\,(x)\ -\ 1}{\eta_{\,0}\,(x)\ +\ 1}\;\sech^{\,2}\,\Uptheta\ +\ \frac{3\ -\ \eta_{\,0}\,(x)}{\eta_{\,0}\,(x)\ +\ 1}\;\sech^{\,4}\,\Uptheta\,\biggr]\,,
\end{multline}
\begin{equation}\label{eq:ic3}
  v_{\,0}\,(x,\,y)\ =\ -\,\sqrt{3\,\alphau^{\,3}}\,(1\ +\ y)\;\frac{\cosh\,\Uptheta\,\sinh\,\Uptheta}{\bigl(\alphau\ +\ \cosh^{\,2}\,\Uptheta\bigr)^{\,2}}\,,
\end{equation}
where $x_{\,0}$ is the wave crest initial position.

The interaction of a solitary wave with amplitude $\alphau\ =\ 0.4$ with a vertical (fixed) wall is shown in Figure~\ref{fig:sw} at different moments of time. Under a wave we show also the adapted grid. On this Figure we can see that after the reflection the solitary wave does not recover completely its initial shape. In particular, behind the main wave we observe a slight dispersive tail. The numerical simulations on refined grids show that it is not a numerical effect. This dispersive tail appears on all grids and it reflects an intrinsic property of the full \textsc{Euler} equations --- their non-integrability \cite{Dyachenko1994a, Dyachenko2013}. In this example we use the following monitoring function in order to adapt the grid to the solution:
\begin{equation}\label{eq:mon}
  \varpi\,(\x,\,t)\ =\ 1\ +\ \upvartheta\,\abs{\eta\,(x,\,t)}\,,
\end{equation}
where $\upvartheta\ >\ 0$ is a positive ad-hoc parameter. In computations presented in the Section we use $\upvartheta\ =\ 10\,$. The main rationale behind this choice of the monitoring function $\varpi\,(\x,\,t)$ is to put more nodes in areas where the waves are large.

In Figure~\ref{fig:sw} we can see that the refined area follows somehow the solitary wave crest during its motion towards and fromwards the vertical wall. From equation \eqref{eq:mon} it can be seen that the monitoring function does not depend on the vertical variable $y\,$. As a result, we obtain the grids with almost vertical lines. Consequently, in accordance with stability condition \eqref{eq:cfl} to determine an admissible local time step $\tau_{\,n}$ we can use the following formula:
\begin{equation*}
  \tau_{\,n}\ =\ \upnu\,\Delta\,x_{\,\min}^{\,n}\,,
\end{equation*}
where $0\ <\ \upnu\ \leq\ 1$ is the \emph{security factor} and $\Delta\,x_{\,\min}^{\,n}$ is the minimal mesh spacing on the free surface at time $t\ =\ t^{\,n}\,$, \ie
\begin{equation*}
  \Delta\,x_{\,\min}^{\,n}\ \eqdef\ \min_{0\ \leq\ j_1\ <\ N_1}\bigl\{\abs{x_{\,j_{\,1}+1,\,N_2}^{\,n}\ -\ x_{\,j_{\,1},\,N_2}^{\,n}}\bigr\}\,.
\end{equation*}

\begin{figure}
  \centering
  \subfigure[$t\ =\ 0$]{\includegraphics[width=0.79\textwidth]{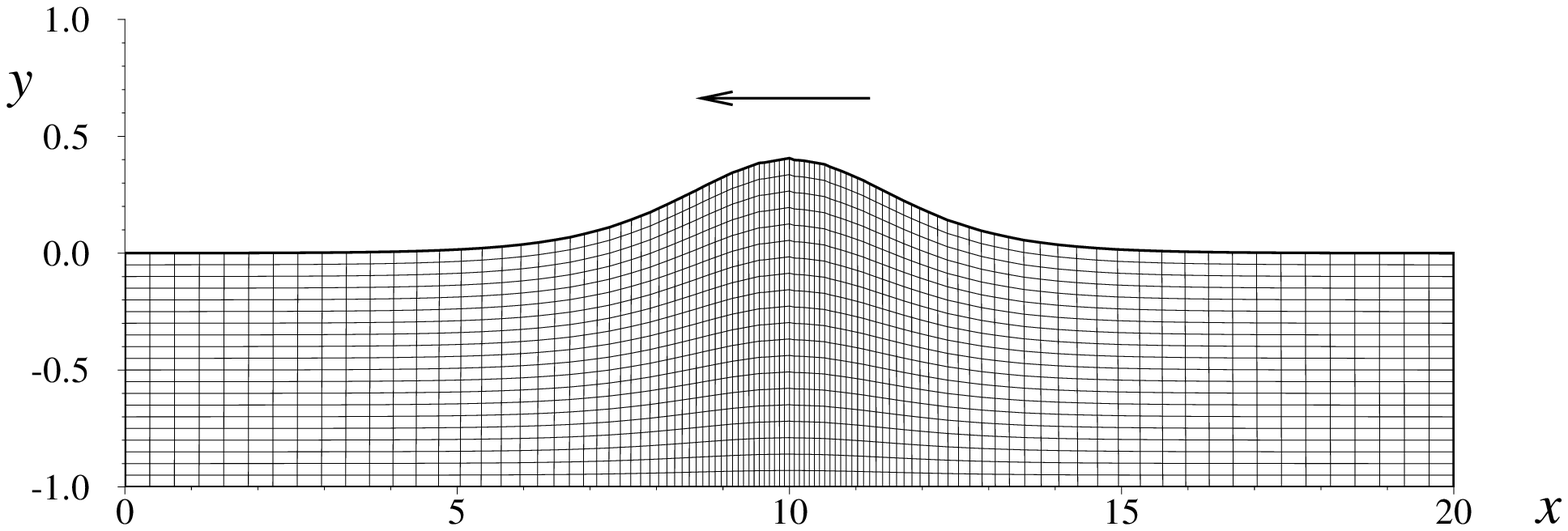}}
  \subfigure[$t\ =\ 7$]{\includegraphics[width=0.79\textwidth]{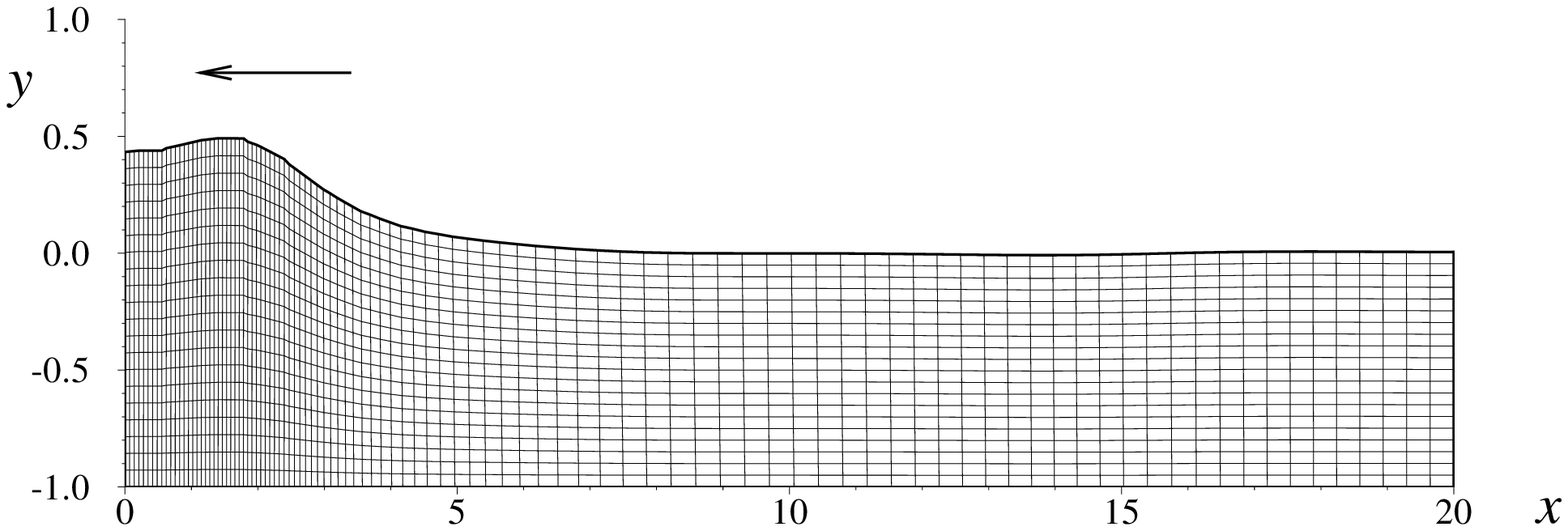}}
  \subfigure[$t\ =\ 9$]{\includegraphics[width=0.79\textwidth]{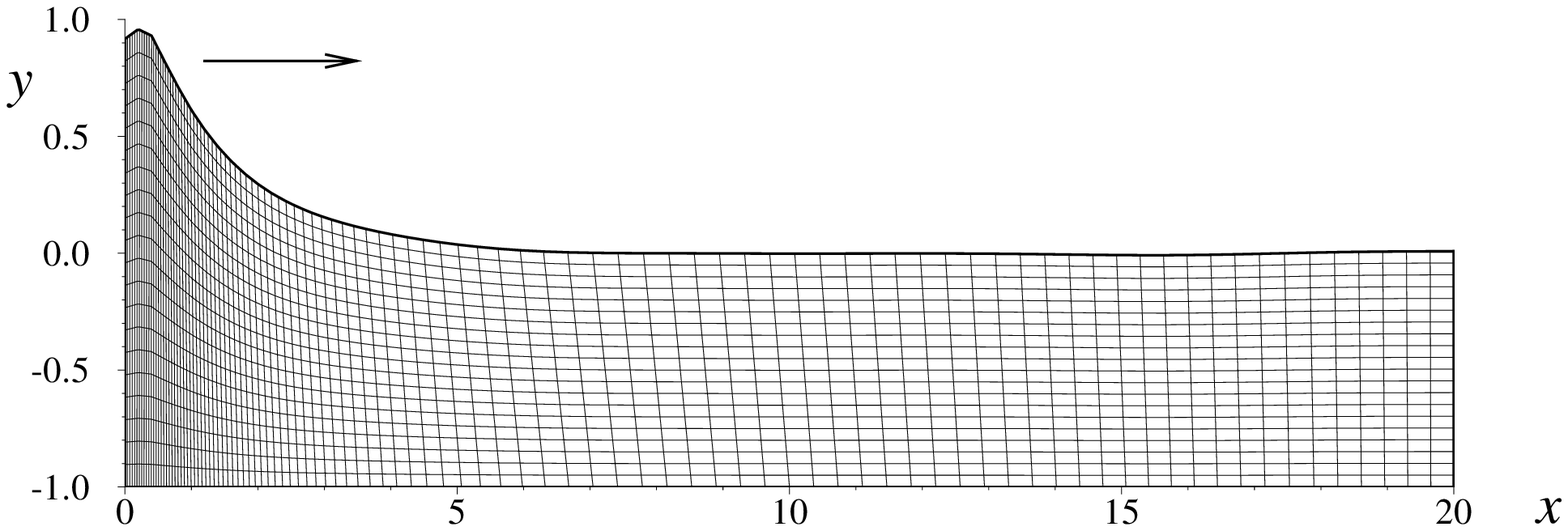}}
  \subfigure[$t\ =\ 18$]{\includegraphics[width=0.79\textwidth]{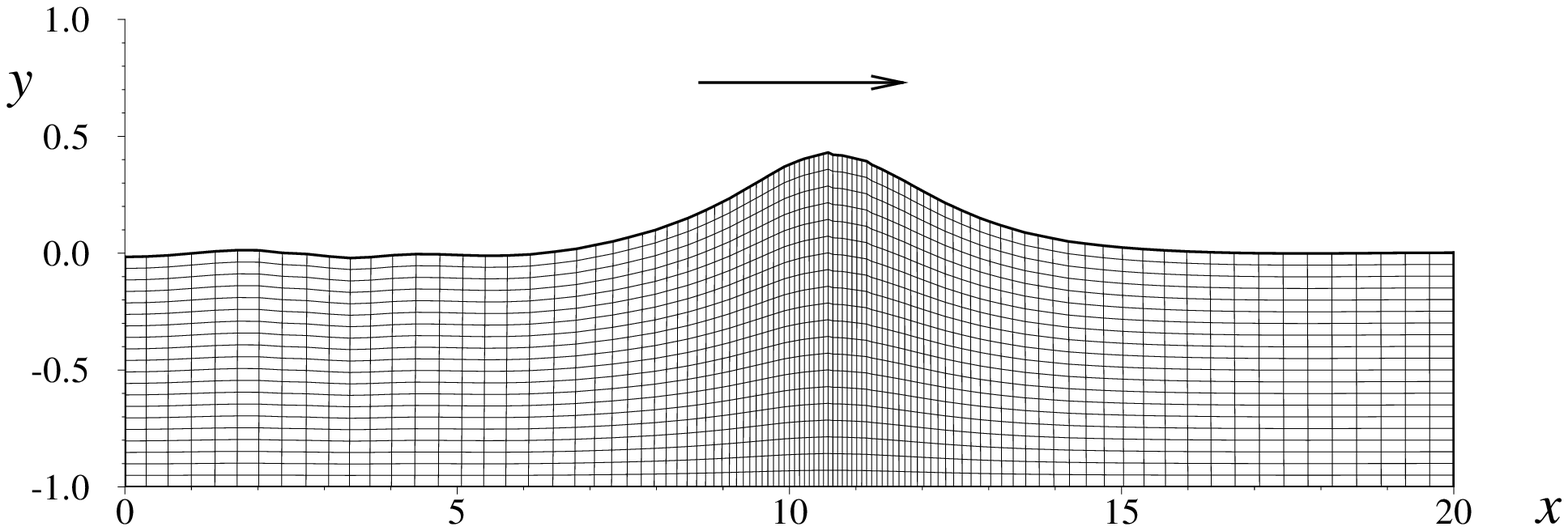}}
  \caption{\small\em Solitary Wave (SW) interaction with a vertical fixed wall at different instances of time. SW's amplitude is $\alphau\ =\ 0.4\,$. The propagation direction is shown with a horizontal arrow.}
  \label{fig:sw}
\end{figure}

One of the main characteristics of the wave/wall interaction is the maximal amplitude of the wave on the wall. This quantity is called the \emph{maximal run-up} and will be denoted in our study as $\Ru_{\,\max}\,$. Obviously, in our experimental conditions this quantity depends on the amplitude of the incident solitary wave, \ie $\Ru_{\,\max}\ =\ \Ru_{\,\max}\,(\alphau)\,$. In Figure~\ref{fig:runup} we represent this dependence according to our numerical simulations (solid black line), experimental data (filled markers), other computations (empty markers) and the following asymptotic analytical prediction \cite{Su1980} (dashed line):
\begin{equation*}
  \Ru_{\,\max}\,(\alphau)\ =\ 2\,\alphau\,\Bigl[\,1\ +\ \frac{1}{4}\;\alphau\ +\ \frac{3}{8}\;\alphau^{\,2}\,\Bigr]\ +\ o(\alphau^{\,3})\,, \qquad\mbox{as}\qquad \alphau\ \to\ 0\,.
\end{equation*}
In Figure~\ref{fig:runup} we can see that there is an overall good agreement among all presented data up to the amplitude $\alphau\ \lesssim\ 0.4\,$. For higher waves some divergences start to appear. However, the agreement between our numerical model with other potential flow solvers \cite{Fenton1982, Cooker1997} continues up to $\alphau\ \lesssim\ 0.6\,$. The experimental points go rather below our predictions. It can be easily explained by the neglection of viscous and friction effects in our numerical (and mathematical) model. Nevertheless, we would like to mention that our numerical results agree particularly well with experimental data reported in \cite{Davletshin1984}.

\begin{figure}
  \centering
  \includegraphics[width=0.89\textwidth]{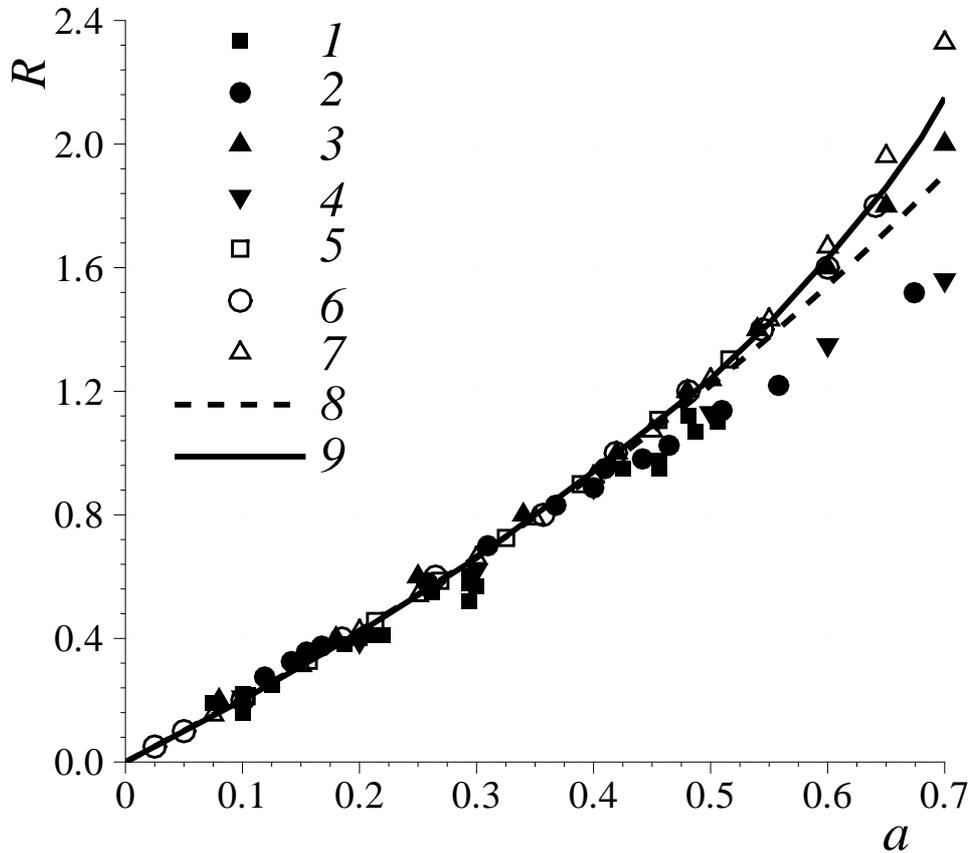}
  \caption{\small\em Dependence of the maximal run-up $\Ru_{\,\max}$ at the vertical wall on the incident solitary wave amplitude $\alphau\,$. Experimental data: (1) \cite{Zagryadskaya1980}, (2) \cite{Maxworthy1976}, (3) \cite{Davletshin1984}, (4) \cite{Manoylin1989}. Numerical data: (5) \cite{Fenton1982}, (6) \cite{Chan1970a}, (7) \cite{Cooker1997}. Analytical prediction \cite{Su1980} (8) (dashed line). The present study (9) is given by the solid line.}
  \label{fig:runup}
\end{figure}


\subsection{Wave generation by a numerical wave maker}

The study of water wave interaction with movable partially submerged bodies and objects is a very important problem of the modern computational hydrodynamics. In the previous Section~\ref{sec:fix} the wall was considered to be fixed in order to compare our numerical predictions with available data. Starting from this Section we allow the left wall as a movable object. As the first step towards the freely moving solid boundary, we consider first the situation where left wall motion is prescribed by a given law. Physically, it corresponds to the vertical piston motion used in many experimental facilities. Traditionally, the generated waves are understood using linear or weakly nonlinear theories \cite{Madsen1971}. Recently this situation was modelled numerically with \textsc{Boussinesq}-type equations in \cite{Orszaghova2012}. The results presented below are fully nonlinear and fully dispersive.

Consider a numerical wave tank with horizontal bottom and of uniform depth $d\ =\ 1\ \mathsf{m}$ (when the water is unperturbed). The left wall is initially located in the point $x\ =\ 0\,$. For times $t\ >\ 0$ the wall moves according to the following law:
\begin{equation}\label{eq:130}
  s\,(t)\ =\ \alpha\,\bigl(1\ -\ \ue^{-\,\beta\,t}\bigr)\,\sin(\omega\,t)\,.
\end{equation}
Thus, during the initial times the wall moves to the right. In this way, the wall motion is completely determined by three parameters $\alpha\ >\ 0\,$, $\beta\ >\ 0$ and $\omega\,$. The first one ($\alpha$) specifies the maximal wall oscillation amplitude (in the horizontal extent), the second parameter $\beta$ controls the speed of the relaxation\footnote{It is easy to see that the wall oscillation amplitude tends to its maximum value $\alpha$ as $t\ \to\ +\infty\,$. This justifies the term `relaxation' associated to this parameter $\beta\,$. The speed of this convergence depends on the value of parameter $\beta\,$, \ie bigger is faster.} towards the stationary periodic regime and $\omega$ is the frequency of wall oscillations (or equivalently it controls the oscillation period).

In Figure~\ref{fig:5ab} we depict a few simulations results, which were obtained by varying the parameter $\omega$ for fixed values of two other parameters:
\begin{equation*}
  \alpha\ =\ 0.4\ \mathsf{m}\,, \qquad
  \beta\ =\ 0.5\ \mathsf{Hz}\,.
\end{equation*}
Figure~\ref{fig:5ab} shows clearly that the generated wave amplitude (measured at the moving wall) depends essentially on the wall oscillation period, which is determined by parameter $\omega\,$. This effect is illustrated in Figure~\ref{fig:5ab} for three different values of the parameter $\omega\ \in\ \bigl\{0.5\;\Hz,\,1.0\;\Hz,\,3.0\;\Hz\bigr\}\,$. For each wall trajectory we show also the corresponding free surface excursion on the moving vertical wall. The increase in $\omega$ results in the reduction of the oscillation period. We can also witness the relaxation effect on the free surface oscillation. Indeed, as the time $t$ goes on, the wave amplitude increases and becomes practically stationary. The speed of relaxation depends on the parameter $\beta\,$. We underline the fact that for all three (considered) values of the wall oscillation frequency $\omega\,$, the amplitude of oscillations $\alpha$ was the same ($\alpha\ \equiv\ 0.4\;\m$). Nevertheless, the wave amplitudes registered on the wall were significantly different. This is the first practical conclusion that we can draw from our numerical simulations: the piston oscillation period has a much bigger influence on the generated wave amplitude, than the amplitude of these oscillations. We can explain this observation with a simple physical argument: during `fast' piston oscillations, the energy is pumped into the fluid close to the piston faster than waves can evacuate this energy by propagating into the wave tank.

The second important observation that we can make based on our simulations is that the maximal wave run-up does not take place at the moments where the wall \emph{displacement} is maximal. A more detailed investigation showed that the wave amplitude on the wall is maximal when the wall \emph{acceleration} is maximal. It happens when the wall passes by its initial (or mean, or unperturbed) position $x\ =\ 0\,$. After passing this point, the wall decelerates and the wave has time to flow away from the piston, which reduces the value of run-up.

Finally, we can also observe that the wave run-up value on the wall is asymmetric with respect to the still (unperturbed) water level. This effect is much better visible for fast wall motions. Indeed, for $\omega\ =\ 3.0\;\Hz$ the maximal run-up is $\ \approx\ 43\ \cm\,$, while the maximal run-down is $\ \approx\ 34\ \cm\,$.

\begin{figure}
  \centering
  \subfigure[]{\includegraphics[width=0.49\textwidth]{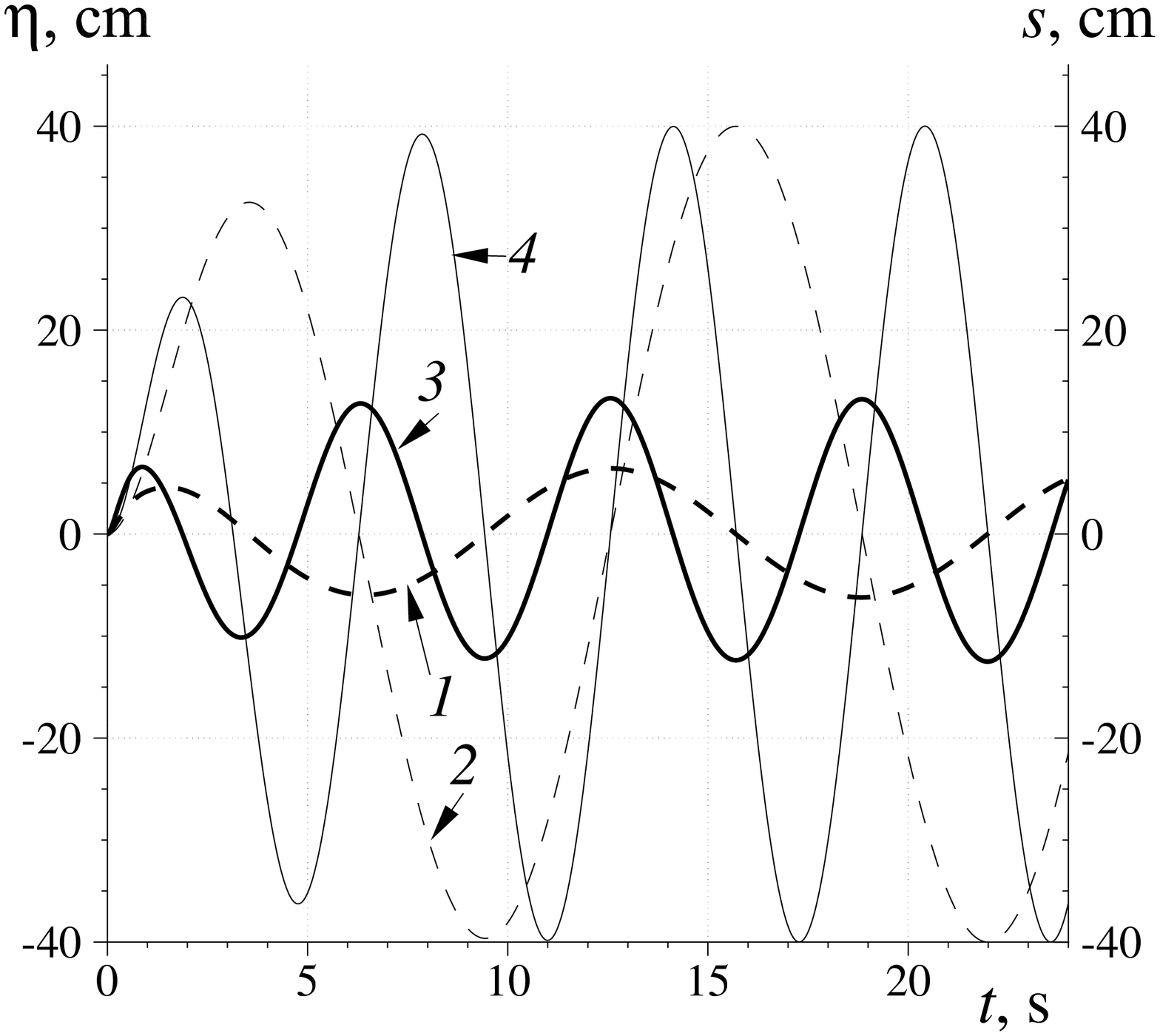}}
  \subfigure[]{\includegraphics[width=0.49\textwidth]{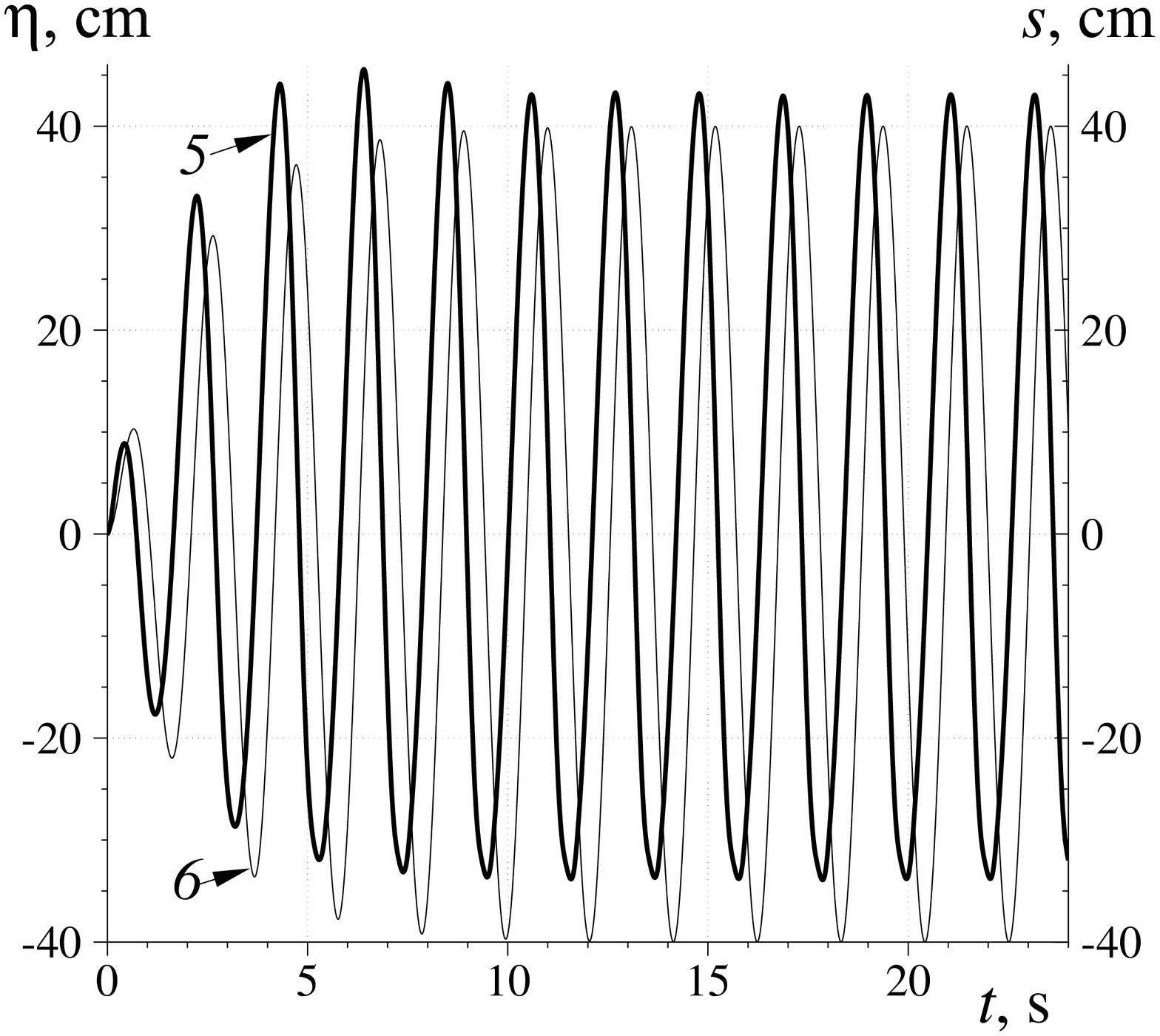}}
  \caption{\small\em Free surface excursion on the moving wall $\eta\,\bigl(s(t),\,t\bigr)$ (thick lines $1\,$, $3$ and $5$) along with the wall trajectory $s\,(t)$ (thin lines $2\,$, $4$ and $6$) for (a) $\omega\ =\ 0.5\;\mathsf{Hz}$ (lines $1$ and $2$); $\omega\ =\ 1.0\;\Hz$ (lines $3$ and $4$); (b) $\omega\ =\ 3.0\;\Hz$ (lines $5$ and $6$).}
  \label{fig:5ab}
\end{figure}


\subsection{Surface wave run-up on a movable wall}

In this Section we consider the wave/wall interaction problem, where the wall motion is not prescribed, but it is determined as a part of the problem solution. The left wall is attached to a spring system, which can deform under the wave action.


\subsubsection{Single pulse}

To begin, we consider first the case of a single localized wave impulse interacting with a moving wall. The channel depth is taken to be $d\ =\ 1\ \m$ and the channel length is $\ell\ =\ 20\ \m\,$. At the initial moment of time $t\ =\ 0\,$, the velocity field is taken to be quiescent and the free surface shape is given by the following formula:
\begin{equation}\label{eq:initcond}
  \eta_{\,0}\,(x)\ =\ \begin{dcases}
  \ \frac{\alphau}{2}\;\biggl[\,1\ +\ \cos\Bigl[\,\frac{\pi}{\Upsilon}\;\bigl(x\ -\ x_{\,0}\bigr)\,\Bigr]\,\biggr]\,, & \abs{x\ -\ x_{\,0}}\ \leq\ \Upsilon\,, \\
  \ 0\,, & \abs{x\ -\ x_{\,0}}\ >\ \Upsilon\,.
  \end{dcases}
\end{equation}
For subsequent times $t\ >\ 0$ this initial condition is separated in two waves of the amplitude $\ \approx\ \frac{\alphau}{2}\,$, which move in opposite directions. The wave moving leftwards interacts with the moving wall. We would like to underline the fact that the initial left wall position ($x\ =\ 0$) does not coincide with the wall \emph{equilibrium} position in the absence of water. Even if the free surface coincides with the still water level, the springs are deformed under the action of the hydrostatic pressure. Thus, the equilibrium wall position $x\ =\ x_{\,e}$ with undeformed springs is located somewhere to the right from its initial position, \ie $x_{\,e}\ >\ 0\,$. We performed a series of numerical simulations by varying the parameters $\alphau\,$, $x_{\,0}\,$ and $\Upsilon$ of the initial condition \eqref{eq:initcond}. Additionally we varied also the wall mass $m$ and the springs rigidity coefficient $k\,$.

First, we conducted a series of numerical experiments for small initial wave amplitudes $\alphau$ in view of comparisons against the numerical simulations reported in \cite{He2009}. The Authors of the preceding study employed the Boundary Integral Equations Method (BIEM) \cite{Brebbia1980, Brebbia1984}. Unfortunately, the Authors of \cite{He2009} did not publish any tabulated data to perform quantitative comparisons. However, we performed qualitative comparisons of numerical results. We choose the same values as in \cite{He2009}, \ie $x_{\,0}\ =\ 0.7\ \m\,$, $\Upsilon\ =\ 0.5\ \m$ and dimensionless spring rigidity coefficient $k\ =\ 10\,$. The wall dimensionless mass is $m\ =\ 1.5\,$.

\begin{figure}
  \centering
  \subfigure[]{\includegraphics[width=0.48\textwidth]{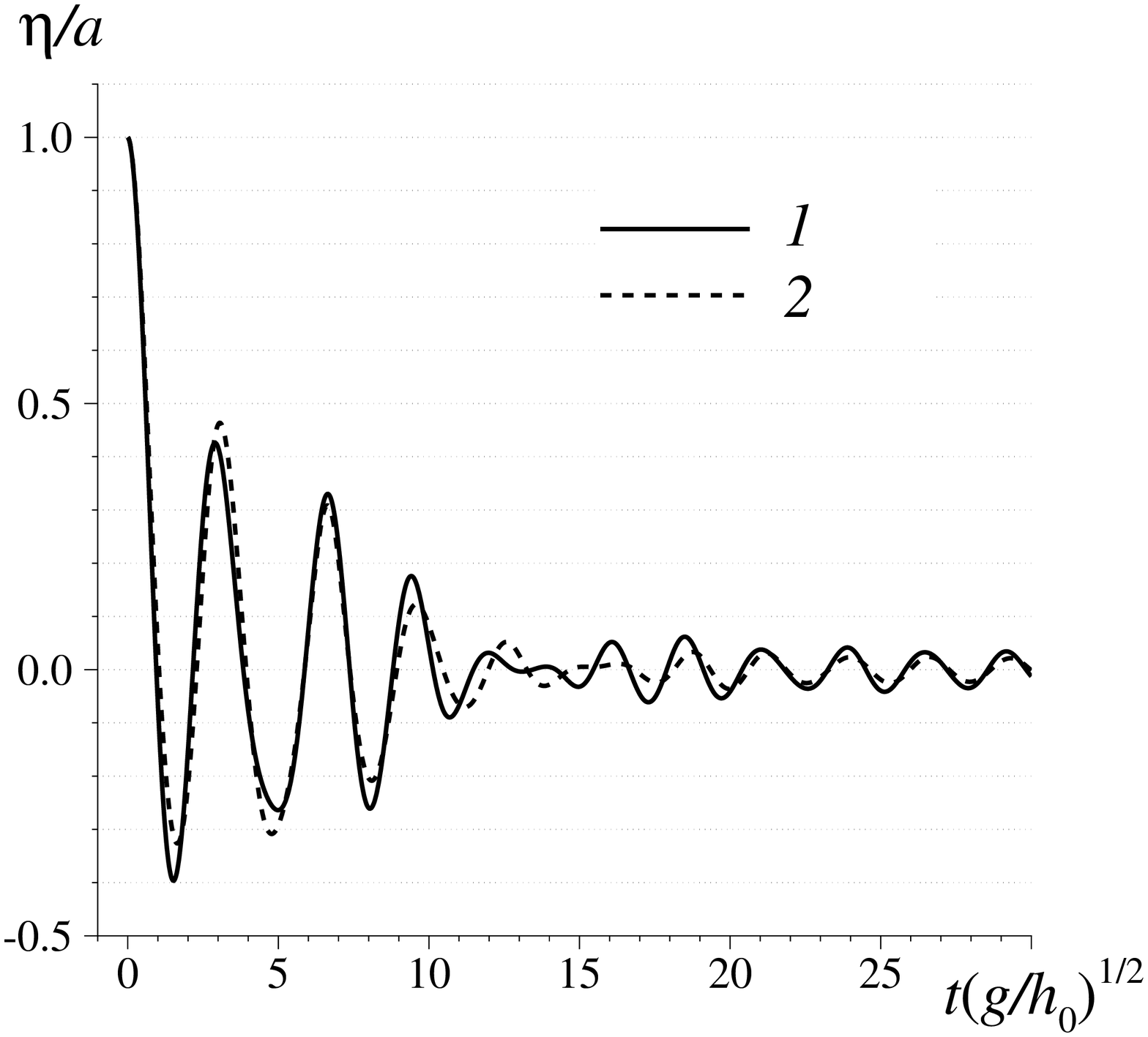}}
  \subfigure[]{\includegraphics[width=0.48\textwidth]{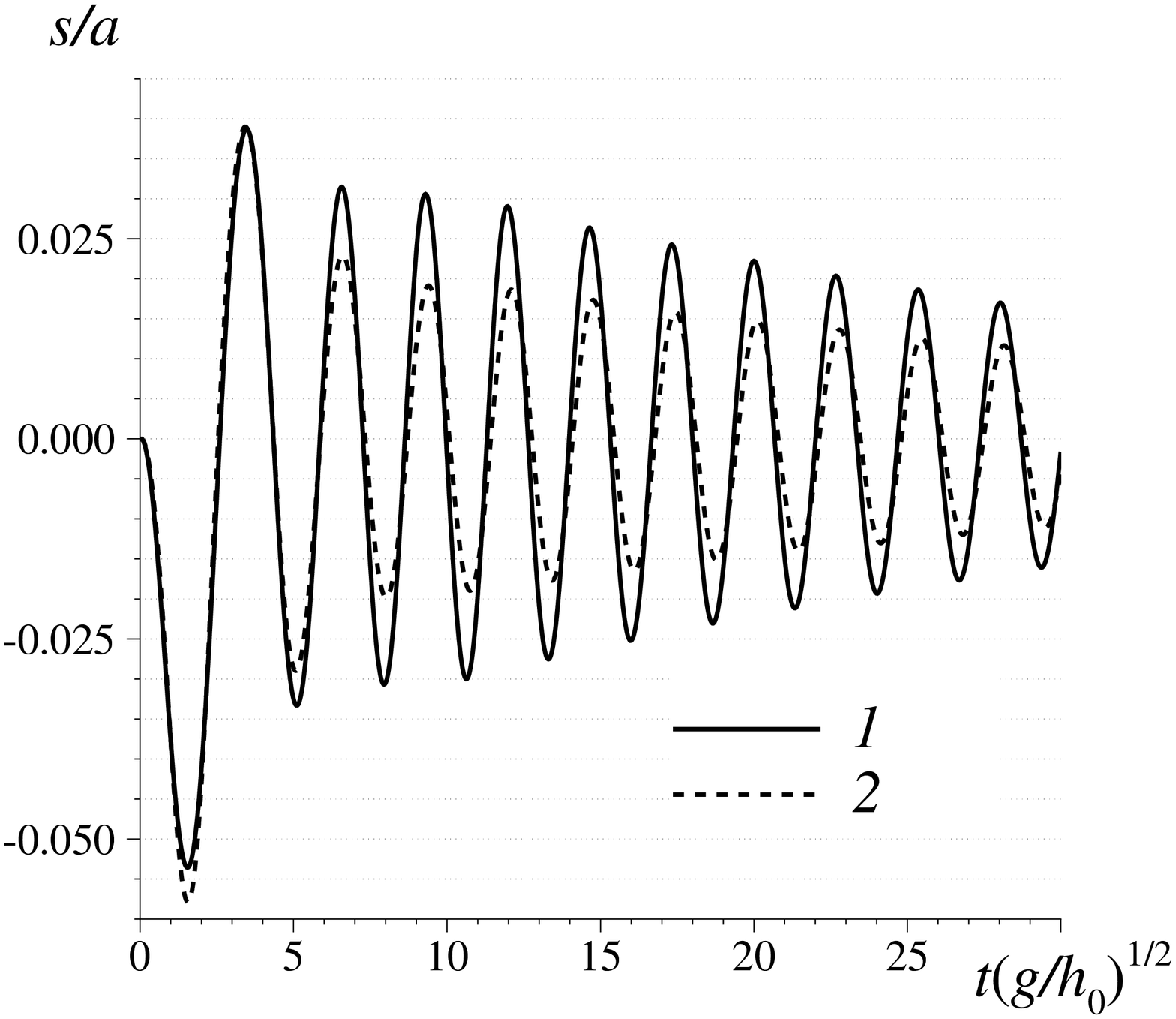}}
  \caption{\small\em Free surface excursion as a function of time at point $x\ =\ x_{\,0}\ =\ 0.7$ (a) and wall oscillations (b) for the initial condition given in Equation~\eqref{eq:initcond}. The parameter $\alphau\ =\ 1\ \cm$ (1) and $\alphau\ =\ 10\ \cm$ (2).}
  \label{fig:9}
\end{figure}

In Figure~\ref{fig:9}(\textit{a}) we show the free surface oscillations recorded at the initial pulse location\footnote{Here we mean the center or pulse crest.} $x\ =\ x_{\,0}\ =\ 0.7\ \m\,$. We show the results for two initial wave amplitudes. The line (1) corresponds to $\alphau\ =\ 1\ \cm\,$, while $\alphau\ =\ 10\ \cm$ is depicted by line (2). Please, notice that curves presented in Figure~\ref{fig:9} are shown relatively to the initial wave amplitude $\alphau\,$. These curves compare quantitatively and qualitatively well to those presented in \cite{He2009}. To the graphical resolution the oscillation decay rates and oscillation periods are the same. In Figure~\ref{fig:9}(\textit{b}) we show moving wall trajectories in time for two initial wave amplitudes $\alphau\ \in\ \bigl\{1\;\cm,\,10\;\cm\bigr\}\,$. These curves also compare well with previous computations reported in \cite{He2009}. We can also see that for small times the solid and dashed curves shown in Figure~\ref{fig:9}(\textit{a,b}) almost coincide. It means that for small initial wave amplitudes $\alphau$ the results of computations depend almost linearly on $\alphau\,$. This observation supports the applicability of the linear theory \cite{Korobkin2009}.

Let us study now the influence of the springs rigidity parameter $k$ on the resulting wave field. We took the following parameters in this computation:
\begin{equation*}
  x_{\,0}\ =\ 2\ \m\,, \qquad \Upsilon\ =\ 1\ \m\,, \qquad \alphau\ =\ 10\ \cm\,, \qquad m\ =\ 1.5\,.
\end{equation*}
The parameter $k$ is variable. In Figure~\ref{fig:10}(\textit{a}) we can see that the first run-up on the wall is large. Then, the following run-ups are generally decaying in amplitude since the wave energy spreads over the whole wave tank during the dispersive (and possibly nonlinear) effects. Moreover, for $k\ =\ 5$ (curve 2) the decay rate is the strongest since a larger part of the wave energy is converted into the wall motion. Thus, we have the wave energy transformation into the elastic energy of the spring system. When the springs rigidity is high, the wall performs small later motions, as it can be seen in Figure~\ref{fig:10}(\textit{b}). The increase in $k$ results also in the increase of the wall oscillation frequency around the initial equilibrium position $x\ =\ 0\,$. In the same time, the wave run-up records on the wall are practically identical with the fixed wall case (\cf curves 3 and 4 in Figure~\ref{fig:10}(\textit{a})). In other words, for high rigidities $k\ \gg\ 1\,$, the `moving' wall interacts in a way very similar to the fixed wall. On the other hand, when the rigidity decreases, the wall oscillation period and amplitude both increase. Wave run-up on the wall decay rate also increases with the parameter $k\,$. However, in contrast to large values of $k\,$, a flexible spring system can allow the maximal wave run-up value during the second (or even further subsequent) wave interactions with the wall. In Figure~\ref{fig:10}(\textit{a}) we can see that the value of the maximal wave run-up on the wall is generally lower for weak springs comparing to quasi-fixed wall results ($k\ \gg\ 1$). From these simulations we can conjecture that flexible spring systems should be used for more efficient wave run-up reduction.

\begin{figure}
  \centering
  \subfigure[]{\includegraphics[width=0.48\textwidth]{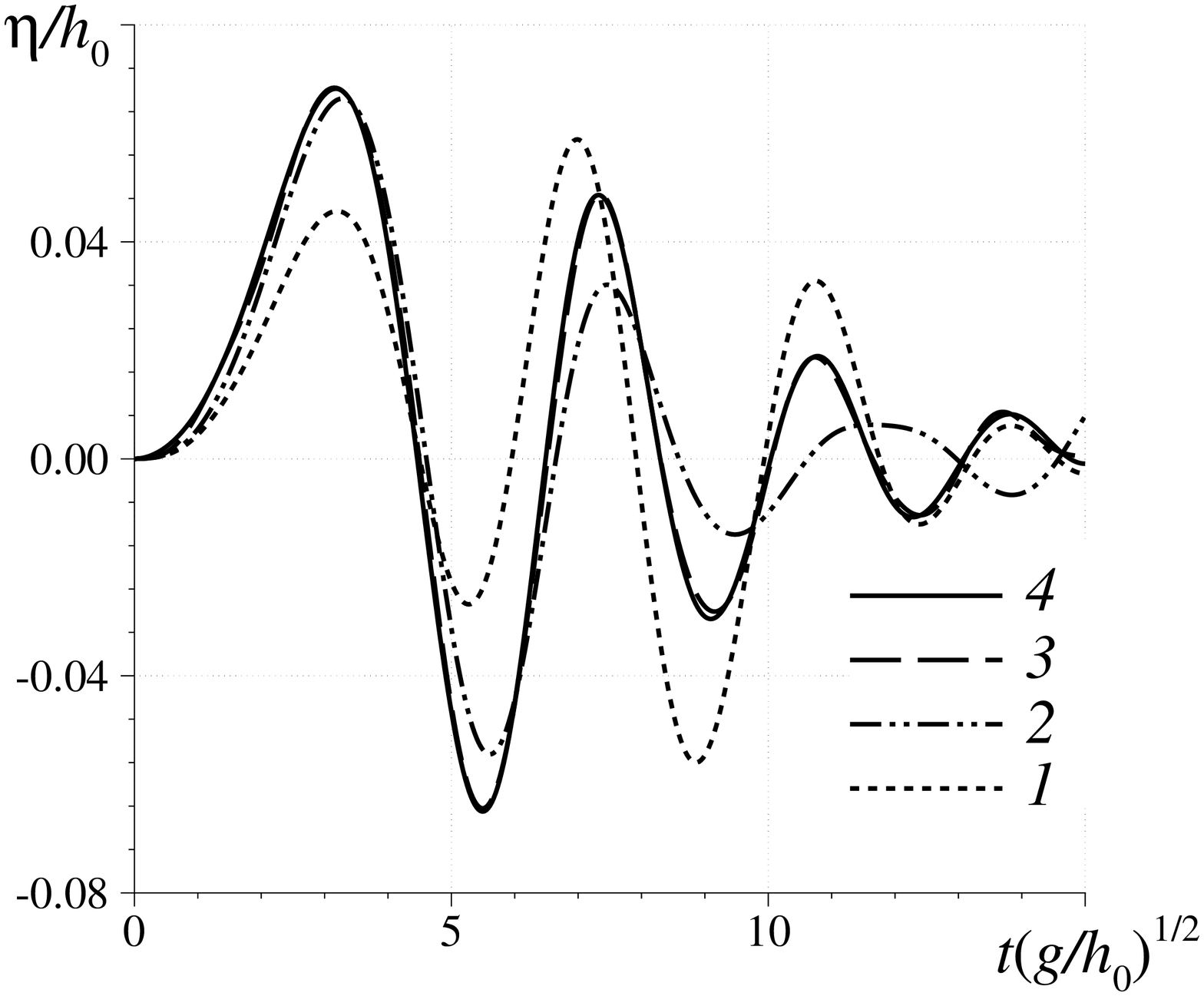}}
  \subfigure[]{\includegraphics[width=0.48\textwidth]{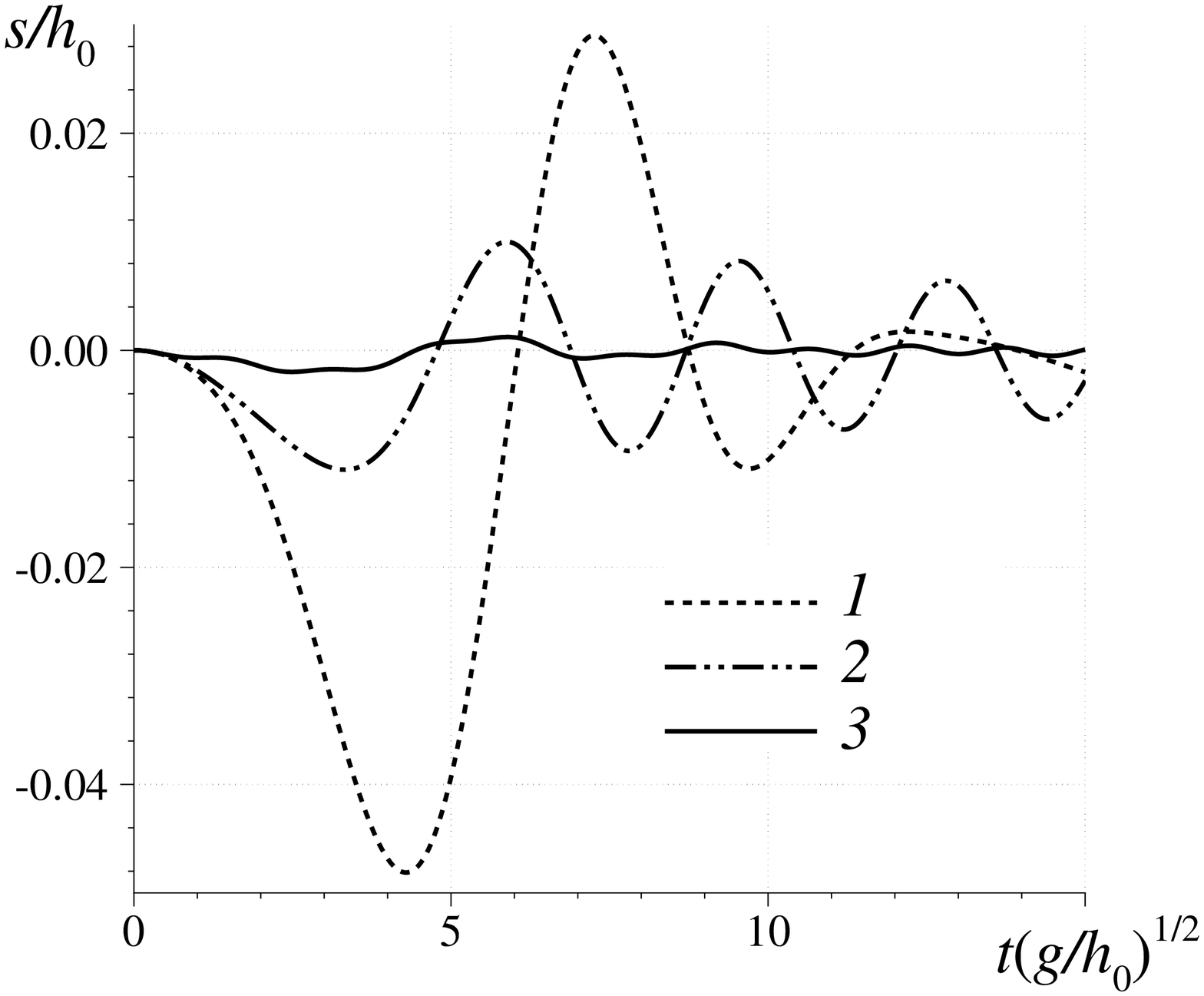}}
  \caption{\small\em Free surface excursion on the vertical moving wall (a) and the wall position (b) as functions of time for the spring rigidity parameter $k\ =\ 1$ (1), $k\ =\ 5$ (2) and $k\ =\ 25$ (3). For the sake of comparison the curve (4) shows the wave run-up on the fixed wall, \ie $k\ =\ \infty\,$.}
  \label{fig:10}
\end{figure}

Let us consider now the case of the weak springs system ($k\ =\ 1$) and a moving wall with mass $m\ =\ 1\,$. The question we would like to investigate is how the wave run-up and wall trajectory depend on the incident wave amplitude $\alphau\,$? In Figure~\ref{fig:11}(\textit{a}) we show the wave run-up value on the moving wall for a moderate amplitude $\alphau\ =\ 20\ \cm$ (curve 1) and a fairly strong amplitude $\alphau\ =\ 60\ \cm$ (curve 2). The initial condition is given by the single pulse formula \eqref{eq:initcond}. The initial impulse center is $x_{\,0}\ =\ 12\ \m$ and its half-length $\Upsilon\ =\ 5\ \m\,$. In Figure~\ref{fig:11}(\textit{b}) we show the wall displacement $s\,(t)\,$. As expected, the wall displacement becomes larger as the incident wave amplitude $\alphau$ increases. It can be seen that for weak springs the wall first retracts under the incident wave action. Then, it comes back towards its initial position and slightly oscillates around this point. It is very interesting to observe the high amplitude pulse (curve 2) interaction with the moving wall depicted in Figure~\ref{fig:11}(\textit{a}). Here the wave run-up value grows until the wall attains its left-most position. Then, on the way back the wall motion provokes the second maximal run-up value. It explains the non-monotonic behaviour of the curve 2 around its maximum. This behaviour is possible only for high amplitude waves and weak (\ie easily deformable) springs.

\begin{figure}
  \centering
  \subfigure[]{\includegraphics[width=0.48\textwidth]{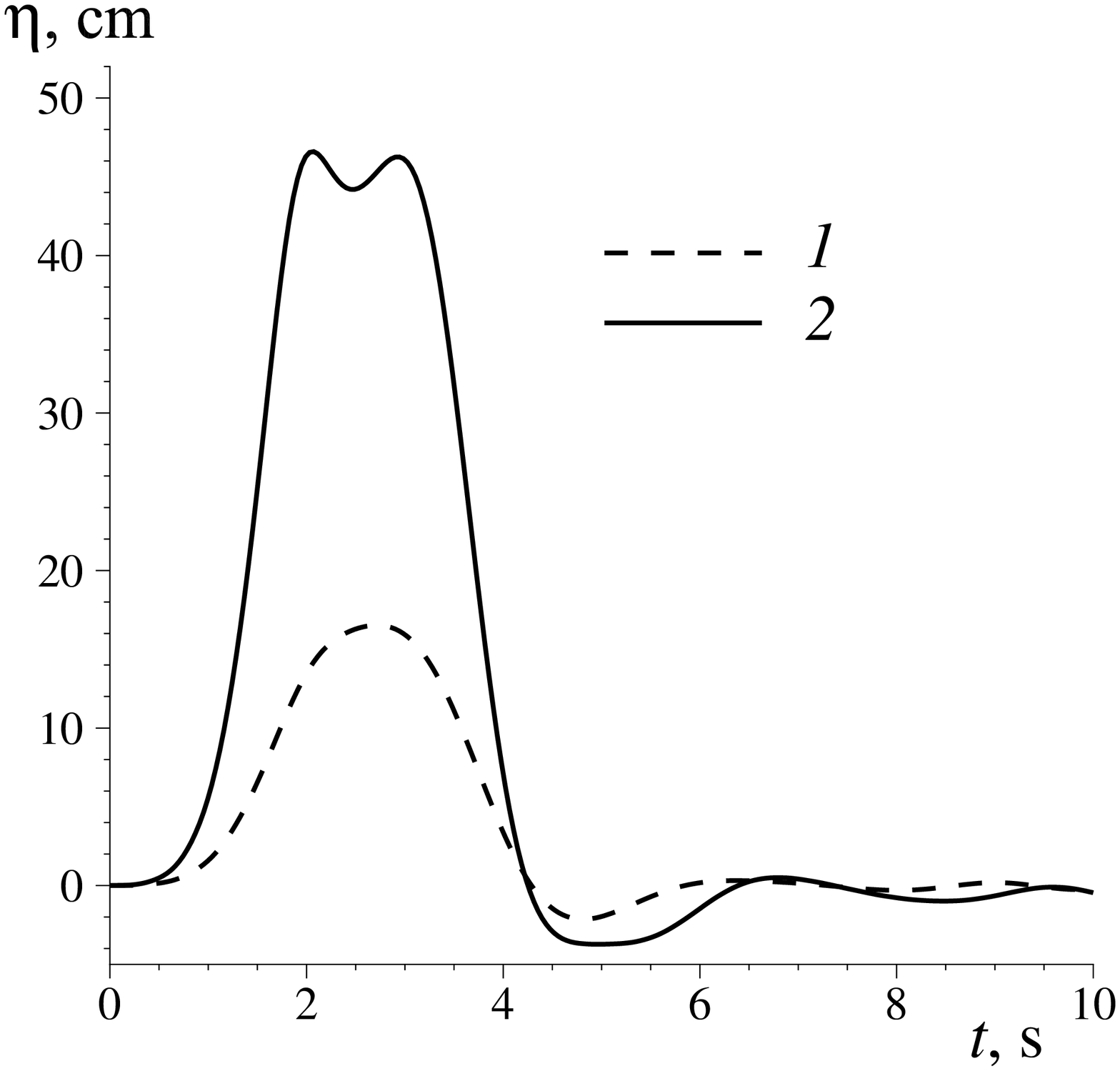}}
  \subfigure[]{\includegraphics[width=0.48\textwidth]{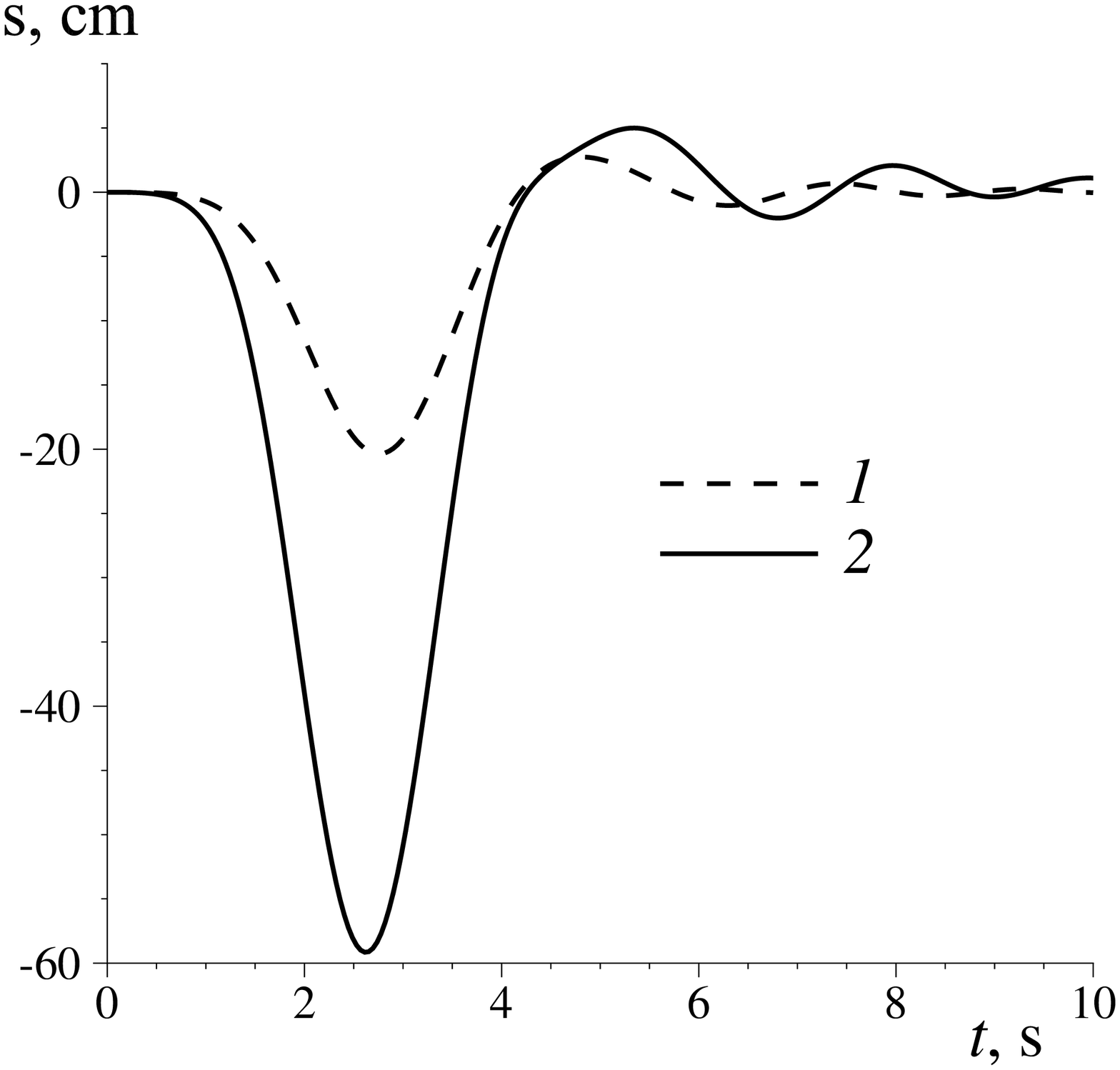}}
  \caption{\small\em Free surface excursion on the vertical moving wall (a) and the wall position (b) as functions of time for $\alphau\ =\ 20\ \cm$ (1) and $\alphau\ =\ 60\ \cm$ (2). The initial impulse center is $x_{\,0}\ =\ 12\ \m$ and its half-length $\Upsilon\ =\ 5\ \m\,$.}
  \label{fig:11}
\end{figure}

We naturally come to the question of the wall mass $m$ and springs rigidity $k$ on the maximal wave run-up value of the impulse \eqref{eq:initcond} on the movable wall. We use the following parameters of the incident pulse:
\begin{equation*}
  x_{\,0}\ =\ 12\ \m\,, \qquad
  \Upsilon\ =\ 5\ \m\,, \qquad
  \alphau\ =\ 20\ \cm\,.
\end{equation*}
In Figure~\ref{fig:12}(\textit{a}) we depict the maximal run-up value as the function of the wall mass $m\,$. It can be easily seen that for a fixed rigidity parameter $k\,$, the maximal run-up is a non-monotonic function of $m\,$. Even more, for every value of $k\,$, one can find the value $m^{\,\star}$ for which the maximal wave run-up will be minimal. It implies that the wave run-up can always be reduced by choosing the appropriate wall parameters $k^{\,\star}$ and $m^{\,\star}\,$. In Figure~\ref{fig:12}(\textit{b}) we show the maximal wave run-up dependence on $k$ for fixed dimensionless wall masses $m\ \in\ \bigl\{1,\,3,\,5\bigr\}\,$. So, it can be seen that the run-up depends monotonically on $k\,$. Namely, it increases with wall rigidity $k\,$. As we increase $k\ \to\ +\infty\,$, the maximal run-up tends to its limiting value, which depends on the incident wave parameters and which is independent on the dimensionless wall mass $m\,$.

\begin{figure}
  \centering
  \subfigure[]{\includegraphics[width=0.48\textwidth]{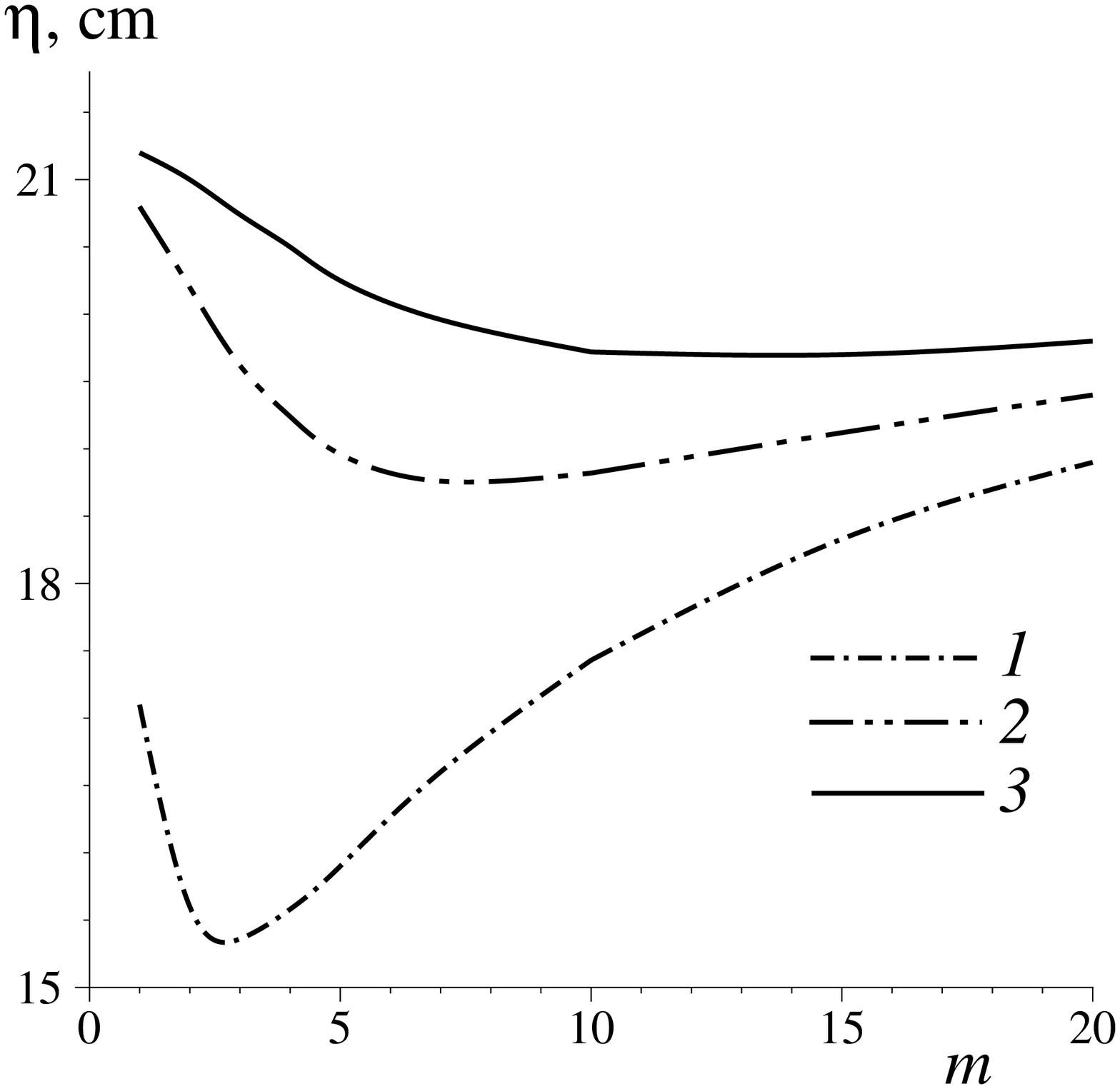}}
  \subfigure[]{\includegraphics[width=0.48\textwidth]{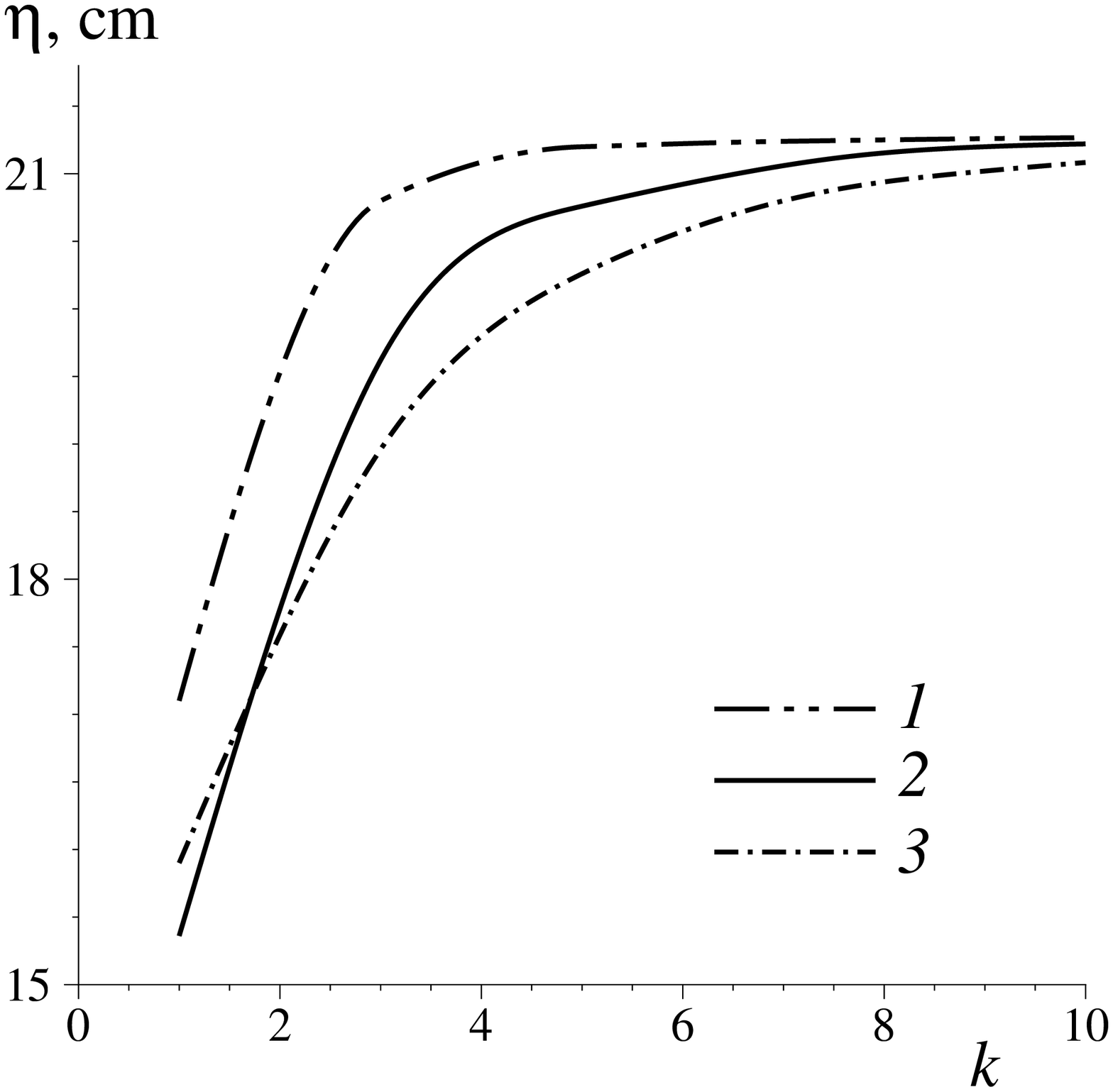}}
  \caption{\small\em Maximal run-up on the vertical wall (during the whole simulation time) as a function of the moving wall mass (a) for fixed springs rigidity $k\ =\ 1$ (1), $k\ =\ 3$ (2) and $k\ =\ 5$ (3). On the right panel we show the same quantity as a function of the springs rigidity (b) for fixed wall masses $m\ =\ 1$ (1), $m\ =\ 3$ (2) and $m\ =\ 5$ (3). This computation corresponds to the incident initial pulse \eqref{eq:initcond}.}
  \label{fig:12}
\end{figure}


\subsubsection{Solitary wave}

Finally, we show some numerical results on the solitary wave interaction with a moving wall. This test case is a natural extension of the case considered earlier in Section~\ref{sec:fix}. The initial condition is provided by the same formulas \eqref{eq:ic1} -- \eqref{eq:ic3}. In order to avoid any interactions of reflected waves from the right wall (during the simulation time), we increase the computational domain length to $\ell\ =\ 40\ \m$ (from $\ell\ =\ 20\ \m$). The computations were performed for a solitary wave with moderate amplitude $\alphau\ =\ 20\ \cm\,$. We would like to study the influence of moving wall parameters on the maximal solitary wave run-up. In Figure~\ref{fig:13} we show the maximal run-up as a function of the wall mass $m$ (\textit{a}) and as a function of springs rigidity $k$ (\textit{b}) (the other parameter being fixed). With a dashed line we show the maximal run-up value on a fixed wall. The comparison of this result with Figure~\ref{fig:12} indicates that there are qualitative similarities in the behaviour of a single pulse with a solitary wave (\cf Figure~\ref{fig:13}). So, for solitary waves the run-up can be significantly reduced as well. All differences observed in Figures~\ref{fig:12} and \ref{fig:13} can be solely explained by the incident wave amplitude and its shape.

\begin{figure}
  \centering
  \subfigure[]{\includegraphics[width=0.48\textwidth]{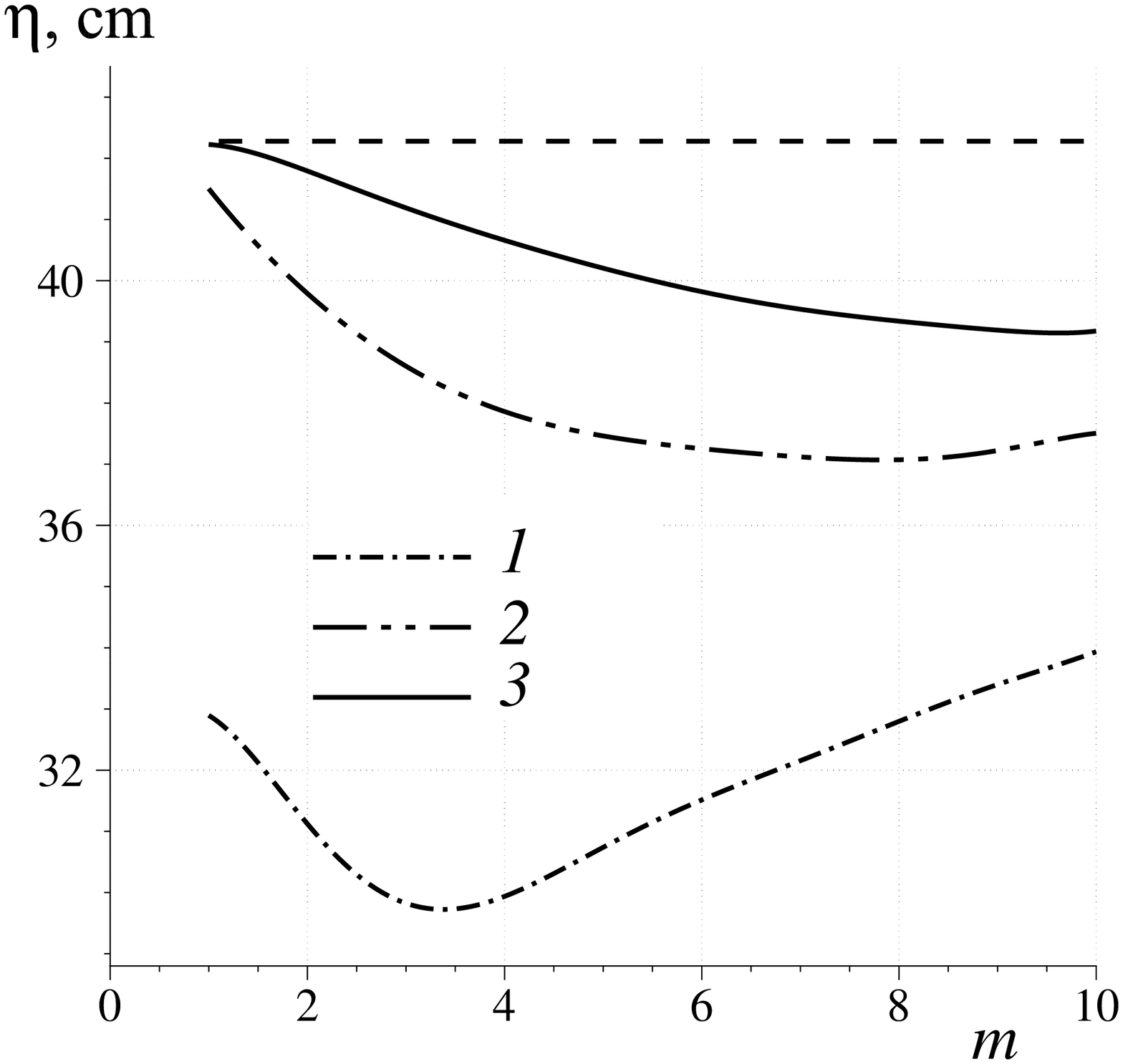}}
  \subfigure[]{\includegraphics[width=0.48\textwidth]{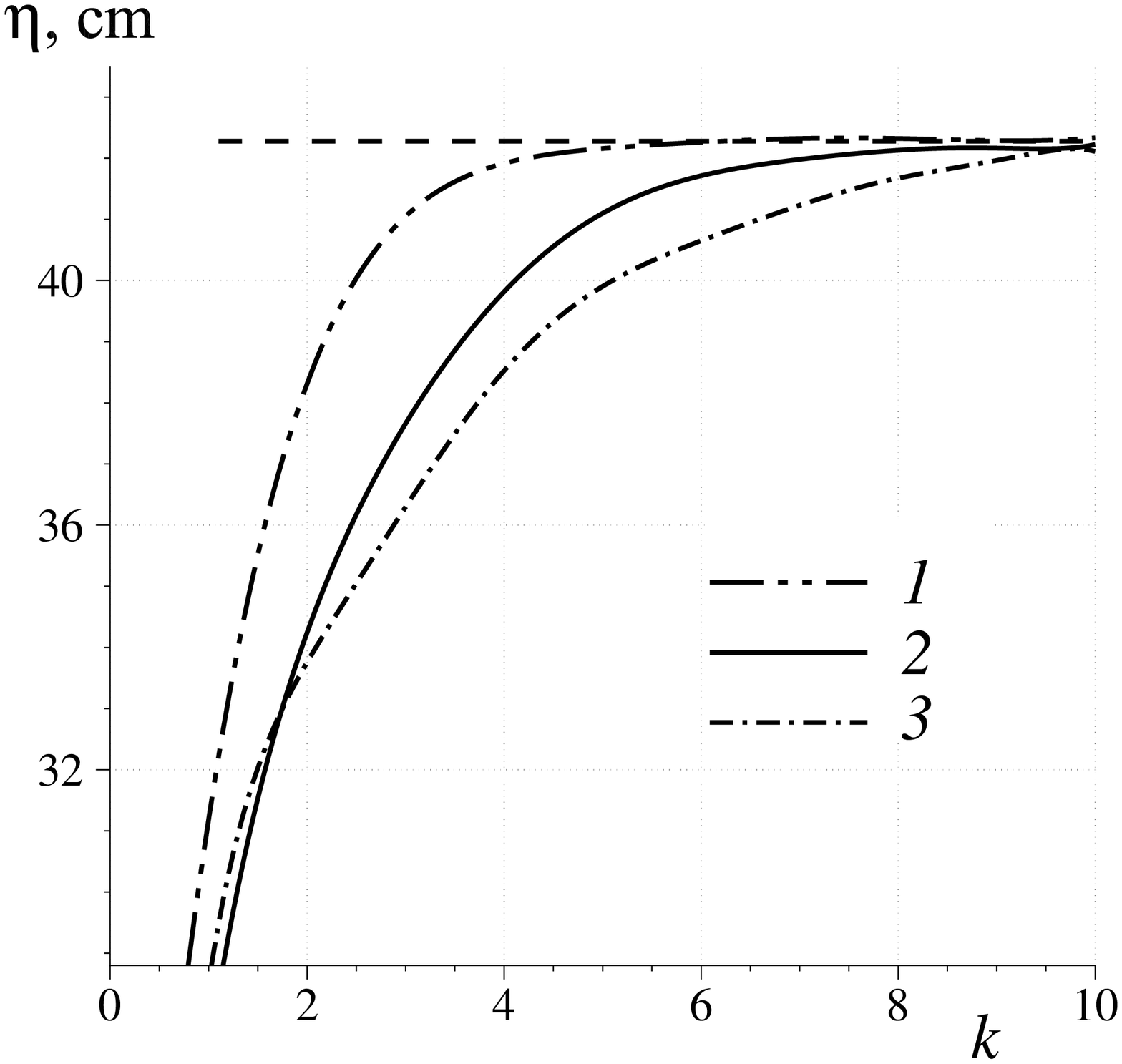}}
  \caption{\small\em Maximal run-up on the vertical wall (during the whole simulation time) as a function of the moving wall mass (a) for fixed springs rigidity $k\ =\ 1$ (1), $k\ =\ 3$ (2) and $k\ =\ 5$ (3). On the right panel we show the same quantity as a function of the springs rigidity (b) for fixed wall masses $m\ =\ 1$ (1), $m\ =\ 3$ (2) and $m\ =\ 5$ (3). This computation corresponds to the incident solitary wave. The dashed line on both panels shows the maximal run-up value on a fixed wall.}
  \label{fig:13}
\end{figure}

In Figure~\ref{fig:14}(\textit{a}) we show the wave run-up record under the action of a solitary wave on the moving wall as a function of time. This computation was performed for the wall mass $m\ =\ 5$ and springs rigidity is $k\ =\ 1\,$. In Figure~\ref{fig:14}(\textit{b}) we report the wall displacement during the wave/wall interaction process. In Figure~\ref{fig:14}(\textit{a}) we depict also the wave run-up record for a fixed wall as well. It can be clearly seen that the presence of a moving wall reduces significantly wave run-up height on it. We can also notice that for a movable wall the wave run-up happens twice in contrast to the fixed wall case. Indeed, the first (and usually maximal) run-up takes place when the wall is retracting under the incident wave action. Then, at some moment the springs accumulate enough elastic energy in order to push back the wall towards its equilibrium position. Right after this turning point the second run-up takes place. However, its value is usually lower than the first one.

\begin{figure}
  \centering
  \subfigure[]{\includegraphics[width=0.48\textwidth]{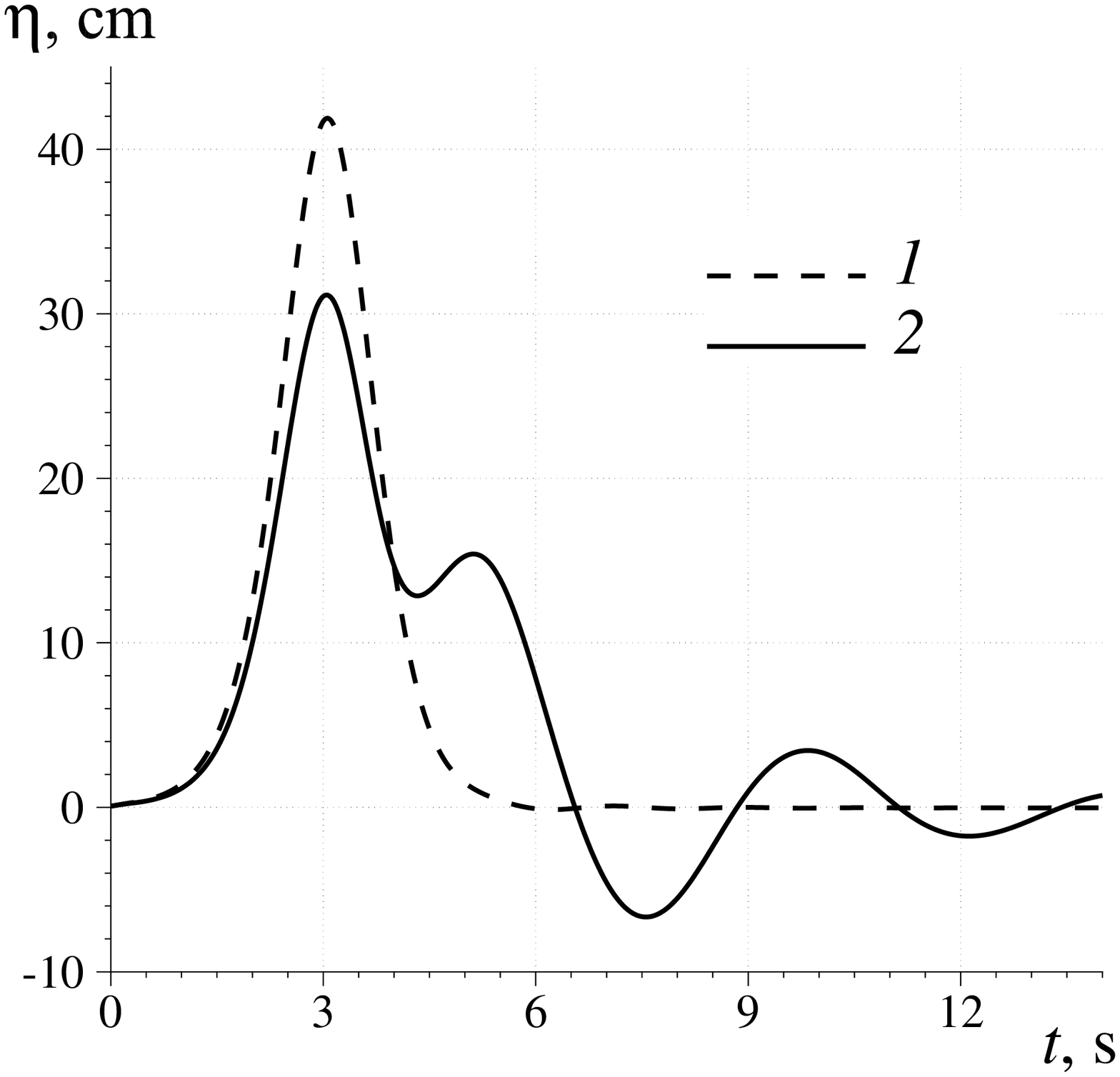}}
  \subfigure[]{\includegraphics[width=0.48\textwidth]{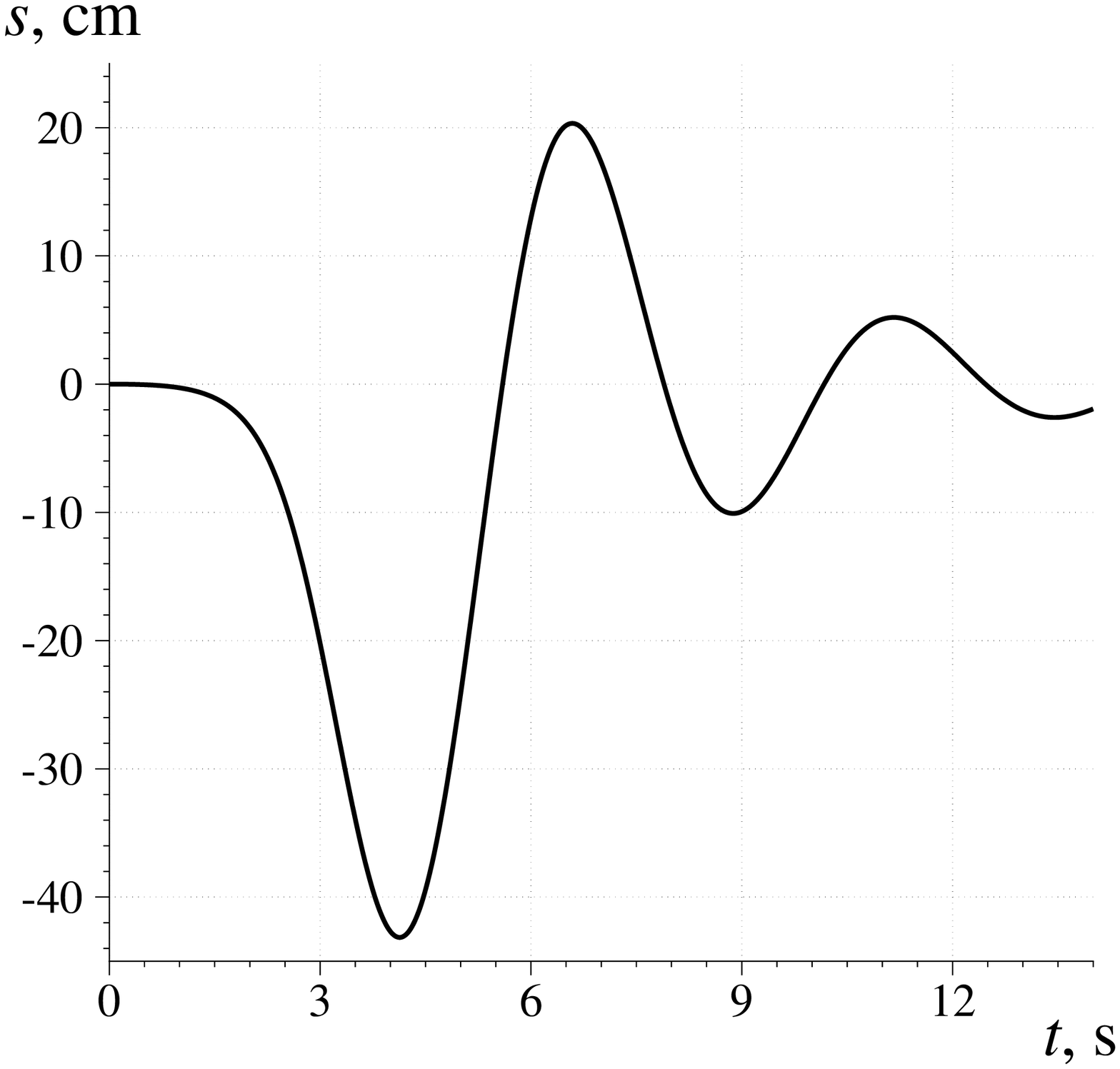}}
  \caption{\small\em Solitary wave run-up on the wall (a): free surface excursion on the fixed (1) and moving (2) walls. Wall position as a function of time (b). The wall mass $m\ =\ 5$ and springs rigidity is $k\ =\ 1\,$.}
  \label{fig:14}
\end{figure}

In Figure~\ref{fig:15} we show also the free surface evolution in space and in time. The left panel \ref{fig:15}(\textit{a}) shows the classical fixed wall case for the sake of comparison. In this case the wave field essentially consists of one incident and one reflected waves. Eventual inelastic collision effects are negligible for a solitary wave of amplitude $\alphau\ =\ 20\ \cm\,$. A much more interesting wave field is generated by the interaction with a moving wall. It is depicted in Figure~\ref{fig:15}(\textit{b}). The wall motion generates multiple reflected wave, which travel with lower speed (in agreement with their lower amplitude).

\begin{figure}
  \centering
  \subfigure[]{\includegraphics[width=0.48\textwidth]{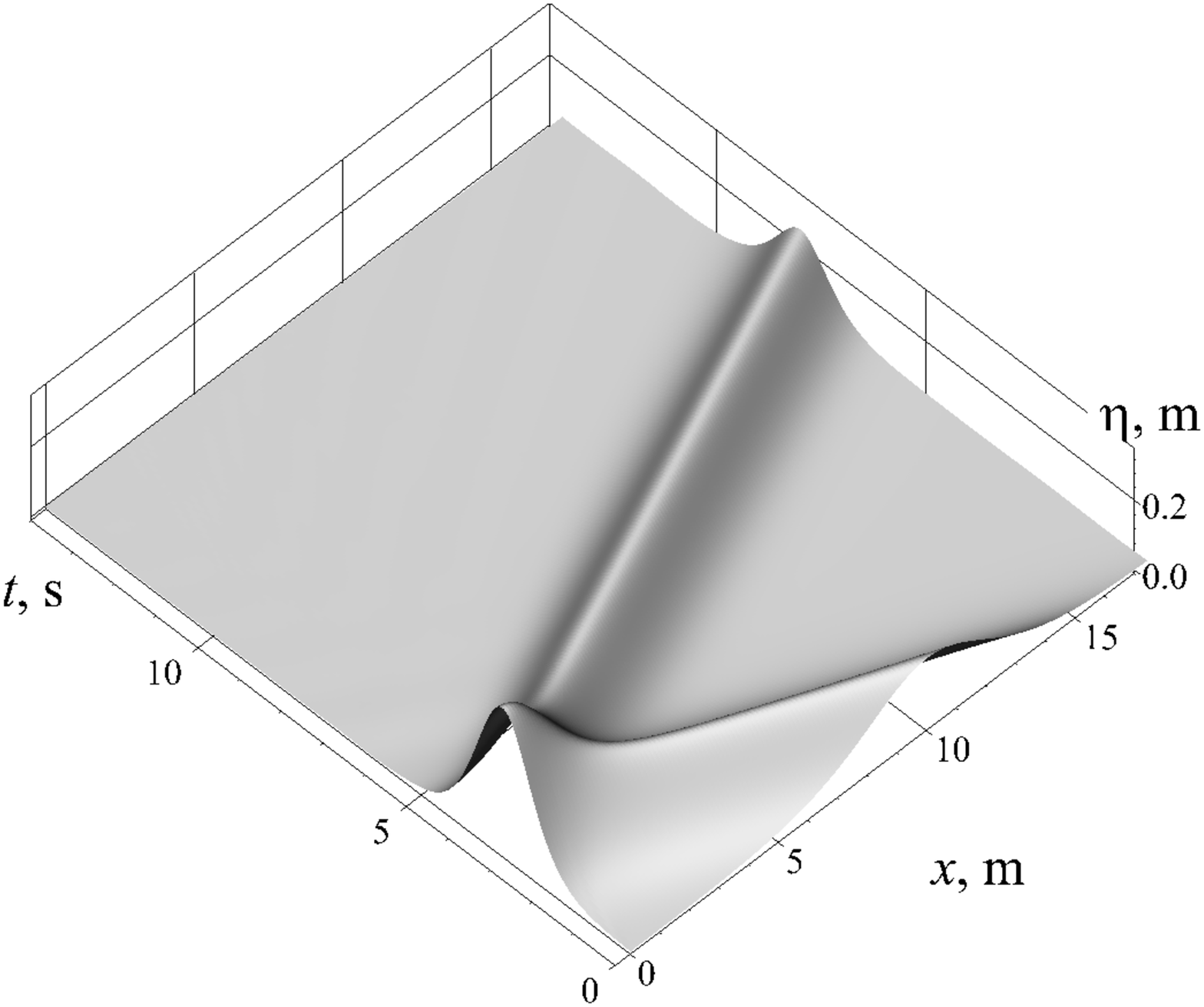}}
  \subfigure[]{\includegraphics[width=0.48\textwidth]{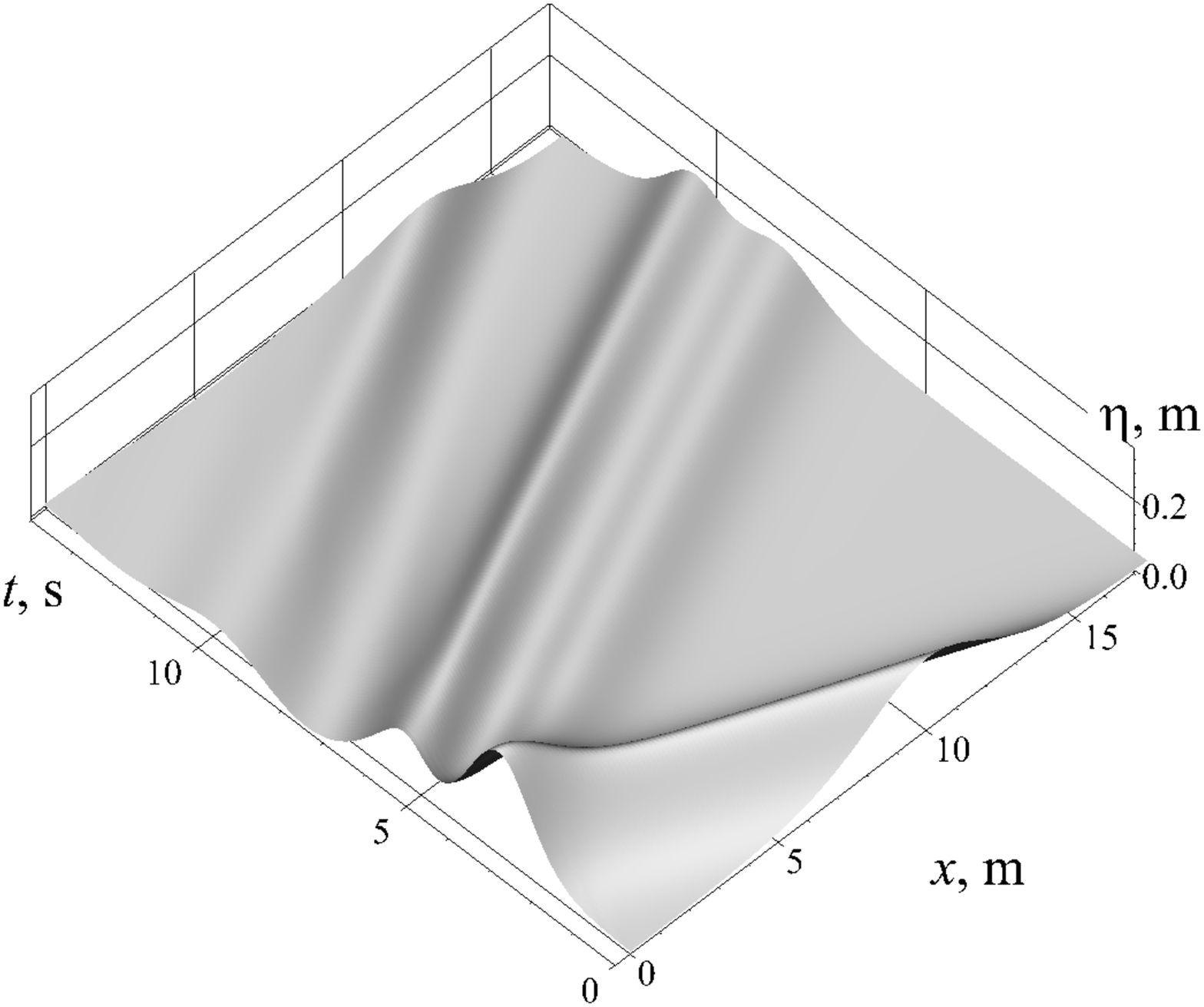}}
  \caption{\small\em Free surface dynamics represented as a surface in space-time $y\ =\ \eta\,(x,\,t)$ during a solitary wave run-up on a fixed (a) and moving (b) vertical wall. The solitary wave amplitude $\alpha\ =\ 20\ \cm$ and $x_{\,0}\ =\ 10\ \m\,$. The axis $O\,x$ shows only a part of the numerical wave tank. The whole length is $\ell\ =\ 40\ \m\,$.}
  \label{fig:15}
\end{figure}


\section{Discussion}
\label{sec:concl}

The present study was devoted to the simulation of free surface waves in a two-dimensional tank of variable size. Namely, a vertical wall is attached to a system of springs and it can move under the action of incident water waves. Modelling and numerical simulation of this problem was discussed in this study. Below we outline the main conclusions and perspectives of the present study.


\subsection{Conclusions}

In the present study we considered the problem of surface water wave modelling in domains with variable geometry. In general, problems on time-varying are known to be notoriously difficult \cite{Knobloch2015} (both theoretically and numerically). Here we considered a special case of a numerical wave tank with a moving wall. This wall can move according to the prescribed law. In this case we model the wave generation process by a piston-type wave maker \cite{Guizien2002}. In our formulation the wall can also freely move under the wave action. It is realized by connecting the wall to a system of elastic springs. In this way, the wall mass $m$ and springs rigidity $k$ are taken into account in our model by writing an additional nonlinear second order differential equation for the wall position. Surface waves are described mathematically using the full water wave problem, \ie the irrotational incompressible perfect fluid flow described by \textsc{Euler} equations \cite{Stoker1957}. No approximation was made and, thus, the presented results are fully nonlinear and fully dispersive. Moreover, we proposed a numerical algorithm on adaptive moving grids. The adaptation was achieved using the so-called equidistribution principle \cite{Khakimzyanov2015b, Khakimzyanov2015a}. The discrete problem was rigorously studied and the proposed discrete operator was shown to be self-adjoint on adaptive grids as well as the continuous problem. Thus, the numerical method preserves qualitatively some properties of the continuous operator, which is referred to as the `structure preserving' property \cite{Hairer2002} in modern numerical analysis.

The proposed algorithm was first validated on the well-known case of a fixed rigid vertical wall. There is enough analytical, numerical \cite{Fenton1982, Cooker1997} and experimental \cite{Davletshin1984} data available to validate the solver. Then, this numerical tool was applied to generate periodic waves \emph{in silico} by a moving piston-type wave maker. Finally, the wall motion and incident wave run-up on this moving wall was studied using our numerical means. Our results compared quite well with some published data \cite{He2009}. The influence and effect of various parameters on the wave run-up was investigated.


\subsection{Perspectives}

In the present article the geometry, mathematical formulation and the numerical method have been presented for the two-dimensional configuration. In future works we plan to extend this methodology for 3D flows. A priori, this generalization is going to be tedious, but straightforward.

As another possible direction for future research, we would like to underline that the moving (solid non-deformable) part in our study was a wall. This object is very important in coastal engineering. However, in future works we would like to incorporate moving objects of more general geometrical shape. As a minor point of improvement, one can notice that in all simulations presented earlier the bottom was taken to be flat. It was done on purpose to isolate the effect of the moving wall on incident waves. However, it is not a limitation of the method. The formulation presented earlier works for general bottoms and the simultaneous effect of the moving boundary and local bathymetric features on the wave run-up on this wall has to be investigated.

Finally, the formulation we used was fully conservative. In other words, no dissipation mechanism was included in our study except for the motion of the wall under the action of water waves (however, the whole wave/wall system was conservative). Of course, this model is an idealization of the reality. In real flows some dissipative effects are also present due to the molecular viscosity, friction, turbulence or other mechanisms. For this purpose a suitable visco-potential formulation has been proposed \cite{Dutykh2008a, Dutykh2007d, Dutykh2007}, which occupies the intermediate level between the potential flows and \textsc{Navier}--\textsc{Stokes} equations.


\subsection{Side effects}

Our work has also at least one nice side effect. Namely, if we simplify the problem considered in the present study by fixing the wall position, \eg $s(t)\ \equiv\ 0\,$, we obtain a robust numerical wave tank for 2D surface waves similar to \cite{Grilli1989}. However, the technique presented here is fully \emph{adaptive}. By the way, we are not aware of any recent study using similar adaptive redistribution methods for the full \textsc{Euler} equations with \emph{free surface}. So, this aspect might be new.


\subsection*{Acknowledgments}
\addcontentsline{toc}{subsection}{Acknowledgments}

This research was supported by RSCF project No 14-17-00219. D.~\textsc{Dutykh} acknowledges the support of the CNRS under the PEPS InPhyNiTi project FARA and project \No EDC26179 --- ``\textit{Wave interaction with an obstacle}'' as well as the hospitality of the Institute of Computational Technologies SB RAS during his visit in October 2015. D.~\textsc{Dutykh} would like to acknowledge the help of his colleague Tom \textsc{Hirschowitz} who timely provided tea leaves in cases of emergency.


\appendix
\section{Transformation of coordinates}
\label{app:comp}

In this Appendix we explain some calculation details of coordinate transformations which allow to pass from the system \eqref{eq:laplace} -- \eqref{eq:pressure} posed on a variable domain $\Omega\,(t)$ to the system \eqref{eq:laplaceQ} -- \eqref{eq:pressureQ} on the fixed domain $\Q^{\,0}\,$. Consider a smooth bijective map \eqref{eq:map} (see also Figure~\ref{fig:map} for an illustration):
\begin{equation}\label{eq:trans}
  \left\{
  \begin{array}{rl}
    x\ &=\ x\,(q^{\,1},\, q^{\,2},\, t)\,, \\
    y\ &=\ y\,(q^{\,1},\, q^{\,2},\, t)\,.
  \end{array}
  \right.
  \qquad \Longleftrightarrow \qquad
  \left\{
  \begin{array}{rl}
    q^{\,1}\ &=\ q^{\,1}\,(x,\, y,\, t)\,, \\
    q^{\,2}\ &=\ q^{\,2}\,(x,\, y,\, t)\,.
  \end{array}
  \right.
\end{equation}
We consider that the time variable is the same in both coordinate systems. Two maps written above are mutually inverse. So, let us write this condition
\begin{align}
  x\ &\equiv\ x\,\bigl[q^{\,1}\,(x,y,t),\, q^{\,2}\,(x,y,t),\, t\bigr]\,, \label{eq:1} \\
  y\ &\equiv\ y\,\bigl[q^{\,1}\,(x,y,t),\, q^{\,2}\,(x,y,t),\, t\bigr]\,. \label{eq:2}
\end{align}
Then, we differentiate \eqref{eq:1} and \eqref{eq:2} with respect to $x\,$. We thus obtain two relations
\begin{equation*}
  \left\{\begin{array}{rl}
    x_{\,q^{\,1}}\,\pdd{q^{\,1}}{x}\ +\ x_{\,q^{\,2}}\,\pdd{q^{\,2}}{x}\ &=\ 1\,, \\ [0.75em]
    y_{\,q^{\,1}}\,\pdd{q^{\,1}}{x}\ +\ y_{\,q^{\,2}}\,\pdd{q^{\,2}}{x}\ &=\ 0\,.
  \end{array}\right.
\end{equation*}
The last two relations can be considered as a system of two linear equations with respect to $\pdd{q^{\,1}}{x}$ and $\pdd{q^{\,2}}{x}$ as unknowns. By trivially applying the \textsc{Cramer} rule we obtain
\begin{equation}\label{eq:rel1}
  \pd{q^{\,1}}{x}\ =\ \frac{y_{q^{\,2}}}{\J}\,, \qquad
  \pd{q^{\,2}}{x}\ =\ -\frac{y_{q^{\,1}}}{\J}\,,
\end{equation}
where $\J\ \eqdef\ x_{q^{\,1}}\,y_{q^{\,2}}\ -\ x_{q^{\,2}}\,y_{q^{\,1}}$ is the \textsc{Jacobian} defined in \eqref{eq:jac}. The \textsc{Jacobian} can be also seen as the determinant of a matrix:
\begin{equation*}
  \J\ =\ \det\begin{pmatrix}
    1 & 0 & 0 \\
    x_{\,t} & x_{\,q^{\,1}} & x_{\,q^{\,2}} \\
    y_{\,t} & y_{\,q^{\,1}} & y_{\,q^{\,2}}
  \end{pmatrix}\ \equiv\ \det\begin{pmatrix}
    x_{\,q^{\,1}} & x_{\,q^{\,2}} \\
    y_{\,q^{\,1}} & y_{\,q^{\,2}}
  \end{pmatrix}
  \ \equiv\ x_{q^{\,1}}\,y_{q^{\,2}}\ -\ x_{q^{\,2}}\,y_{q^{\,1}}\,.
\end{equation*}

Similarly, by differentiating \eqref{eq:1} and \eqref{eq:2} with respect to $y$ and using \textsc{Cramer}'s rule, one can show that
\begin{equation}\label{eq:rel2}
  \pd{q^{\,1}}{y}\ =\ -\frac{x_{q^{\,2}}}{\J}\,, \qquad
  \pd{q^{\,2}}{y}\ =\ \frac{x_{q^{\,1}}}{\J}\,.
\end{equation}
Now, by differentiating \eqref{eq:1} and \eqref{eq:2} with respect to $t$ we obtain the following two equations
\begin{equation*}
  \left\{
  \begin{array}{rl}
    x_{\,q^{\,1}}\,\pdd{q^{\,1}}{t}\ +\ x_{\,q^{\,2}}\,\pdd{q^{\,2}}{t}\ &=\ -x_{\,t}\,, \\ [0.75em]
    y_{\,q^{\,1}}\,\pdd{q^{\,1}}{t}\ +\ y_{q^{\,2}}\,\pdd{q^{\,2}}{t}\ &=\ -\,y_{\,t}\,.
  \end{array}
  \right.
\end{equation*}
Solving the last linear system with respect to $\pdd{q^{\,1}}{t}$ and $\pdd{q^{\,2}}{t}$ yields
\begin{equation}\label{eq:qt}
  \pd{q^{\,1}}{t}\ =\ \frac{y_{\,t}\cdot x_{\,q^{\,2}}\ -\ x_{\,t}\cdot y_{\,q^{\,2}}}{\J}\,, \qquad
  \pd{q^{\,2}}{t}\ =\ \frac{x_{\,t}\cdot y_{\,q^{\,1}}\ -\ y_{\,t}\cdot x_{\,q^{\,1}}}{\J}\,.
\end{equation}
This concludes the computation of partial derivatives of the inverse mapping $\q\ =\ \q\,(\x,\,t)\,$.


\section{Transformation of the Laplace operator}
\label{app:lapl}

Now we can rewrite the \textsc{Laplace} operator \eqref{eq:laplace} in curvilinear coordinates $\bigl(\q,\,t\bigr)\,$. The idea consists in viewing the velocity potential as
\begin{equation*}
  \phi\,(\x,\,t)\ \equiv\ \phi\,\bigl(\q\,(\x,\,t),\, t\bigr)\,.
\end{equation*}
The first derivatives can be easily computed:
\begin{equation*}
  \pd{\phi}{x}\ =\ \pd{\phi}{q^{\,1}}\cdot\pd{q^{\,1}}{x}\ +\ \pd{\phi}{q^{\,2}}\cdot\pd{q^{\,2}}{x}\,,
\end{equation*}
and using the just derived relation \eqref{eq:rel1} we can express everything in terms of derivatives solely with respect to $q_1$, $q_2$:
\begin{equation*}
  \pd{\phi}{x}\,(\q,\,t)\ =\ \frac{1}{\J}\;\Bigl[\,y_{\,q^{\,2}}\cdot\pd{\phi}{q^{\,1}}\ -\ y_{\,q^{\,1}}\cdot\pd{\phi}{q^{\,2}}\,\Bigr]\,.
\end{equation*}
Similarly, we can compute the other partial derivative of the velocity potential $\phi$ with respect to $y$ in new coordinates:
\begin{equation*}
  \pd{\phi}{y}\,(\q,\,t)\ =\ \frac{1}{\J}\;\Bigl[\,-\,x_{\,q^{\,2}}\cdot\pd{\phi}{q_1}\ +\ x_{\,q^{\,1}}\cdot\pd{\phi}{q^{\,2}}\,\Bigr]\,.
\end{equation*}
Now we can apply recursively the just obtained results for $\pd{\phi}{x}\,(\q,\,t)$ and $\pd{\phi}{y}\,(\q,\,t)$ to compute the second derivatives as well:
\begin{multline*}
  \pd{^{\,2}\,\phi}{x^{\,2}}\ =\ \pd{}{x}\Bigl(\pd{\phi}{x}\Bigr)\ =\ \pd{}{q^{\,1}}\;\Bigl(\pd{\phi}{x}\Bigr)\cdot\pd{q^{\,1}}{x}\ +\ \pd{}{q^{\,2}}\;\Bigl(\pd{\phi}{x}\Bigr)\cdot\pd{q^{\,2}}{x}\ = \\
  \pd{}{q^{\,1}}\;\Bigl(\pd{\phi}{q^{\,1}}\cdot\pd{q^{\,1}}{x}\ +\ \pd{\phi}{q^{\,2}}\cdot\pd{q^{\,2}}{x}\Bigr)\cdot\pd{q^{\,1}}{x}\ +\ \pd{}{q^{\,2}}\;\Bigl(\pd{\phi}{q^{\,1}}\cdot\pd{q^{\,1}}{x}\ +\ \pd{\phi}{q^{\,2}}\cdot\pd{q^{\,2}}{x}\Bigr)\cdot\pd{q^{\,2}}{x}\,.
\end{multline*}
By taking into account the relations derived earlier, we obtain that
\begin{equation*}
  \pd{^{\,2}\,\phi}{x^{\,2}}\,(\q,\,t)\ =\ \frac{1}{\J}\;\biggl\{\,\pd{}{q^{\,1}}\;\Bigl[\,\frac{\phi_{\,q^{\,1}}\,y_{\,q^{\,2}}^{\,2}\ -\ \phi_{\,q^{\,2}}\,y_{\,q^{\,1}}\,y_{\,q^{\,2}}}{\J}\,\Bigr]\ +\ \pd{}{q^{\,2}}\;\Bigl[\,\frac{-\,\phi_{\,q^{\,1}}\,y_{\,q^{\,1}}\,y_{\,q^{\,2}}\ +\ \phi_{\,q^{\,2}}\,y_{\,q^{\,1}}^{\,2}}{\J}\,\Bigr]\,\biggr\}\,.
\end{equation*}
Using similar methods we can obtain the following expression for $\pd{^{\,2}\,\phi}{y^{\,2}}\,(\q,\,t)\,$:
\begin{equation*}
  \pd{^{\,2}\,\phi}{y^{\,2}}\,(\q,\,t)\ =\ \frac{1}{\J}\;\biggl\{\,\pd{}{q^{\,1}}\;\Bigl[\,\frac{\phi_{\,q^{\,1}}\,x_{\,q^{\,2}}^{\,2}\ -\ \phi_{\,q^{\,2}}\,x_{\,q^{\,1}}\,x_{\,q^{\,2}}}{\J}\,\Bigr]\ +\ \pd{}{q^{\,2}}\;\Bigl[\,\frac{-\,\phi_{\,q^{\,1}}\,x_{\,q^{\,1}}\,x_{\,q^{\,2}}\ +\ \phi_{\,q^{\,2}}\,x_{\,q^{\,1}}^{\,2}}{\J}\,\Bigr]\,\biggr\}\,.
\end{equation*}
The \textsc{Laplace} operator can be particularly compactly written in new variables if we introduce the metric tensor components:
\begin{equation*}
  g_{\,1\,1}\ \eqdef\ x_{\,q^{\,1}}^{\,2}\ +\ y_{\,q^{\,1}}^{\,2}\,, \qquad
  g_{\,1\,2}\ \equiv\ g_{\,2\,1}\ \eqdef\ x_{\,q^{\,1}}\cdot x_{\,q^{\,2}}\ +\ y_{\,q^{\,1}}\cdot y_{\,q^{\,2}}\,, \qquad
  g_{\,2\,2}\ \eqdef\ x_{\,q^{\,2}}^{\,2}\ +\ y_{\,q^{\,2}}^{\,2}\,,
\end{equation*}
along with the \emph{implied} anisotropic diffusion coefficients:
\begin{equation*}
  \K_{\,1\,1}\ \eqdef\ \frac{g_{\,2\,2}}{\J}\,, \qquad
  \K_{\,1\,2}\ \equiv\ \K_{\,2\,1} \eqdef\ -\,\frac{g_{\,1\,2}}{\J}\,, \qquad
  \K_{\,2\,2}\ \eqdef\ \frac{g_{\,1\,1}}{\J}\,.
\end{equation*}
Using this notation we can finally write the \textsc{Laplace} operator in transformed coordinates:
\begin{equation*}
  \pd{^{\,2}\,\phi}{x^{\,2}}\ +\ \pd{^{\,2}\,\phi}{y^{\,2}}\ \equiv\ \frac{1}{\J}\;\biggl[\,\pd{}{q^{\,1}}\;\Bigl(\K_{\,1\,1}\,\phi_{\,q^{\,1}}\ +\ \K_{\,1\,2}\,\phi_{\,q^{\,2}}\Bigr)\ +\ \pd{}{q^{\,2}}\;\Bigl(\K_{\,2\,1}\,\phi_{\,q^{\,1}}\ +\ \K_{\,2\,2}\,\phi_{\,q^{\,2}}\Bigr)\,\biggr]\,.
\end{equation*}
As a result, the \textsc{Laplace} equation reads:
\begin{equation*}
  \pd{}{q^{\,\alpha}}\;\biggl[\,\K_{\,\alpha\,\beta}\,\pd{\phi}{q^{\,\beta}}\,\biggr]\ =\ 0\,, \qquad
  \alpha,\,\beta\ =\ 1,\,2\,,
\end{equation*}
where we employed an implicit summation over repeated indices\footnote{In the literature this rule is referred to as the \textsc{Einstein} summation convention.}. This result coincides exactly with equation \eqref{eq:laplaceQ} presented earlier in the manuscript.


\section{Boundary conditions transformation}
\label{app:bc}

Finally, let us derive some relations which turn out to be useful in the transformation of boundary conditions. In particular, in this Appendix we focus on the kinematic condition \eqref{eq:kinem}. First of all, during the derivation of boundary conditions one needs time derivatives of the inverse mapping in \eqref{eq:trans}:
\begin{equation*}
  \pd{q^{\,1}}{t}\ =\ \frac{1}{\J}\;\biggl[\,x_{\,q^{\,2}}\;\pd{y}{t}\ -\ y_{\,q^{\,2}}\;\pd{x}{t}\,\biggr]\,, \qquad
  \pd{q^{\,2}}{t}\ =\ \frac{1}{\J}\;\biggl[\,y_{\,q^{\,1}}\;\pd{x}{t}\ -\ x_{\,q^{\,1}}\;\pd{y}{t}\,\biggr]\,.
\end{equation*}
These formulas were given also earlier in \eqref{eq:qt} with more details on their derivation. Let us compute \textsc{Cartesian} components of the velocity vector:
\begin{align*}
  u\,(\q,\,t)\ &=\ \pd{\phi}{x}\,(\q,\,t)\ =\ \phi_{\,q^{\,1}}\;\pd{q^{\,1}}{x}\ +\ \phi_{\,q^{\,2}}\;\pd{q^{\,2}}{x}\ \stackrel{\eqref{eq:rel1}}{=}\ \frac{\phi_{\,q^{\,1}}\,y_{\,q^{\,2}}\ -\ \phi_{\,q^{\,2}}\,y_{\,q^{\,1}}}{\J}\,, \\
  v\,(\q,\,t)\ &=\ \pd{\phi}{y}\,(\q,\,t)\ =\ \phi_{\,q^{\,1}}\;\pd{q^{\,1}}{y}\ +\ \phi_{\,q^{\,2}}\;\pd{q^{\,2}}{y}\ \stackrel{\eqref{eq:rel2}}{=}\ \frac{-\,\phi_{\,q^{\,1}}\,x_{\,q^{\,2}}\ +\ \phi_{\,q^{\,2}}\,x_{\,q^{\,1}}}{\J}\,.
\end{align*}
The free surface elevation $\eta$ on the transformed domain becomes a function of the variable $q^{\,1}$ and time $t\,$:
\begin{equation*}
  \eta\ =\ \tilde{\eta}\,(q^{\,1},\,t)\ \equiv\ \eta\,\bigl(x\,(q^{\,1},\,q^{\,2}\,\equiv\,1,\,t),\,t\bigr)\,.
\end{equation*}
Now, we can compute partial derivatives of the function $\eta\,$, which appear in the boundary conditions:
\begin{align*}
  \eta_{\,x}\ &=\ \tilde{\eta}_{\,q^{\,1}}\;\pd{q^{\,1}}{x}\ +\ \underbrace{\eta_{\,q^{\,2}}}_{\displaystyle{\equiv\ 0}}\;\pd{q^{\,2}}{x}\,, \\
  \eta_{\,t}\ &=\ \tilde{\eta}_{\,t}\ +\ \tilde{\eta}_{\,q^{\,1}}\;\pd{q^{\,1}}{t}\,.
\end{align*}
By substituting all these elements into equation \eqref{eq:kinem}, we obtain the following condition:
\begin{equation}\label{eq:last}
  \tilde{\eta}_{\,t}\ +\ v^{\,1}\,\tilde{\eta}_{\,q^{\,1}}\ -\ v\ =\ 0\,,
\end{equation}
where $v^{\,1}$ is a contravariant component of the velocity, which can be computed according to this formula:
\begin{equation*}
  v^{\,1}\ =\ \pd{q^{\,1}}{t}\ +\ u\;\pd{q^{\,1}}{x}\ +\ v\;\pd{q^{\,1}}{y}\,.
\end{equation*}
The last equation \eqref{eq:last} completes the derivation of the free surface kinematic boundary conditions. Other boundary conditions can be derived in a similar manner.

\begin{remark}
Above, in the main text of our manuscript we do not use the tilde notation $\tilde{\eta}$ in order to avoid overloading of the text with multiple notations. It is assumed that the reader can make the difference between $\eta\,(x,\,t)$ and $\tilde{\eta}\,(q^{\,1},\,t)$ depending on the context.
\end{remark}


\section{Piston motion modelling}
\label{app:pist}

Let us introduce a \textsc{Cartesian} coordinate system with the horizontal axis $O\,\chi$ looking rightwards such that the left moving wall is located at the point $\chi\ =\ 0$ in the absence of any kind of loading (only the hydrostatic \emph{atmospheric} pressure that we neglect here). In this case the wave tank is dry. Now we fill it with still water of the uniform depth $h_{\,0}\ >\ 0\,$. The hydrostatic force $\F\,(0)\ >\ 0$ (defined in equations \eqref{eq:f0}) appears on the left wall and the springs accumulate some elastic deformation. It results in the wall displacement of magnitude $\Uptheta\,(0)\ >\ 0$ in the negative direction of the axis $O\,\chi\,$. The displacement $\Uptheta\,(0)$ can be found by applying the \textsc{Hook} law \cite{Landau1986}:
\begin{equation*}
  k\,\Uptheta\,(0)\ =\ \F\,(0)\,,
\end{equation*}
where $k$ is the springs rigidity. For $t\ >\ 0$ we have another force $\F\,(t)$ acting on the left wall due to the wave motion. This force acts against the wall. Thus, in the second law of \textsc{Newton} we take it with the opposite sign \cite{Newton1687, Landau1976}:
\begin{equation}\label{eq:chi}
  m\,\ddot{\chi}\ +\ k\,\chi\ =\ -\,\F\,(t)\,,
\end{equation}
where $m$ is the wall mass. The initial conditions for this ODE are:
\begin{equation*}
  \chi\,(0)\ =\ -\,\Uptheta\,(0)\,, \qquad \dot{\chi}\,(0)\ =\ 0\,.
\end{equation*}
Now, we make a change of coordinates $0\,\chi\ \rightsquigarrow\ O\,x$ so that the wall is located initially at the point $x\ =\ 0\,$, \ie
\begin{equation*}
  x\ =\ \chi\ +\ \Uptheta\,(0)\,.
\end{equation*}
Then, the initial conditions become by construction:
\begin{equation*}
  x\,(0)\ =\ 0\,, \qquad \dot{x}\,(0)\ =\ 0\,.
\end{equation*}
Equation of motion \eqref{eq:chi} transforms to:
\begin{equation*}
  m\,\ddot{x}\ +\ k\,x\ -\ \underbrace{k\,\Uptheta\,(0)}_{\displaystyle{\equiv\ \F\,(0)}}\ =\ -\,\F\,(t)\,.
\end{equation*}
The last equation can be rewritten as
\begin{equation*}
  m\,\ddot{x}\ +\ k\,x\ =\ -\,\bigl[\,\F\,(t)\ -\ \F\,(0)\,\bigr]\,.
\end{equation*}
Finally, by denoting the wall displacement by $s\,(t)\ \eqdef\ x\,(t)$ we recover equation \eqref{eq:piston}.


\bigskip\bigskip
\addcontentsline{toc}{section}{References}
\bibliographystyle{abbrv}

\bigskip\bigskip

\end{document}